\documentclass[11pt]{article}
\usepackage{amsmath, amssymb, amsfonts, amsthm}
\usepackage{psfrag}
\newtheorem{theorem}{Theorem}[section]
\newtheorem{prop}[theorem]{Proposition}
\newtheorem{cor}[theorem]{Corollary}
\newtheorem{lemma}[theorem]{Lemma}
\newtheorem{defi}[theorem]{Definition}
\newtheorem{rem}[theorem]{Remark}
\newtheorem{prob}[theorem]{Problem}%

\psfrag{t1}{\fontsize{20}{20}$t_1$}
\psfrag{t2}{\fontsize{20}{20}$t_2$}
\psfrag{O}{\fontsize{20}{20}$0$}
\psfrag{1}{\fontsize{20}{20}$1$}
\psfrag{A}{\fontsize{20}{20}$\om > \om^\star$}
\psfrag{B}{\fontsize{20}{20}$\om = \om^\star$}
\psfrag{C}{\fontsize{20}{20}$\om < \om^\star$}
\psfrag{D}{\fontsize{20}{20}$\om_1$}
\psfrag{E}{\fontsize{20}{20}$\om_2$}
\psfrag{u}{\fontsize{20}{20}$T_1$}
\psfrag{v}{\fontsize{20}{20}$T_2$}

\def\z*{\bar z}

\def\RE{\mathbb R}

\def\p{\par\noindent}

\newcommand{\be}{\begin{equation}}
\newcommand{\ee}{\end{equation}}
\newcommand{\bey}{\begin{eqnarray}}
\newcommand{\eey}{\end{eqnarray}}

\newcommand{\bes}{{\begin{split}}}
\newcommand{\ees}{{\end{split}}}
\newcommand{\bees}{{\begin{equation}\begin{split}}}
\newcommand{\es}{{\end{split}\end{equation}}}
\newcommand{\erre}{{\mathbb R}}

\newcommand{\sech} {{\rm{sech}}}

\newcommand{\tah}{{\rm{tanh}}}
\newcommand{\comple}{{\mathbb C}}

\newcommand{\bea}{\begin{eqnarray}}

\newcommand{\eea}{\end{eqnarray}}

\newcommand{\n}{\noindent}

\newcommand{\f}{\frac}

\newcommand{\ve}{\varepsilon}

\newcommand{\om}{\omega}

\newcommand{\nnn} {\nonumber}

\newcommand{\mc}{\mathcal}

\numberwithin{equation}{section}

\newcommand{\wt}{\widetilde}

\newcommand{\ov}{\overline}

\newcommand{\donothing}[1]{}

\begin{document}
\title{Stability and symmetry-breaking bifurcation for the ground
  states  of a NLS with a $\delta^\prime$ interaction} 

\author{Riccardo Adami and Diego Noja
\\
\\
Dipartimento di Matematica e Applicazioni, Universit\`a
 di Milano
Bicocca \\
via R. Cozzi, 53, 20125 Milano, Italy \\ \\ and \\ \\
Istituto di Matematica Applicata e Tecnologie Informatiche \\
Consiglio Nazionale delle Ricerche \\
via Ferrata, 1, 27100, Pavia, Italy
}


\maketitle

\begin{abstract}
\par\noindent
We determine and study the ground states of a focusing
Schr\"odinger equation in dimension one with a power nonlinearity
$|\psi|^{2\mu} \psi$ and a strong inhomogeneity represented 
by a singular point perturbation, the so-called (attractive) $\delta^\prime$ 
interaction, located at the origin. 

\n
The time-dependent
problem turns out to be globally well posed in the subcritical regime,
and locally well posed in the supercritical and critical regime in the
appropriate energy space.  The set of the (nonlinear) ground states 
is completely determined. 
For any value of the nonlinearity power, 
it
exhibits a symmetry breaking bifurcation structure as a function of
the frequency (i.e., the nonlinear eigenvalue) $\omega$. More precisely,
 there exists
a critical value $\om^*$ of the nonlinear eigenvalue $\om$, such that:
if $\om_0 < \om < \om^*$, then there is a single
ground state and it is an odd function; if $\om > \om^*$ then
there exist two non-symmetric ground states.

\n
We prove that before bifurcation (i.e., for $\om < \om^*$) and for any
subcritical power, every ground state
is orbitally stable. After bifurcation ($\om =\om^*+0$), ground states are 
stable if $\mu$ does not exceed a value $\mu^\star$ that lies between $2$ and $2.5$,
and become unstable for $\mu > \mu^*$. Finally, for $\mu > 2$ and $\om \gg \om^*$,
all ground states are unstable. The branch of odd ground states for 
$\om < \om^*$ can be 
continued at any $\om > \om^*$, obtaining a family of 
orbitally unstable stationary states.

\n
Existence of ground states is proved by variational techniques, 
and the stability properties of stationary states are investigated by means of the 
Grillakis-Shatah-Strauss framework, where some non standard 
techniques have to be used to establish the needed properties 
of linearization operators. 
\end{abstract}

\section{Introduction}
\setcounter{equation}{0}
The present paper is devoted to the analysis of existence and stability of the ground states
 of a nonlinear Schr\"odinger equation with a point defect in dimension one.  The
Schr\"odinger equation bears an attractive power nonlinearity,  and the defect is
described by a particular point interaction in dimension one, the
so-called attractive $\delta'$ interaction. In a formal way, the time
dependent equation to be studied is given by 
\be \begin{split}\label{nlsdelta'} \left\{
\begin{array} {ccc} i \partial_t \psi (t) &=-&\partial_{xx}\psi(t) - \gamma \delta'_0  \psi (t) - \lambda
 | \psi (t) |^{2\mu }\psi (t) \\
\psi(0)&=& \psi_0 \end{array}
\right.
\end{split}
\ee
where $\psi_0$ represents the initial data, and $\lambda>0$,
$\gamma>0$, $\mu > 0$. 
The non rigorous character of the expression  
\eqref{nlsdelta'} is due to the
fact that the combination $\frac{1}{2}\partial_{xx}\psi(t)+ \gamma
\delta'_0  \psi (t)$ is meaningless if literally interpreted as
an operator sum or as a form sum, due to the exceedingly singular
character of the perturbation given by $\delta'_0$. Nevertheless, it is
well known (\cite {[AGHH]}) that it is possible to define a singular
perturbation $H_\gamma$ of the one-dimensional laplacian
$-\frac{d^2}{dx^2}$ 
which is a
self-adjoint operator in $L^2(\erre)$ and has the expected properties of the stated
formal expression. The self-adjointness is implemented through the singular
boundary condition which defines the domain of the operator, i.e  
$$D(H_\gamma)=\{ \psi \in H^2(\erre \backslash \{ 0 \}) , 
\quad |\quad \psi'(0+)= \psi' (0-)\ , \quad
\psi (0+) - \psi (0-) = -\gamma \psi' (0+)\ \}$$  and the action is
$H_{\gamma}\psi = -\psi''$ out of the origin. 

The $\delta'$ is called {\it repulsive} when $\gamma<0$ and {\it
  attractive} when $\gamma>0$.  
Note that the generic element of the domain of $\delta'$
interactions is not continuous at the origin (but its left and right derivative 
exist and coincide), at variance with the milder and better known case of the $\delta$
interaction (which is defined by $\psi\in H^1(\RE)$ and boundary condition 
$\psi^\prime(0^+)-\psi^\prime(0^-)= \alpha \psi(0)$, where the
  real parameter $\alpha$ is interpreted as the strength of the
  interaction). 
The quadratic form computed on a domain element $\psi$ turns out to be 
$$ 
( H_\gamma \psi, \psi) = \|\psi^{\prime}\|^2_{L^2} - 
\gamma|\psi^{\prime}(0)|^2
$$ which justifies the name of $\delta'$ interaction given to
$H_{\gamma}$. Concerning the physical meaning of the $\delta'$ interaction, 
perhaps its best known use is in the analysis of Wannier-Stark effect (see \cite{[AEL]}). 
It is known that the spectrum of Wannier-Stark hamiltonians in the presence of a periodic 
array of $\delta'$ interactions shows a remarkably different behaviour with respect 
to the case of regular periodic  potential; in particular, the spectrum has no absolutely continuous 
part and it is typically pure point (by the way, the corresponding properties 
of Wannier-Stark for the $\delta$ array are not known). In the present paper, 
the core interpretation is that of a strongly singular and non trivial scatterer. 
It is known that
the $\delta'$ 
interaction cannot be obtained as
the limit of Schr\"odinger operators in which the potential 
is a derivative of a $\delta$--like regular function, as the name could 
erroneously suggest (see \cite{[CS],[ENZ]} for a thorough analysis of this problem). 
An approximation through three scaled $\delta$ potentials exists, but the scaling is 
nonlinear as the distance of the centres vanishes. Nevertheless, the $\delta'$ 
interaction has a high energy scattering behaviour that can be reproduced, 
up to a phase factor, through scaling limits of scatterers with internal structure, the 
so-called spiked-onion graphs (see \cite{[AEL]}). 
These are obtained joining two halflines by $N$ edges 
of length $L$ and letting $L\rightarrow 0$, $N \rightarrow \infty$ 
while keeping the product $NL$ fixed. An analogous behaviour is 
obtained considering a sphere with two halflines attached. These 
results enforce the interpretation of the $\delta'$ interaction 
as an effective model of a scatterer with non elementary structure. 
By the way, in this respect the results of the present paper give 
support to this view through the analysis of the bifurcation of 
nonlinear bound states.

  \par\noindent From a more abstract point of view, both $\delta$
and $\delta'$ interactions are members of a $4$-parameter family
of self-adjoint perturbations of the one dimensional laplacian, the 
so-called $1$-dimensional point interactions (see \cite
{[ABD],[AGHH],[EG]}).   As explained above, we interpret the presence of a point 
interaction in the equation \eqref{nlsdelta'} as a model of strongly singular interaction between
nonlinear waves and an inhomogeneity. When the inhomogeneity is
described by a $\delta$ interaction, a fairly extended literature
exists of both physical and numerical character; more recently, 
there has been a growing interest in this model from the mathematical
side,  in the
attempt to establish rigorous results concerning the existence of
stationary states (\cite{[CM],[WMK]}), the asymptotic
behaviour in time (\cite {[HMZ]}), and the reduced dynamics on the
stable soliton manifold (\cite {[GHW]}). To
the knowledge of the authors the only rigorous result about the NLS with
$\delta'$ interaction concerns the well-posedness of the dynamics, and
is contained in \cite {[AN]}, where the whole family of point
perturbation is treated in the presence of a cubic nonlinearity.

\noindent
 We
give here a brief description of the results of the present paper.\p
We first extend the result given in \cite {[AN]} to cover global
well-posedness for the problem \ref{nlsdelta'} in the energy space ($
Q \ = \ H^1 (\erre^+) \oplus H^1 (\erre^-)$, coinciding with the form
domain of $H_\gamma$) in the subcritical case 
$\mu < 2$ and local well-posedness for $\mu \geq 2$. In particular, the
equation has two conserved quantities, the energy $E$ and the charge
(i.e. the $L^2$-norm)
$ M$, as in the unperturbed NLS equation. They are associated to
symmetries of the space of solutions, respectively time translation
and phase invariance. The free NLS equation has one more conserved
quantity, linear momentum, associated to space translation; this
symmetry is broken by the defect, and correspondingly there is not
momentum conservation.\p In the main part of the paper we are
concerned with the identification of the ground states for 
equation \ref{nlsdelta'} and the analysis of their stability in the
case of attractive $\delta'$ interaction. 

\noindent
A nonlinear standing wave, or a nonlinear bound state in physical
terminology, as in the linear case is a solution of the form 
$$ \psi (t) \ = \ e^{i \om t} \phi_\omega\ .$$
Correspondingly, $\phi_\omega$ fulfils the stationary equation
\be \label{stat}
H_\gamma \phi_\omega - \lambda | \phi_\omega |^{2 \mu} \phi_\omega = -\om \phi_\omega.
\ee

We call $\mathcal A$ the set of stationary states of the equation \ref{nlsdelta'}.\p
Every member of $\mathcal A$ has to be a  classical ($C^2(\RE^\pm)$)
and square integrable solution of the  standard NLS to the left and to the right of the singularity. 
This gives, for the spatial part of the standing wave, the only 
possible forms
\begin{equation*}
\phi_\omega (x) = 
\left\{ \begin{array}{cc}
\pm \lambda^{-\f 1 {2\mu}} (\mu + 1)^{\f 1
  {2\mu}} \om^{\f 1 {2\mu}} \cosh^{-\f 1 \mu} [\mu \sqrt \om ( x - x_1)], & x<0 \\
\lambda^{-\f 1 {2\mu}} (\mu + 1)^{\f 1
  {2\mu}} \om^{\f 1 {2 \mu}} \cosh^{-\f 1 \mu} [\mu \sqrt \om ( x - x_2)], & x>0.
\end{array}
\right.
\end{equation*}
Note that in the introduction we omit, to simplify notation, the dependence of $\phi$ from every parameter other than the frequency $\omega$.
The standing wave solution is represented as a solitary wave of the NLS
centered at $x_1$ on the left of the origin and a solitary wave of the
same NLS centered at $x_2$ on the right of the origin; the parameters
 $x_1, x_2$ defining the solution are to be chosen in such a
way that the function $\phi_\omega$
satisfies the boundary conditions embodied in domain of
$H_\gamma$. Eventually, they depend on the parameter $\lambda,\
\mu,\ \gamma,\ \omega$ which enter the equation. It turns
out that there exist two families of stationary states, a family
$\mathcal F_1$ whose members respect the symmetry $x_1=-x_2 $  and a
family $\mathcal F_2$ whose members do not enjoy this symmetry. 

It is an important point that the analysis of standing waves by ODE methods has a variational counterpart in the fact that the standing waves turn out to be critical points of an action functional. The action
functional for our problem is defined as
\be
\label{action0}
S_{\omega}(\phi) = \frac{1}{2}F_{\gamma}({\phi}) - \frac{\lambda}{2\mu
  +2}\|\phi\|_{2\mu+2}^{2\mu+2}+\frac{\omega}{2}\|\phi\|_2^2 = E (\phi)
+ \f \omega 2 M (\phi),\quad\quad \phi\in Q, \ee 
where we indicate by
$F_{\gamma}({\phi})$ the quadratic form associated to the $\delta'$
point interaction $H_{\gamma}$.\p It is easy to see that solutions of \eqref{stat} are critical point of \eqref{action0}. The above action is easily seen to be unbounded from below for focusing nonlinearities. We show however that the action $S_\om$ attains
a minimum when constrained on the
natural constraint $\langle S_\omega^\prime(\phi),\phi\rangle=0$ (the
so-called  Nehari manifold) which obviously contains the set of solutions to Euler-Lagrange
equations for $S_\omega$. \n
In the present paper we adhere to the customary mathematical use to call {\it ground states} the minimizers of the action on the natural constraint.
\n To prove that the constrained minimum
problem has a solution we exploit a) the boundedness from below of the
action on the associated Nehari manifold; b) the fact that the Nehari manifold is bounded
away from zero; c) classical Brezis-Lieb inequalities showing that the limit
of a minimizing sequence exists and it is an element of the minimization domain;
d) finally, minimizers turn out to be  elements of $D(H_{\gamma})$ and 
satisfy equation \eqref{stat}. 

\n
For an analysis of the analogous problem in the simpler case of NLS with a $\delta $ interaction see \cite{reika}. \par\noindent
It turns out that ground states do not exist for $\omega \leq
\frac{4}{\gamma^2}$ and that  
for every $\omega > \frac{4}{\gamma^2}$ there is at least one ground
state. More precisely, for 
every $\omega\in (\frac{4}{\gamma^2}, \frac{4(\mu+1)}{\gamma^2\mu}) $,
there exists a 
unique (up to a phase factor) stationary state. Furthermore, it
belongs to  the family $\mathcal F_1$, so it is symmetric with respect
to the defect.  From such a family,  
a couple of non symmetric stationary states, i.e., belonging to the
family $\mathcal F_2$, bifurcates in correspondence of the value $\om^*
= \frac{4(\mu+1)}{\gamma^2\mu}$.\p 

Before focusing on the issue of determining the stability of the
stationary states,  
we recall that,
to the $U(1)$-invariance of equation \eqref{nlsdelta'},
the appropriate notion of stability for our problem 
corresponds to the so-called {\it orbital
  stability}, which is Lyapunov stability up to symmetries. The
stationary state $\phi_{\omega}$ is orbitally stable if given a
tubular neighbourhood $U({\phi_{\omega}})$ of the orbit of the ground
state (i.e., a neighbourhood modulo symmetries), the evolution
$\psi(t)$ is in  $U({\phi_{\omega}})$ for all 
times if the $\psi(0)$ datum is near as it needs to the orbit of
$\phi_{\omega}$. Precise definitions are given in section \ref{sec:stab}. 

There are two
main approaches to determine orbital stability or instability of stationary states:
the variational or direct method, pioneered in the classic paper \cite{[CL]}, and linearization.  
In our analysis we employ the linearization method, studied in a
rigorous way first by Weinstein (\cite{[W2]} \cite{[W3]}) and then developed 
as a general theory for hamiltonian systems by Grillakis, Shatah, and Strauss (\cite{[GSS1]},
\cite{[GSS2]}). ``Linearization'' in this context is a conventional
denomination because the theory includes in fact the control of the
nonlinear remainder.

A summary of the essential steps of the method can be sketched as
follows. The equation \eqref{nlsdelta'} can be written as a
hamiltonian system on a real Hilbert space after separating real and
imaginary part of $\psi=u+iv$. Given a stationary state, the linearization of the so constructed hamiltonian system at the point
 $\phi_\om$ is described by  the second derivative of the action
 $S'' (\phi_\om)$. Such a quantity is identified with a linear operator, whose action
 is suitably represented with the aid of operators $L_1^{\gamma, \om}$ and
 $L_2^{\gamma, \om}$ defined through
$$S''(\phi_{\omega})(u+iv)\ = \ L_1^{\gamma, \om} u + iL_2^{\gamma, \om} v\ $$
(we refer to \ref{elleunoeddue} and subsequent comments for the explicit definition).
$L_1^{\gamma, \om}$ and $L_2^{\gamma, \om}$ are easily seen to be
self-adjoint on their natural domains.\par\noindent 
Then,
 denote by $n(L_1^{\gamma, \om})$ is the number of
negative eigenvalues of $L_1^{\gamma, \om}$, and define the function

\be \nnn 
p(\omega) = \left\{             
\begin{array}{ll}                   
 1\quad {\mbox{if}} \quad \f{d^2} {d\om ^2} S_\om (\phi_\omega)>0 \\                   
   0 \quad {\mbox{if}} \quad \f{d^2} {d \om^2}
S_\om (\phi_\om) < 0          
\end{array}       \right. .
\ee


\n
Now, provided that

\n
a) the essential spectrum of $S''_\omega(\phi_\om)$ is bounded away
from zero,\par\noindent 
b) ${\mbox{Ker}}(L_2^{\gamma, \om})= {\mbox{Span}}\{\phi_\omega \}$,\par\noindent
c) $n(L_1^{\gamma, \om})=n$, \par\noindent
the stationary state $\phi_\om$ is
stable if $n-p=0$, and unstable if $ n-p$ is odd.

We
accurately compute $\f{d^2} {d \om^2} S_\om (\phi_\om)$ for the
ground states of our model, in critical and supercritical regime, and prove the occurrence of
an exchange of stability between the two subfamilies of
$\mathcal F_1$. For subcritical and critical nonlinearity power $\mu \leq 2$ and low frequency
$\omega\in (\frac{4}{\gamma^2},\frac{4(\mu+1)}{\gamma^2\mu}
) $, the symmetric (odd) states are stable, while crossing the critical
frequency $\omega^*$  the symmetric branch of stationary states becomes unstable 
and the two newborn asymmetric states prove stable.
After bifurcation ($\om =\om^*+0$), ground states are stable for not
too strong nonlinearities 
$2< \mu< \mu^\star$, where $\mu^\star$ lies between $2$ and $2.5$,
and they become unstable for $\mu > \mu^\star$. Finally, in the
supercritical  regime $\mu>2$ and large frequency $\om \gg \om^*$,
all ground states are unstable.

We are then in the
presence of a pitchfork bifurcation, 
accompanied by a spontaneous symmetry breaking on the stable branches
of the bifurcation.  The phenomenon of bifurcation of asymmetric
ground states from a branch of symmetric ones was discovered by
Akhmedeev (\cite{[AKM]}) in a model of propagation of electromagnetic
pulses in nonlinear layered media, and then studied in a rigorous way
in several mathematical papers (see \cite{[MJS]} and references
therein) with various generalizations of the original result, which
was concerned with an exactly solvable example. In these works the
bifurcation is induced by a change of sign in the nonlinearity
(i.e., a transition
to a defocusing regime) in a localized region.  Here we have an
analogous effect in a different model, in which the bifurcation is
induced by a strong point defect.  \p We remark that no such a
phenomenon shows up in the case of a single $\delta$ perturbation of
NLS. Nevertheless, a situation similar in some respects appears in the case of the
NLS with {\it two} attractive $\delta$ interactions separated by a
distance, a {\it double $\delta$ - well}. This model is studied in
\cite {[JW]}, where an analysis of the model is given by means of
dynamical systems techniques, and \cite{[FS]}, where the bifurcation
is explored in the semiclassical regime. See also \cite{[KKP]} for the
analogous phenomenon with a regular potential of double well type and
\cite{tuoc} for the introduction of a related normal form.  The analogy is
 a reminiscence of the
the fact that the $\delta'$ interaction can be approximated in the
norm resolvent sense by a suitable scaled set of three $\delta$
interactions (see \cite{[CS]} and \cite{[ENZ]}), so it may be
considered as
a singular model of a multiple well.  We finally recall that a
different definition of ground state for a semilinear Schr\"odinger
equation requires that they are minimizers of the total energy
constrained to have constant mass. The two definitions are related,
but not in an obvious way, because the constrained action is bounded
from below irrespective of the power nonlinearity $\mu$ and the
constrained energy on the contrary is bounded from below only in the
subcritical case $\mu<2$. On the other hand, the general GSS theory
guarantees that the ground states are {\em local} minima of the
constrained energy if and only if they are stable, and this gives a
connection between two definitions.

The paper is organized as
follows: after giving the main notation, we introduce the model by
recalling few elementary properties of the NLS and of the point
interaction (sections 2.1. and 2.2). Then we state the variational problem
that embodies the search for the ground states and two variants of it
(section 2.3). In section 3 we establish the well-posedness of the
dynamical problem, that turns out to be only local in time for strong
nonlinearity ($\mu \geq 2$) and global in time for weak nonlinearity
($\mu < 2$).  In section 4 we prove, by variational techniques, the
existence of a ground state (theorem \ref{teo:exmin}), 
while in section 5 we explicitly compute
the ground state (proposition \ref{propsystem}) and show the
occurrence of a bifurcation with symmetry breaking (theorem
\ref{theo:bifurcation}).  Section 6 is devoted to the issue of determining the
stability and instability of the ground states: first, in section 6.1,
we prove the spectral properties a), b), and c) for $S_\om''$
evaluated at the ground states (propositions \ref{diecitre},
\ref{l1simm}, and \ref{l1asimm}); then, in section 6.2, we compute $\f
{d^2} {d\om^2} S_\om (\psi_\om)$ and study its sign (proposition
\ref{propdisecondo}); finally, in 
section 6.3 we collect all information and prove stability,
instability and the occurrence of a pitchfork bifurcation (proposition
\ref{nature} and theorem \ref{pitchfork}).

\subsection{Notation}
For the convenience of the reader,
we collect here some notation that will be used through the paper. Most
symbols will be defined again, when introduced.

\n
- $\om_0 \ : = \ \f 4 {\gamma^2}$ is the frequency of the linear ground
state.

\n
- $\om^* \ : = \ \f 4 {\gamma^2} \f {\mu + 1} \mu$ is the frequency of bifurcation.

\n
- The symbol $f'$ denotes the derivative of the function $f$ with respect to the
space variable $x$. Time derivative is explicited through the symbol $\partial_t$. 

\n - We denote the $L^p$-norm by $\| \cdot \|_p$. When $p=2$ we omit
the subscript. The squared $H^1$-norm of $f$ is defined as the sum of
the squared $L^2$-norms of $f$ and of $f'$.

\n
- The following abuse of notation is repeatedly used:
\begin{equation*} \begin{split}
\| \psi^\prime \|^p_p \ : = \ & \lim_{\ve \to 0+} \left( \int_{-\infty}^{- \ve}
| \psi^\prime (x) |^p \, dx + \int^{\infty}_{\ve}
| \psi^\prime (x) |^p \, dx, \qquad 1 \leq p < \infty 
\right) 
\end{split}
\end{equation*}

\n
- Our framework is the energy space
\begin{equation*} \begin{split}
Q \ : = \ & H^1 (\erre^+) \oplus H^1 (\erre^-) \\
\| \psi \|_Q^2 \ : = \ & \| \psi \|^2 + \| \psi^\prime \|^2 \\
\end{split}
\end{equation*}

\n
- The ordinary scalar product in $L^2$ is denoted by $( \cdot, \cdot)$ and is antilinear
on the left factor. We often use the duality product in the space $Q$, denoted by
$\langle f, g \rangle$, where $f \in Q^\star$ and $g \in Q$. Again, it is antilinear in
$f$. For simplicity, in the bracket $\langle \cdot, \cdot \rangle$ we
often  exchange in the bracket the place of the element
of $Q$ with the place of the element of $Q^*$.

\n
- The symbol $\epsilon$ denotes the sign function.

\section{Preliminaries}
\noindent
Here we recall some well-known facts.

\subsection{Free stationary NLS}

The set of the solutions to the free (i.e. without point perturbation)
stationary  
Schr\"odinger equation with focusing power nonlinearity 
\be \nonumber 
-\psi'' - \lambda |\psi|^{2\mu} \psi = -\om \psi, 
  \qquad \om \, > \, 0, \ \, \lambda \, > \, 0, \ \, \mu \, > \, 0   
\ee
is given by $\left\{ \phi^{x_0}_\om, x_0 \in \erre, \, \om > 0 \right\}$, where
\be \label{solitons}
\phi^{x_0}_\om (x) \ = \ \lambda^{-\f 1 {2\mu}} (\mu + 1)^{\f 1
  {2\mu}} \om^{\f 1 {2\mu}} \sech^{\f 1 \mu} [\mu \sqrt \om ( x - x_0)]  
\ee
Note that they are smooth, exponentially decaying functions.

\subsection{$\delta^\prime$ interaction: hamiltonian operator}
The hamiltonian operator describing the so-called 
$\delta^\prime_0$ interaction is defined on the domain
\be \label{domgamma}
D(H_\gamma) \ = \ \{ \psi \in H^2(\erre \backslash \{ 0 \}) , 
\, \psi^\prime (0^+)=\psi^\prime (0^-)\ , 
\  \psi (0^+) - \psi (0^-) =- \gamma \psi^\prime(0^+) \}
\ee
and its action is given by
$$
H_\gamma \psi= -\psi^{\prime\prime}\ , \ \ \ \ x \neq 0\ .
$$
Note that in the current literature the opposite sign of $\gamma$ appears 
in the definition of $H_\gamma$. This is a matter of convention, 
and we prefer to adopt the present one because we consider
 the case of attractive $\delta^\prime$ interaction only.\p
An equivalent description of the elements of the operators domain $D(H_\gamma)$ is
\be \label{domgamma_2}
 \psi = \chi_+ \psi_+ + \chi_- \psi_-, 
\ee
where
$$
\psi_\pm \in  H^2(\erre\backslash \{ 0 \}) \cap H^1 (\erre)
$$
are {\it even} functions which satisfy the boundary conditions
$$ 
\psi_+ (0) - \psi_- (0) = - \gamma \psi_+^\prime (0) =
- \gamma \psi_-^\prime (0). 
$$
In this representation, the action of $H_\gamma$ is
\be \nnn 
H_\gamma \psi (x) : = - \chi_+ (x) \psi^{\prime \prime}_+ (x)
- \chi_- (x) \psi^{\prime \prime}_- (x).
\ee
Let us recall the main spectral properties of $H_{\gamma}$. 
The essential spectrum is purely absolutely continuous: $\sigma_{ess}
(H_{\gamma})=[0,+\infty)$.\p
Concerning the discrete spectrum,
if $\gamma\leq 0$, then
 $\sigma_p(H_{\gamma})=\emptyset$; if $\gamma>0$, then there exists a unique eigenvalue, 
given by $\sigma_p(H_{\gamma})=\{-\frac{4}{\gamma^2}\}$.
The corresponding normalized eigenfunction is given, $\forall\ \gamma \in (0,+\infty)$, by
$$
\psi_{\gamma}(x)=\left(\frac 2 {\gamma} \right)^{\frac{1}{2}}\ \epsilon(x)e^{-\frac 2 {\gamma}|x|}\ ,
 \quad \epsilon(x)\equiv\frac{x}{|x|} \ .
$$

\n
Finally, the singular continuous spectrum is empty: $\sigma_{sc} (H_{\gamma})=\emptyset$.\p

\n
In the following we use the notation $\om_0 = \f 4 {\gamma^2}$.

\subsection{$\delta^\prime$-interaction: quadratic form}

It is known (see \cite {[AN]} that the description of the dynamics generated by $H_\gamma$ can be
equivalently accomplished by using the quadratic form 
\be 
\label{formdelta'}
F_{\gamma} (\psi) \, : = \, \| \psi^\prime \|^2 - \gamma^{-1}| \psi(0^+)-\psi(0^-) |^2 , 
\ee
where, with an abuse of notation, we denoted $\psi^\prime : = \chi_+ \psi^\prime_+ 
+ \chi_- \psi^\prime_-$, according to the notation introduced in Section 1.1.
 The domain of $F_\gamma$ is given by  (see also Section 2.2 in \cite{[AN]})
\be 
\label{domformdelta'}
D(F_\gamma) \ := \ H^1 (\erre^+) \oplus H^1 (\erre^-).
\ee
We will consider the form domain as the natural energy space for the $\delta^\prime$ interaction; 
it is independent of $\gamma$ and from now on we will indicate it by $Q$. 
If $\psi$ belongs to the operator domain of a
$\delta^\prime$-interaction  with strength $\gamma$, then one has
\be \nnn 
F_{\gamma} (\psi) \, : = \, \| \psi^\prime \|^2 - \gamma| \psi^\prime (0) |^2\ . 
\ee
which explains the name given to the operator.


Due to \eqref{domformdelta'}, it is natural to endow the space $Q$ with the norm
\be \nnn 
\| \psi \|_Q^2 \ : = \ 
\f 1 2 \| \psi_+ \|^2_{H^1 (\erre)} + \f 1 2 \| \psi_- \|^2_{H^1 (\erre)} \ = \
\| \psi \|^2 + \| \psi^\prime \|^2.
\ee
It is immediately seen that the following proposition holds
\begin{prop} 
The form $F_\gamma$ is continuous in the topology induced by the norm $\| \cdot \|_Q$.
\end{prop}

\begin{proof}
In order to estimate the pointwise term in $F_\gamma$, notice that 
\be \begin{split}
| \psi (0+) - \psi (0-) |^2 \ \leq & \ 2 | \psi (0+) |^2 + 2 | \psi (0-) 
|^2 \ \leq \  2 \| \psi_+ \| \| \psi^\prime_+
\| +  2 \| \psi_- \| \| \psi^\prime_- \| 
 \ \leq  \ 2 \| \psi \|_Q^2,
\end{split}
\ee
where we used the well-known fact that $\| \phi \|_\infty^2 \leq \|
\phi \|_2 \| \phi'\|_2$ for any $\phi \in H^1 (\erre)$.
The first term in \eqref{formdelta'} is trivially estimated by the
squared $Q$-norm, so the proof is complete.
\end{proof}
Finally, we notice that, for any $m \geq - \om_0$ one has
\be \label{m-rep}
F_\gamma (\psi) \ = \ ( (H_\gamma + m)^{\f 1 2} \psi, (H_\gamma +
m)^{\f 1 2} \psi ) - m (\psi, \psi).
\ee
This is immediate for $\psi \in D (H_\gamma)$, then a density argument
shows that \eqref{m-rep} holds
for any $\psi \in Q$.

\subsection{Functionals and variational problems}
We first define the functionals
we use through the paper.
All of them
 act on the energy space $Q$, defined in
\eqref{domformdelta'}.\p 
We define the Hamiltonian of the NLS with a point defect as the sum of
the linear  Hamiltonian of the corresponding point interaction and of
the  nonlinear self-interacting part. \par So,  in the particular case of the $H_\gamma$ interaction, or $\delta^\prime$ defect, we have for the total energy 
\be \label{energy}
E (\psi) \ = \ \f 1 2 \| \psi^\prime \|_2^2 - \f 1 {2 \gamma} | \psi (0+) - \psi
(0-) |^2 - \f \lambda {2 \mu + 2} \| \psi \|_{2 \mu + 2}^{2 \mu + 2}\ .
\ee 
Standard results in the calculus of variations (see for example \cite{[C]} show that $E\in C^1(Q, \RE)$ and the Fr\'echet derivative 
$$
E^\prime(\psi)=H_\gamma \psi -\lambda|\psi|^{2\mu}\psi \in Q^* \quad\quad \forall \psi\in Q\ .
$$
Moreover we define the {\em mass
functional} (sometimes called charge)
\be \nnn 
 M (\psi) \ : = \ \| \psi \|^2.
\ee
The mass obviously belongs to $C^1(Q, \RE)$. Both mass and energy are conserved by the flow (see Proposition \ref{conslaws}).\p
A prominent role in the variational characterization of stationary
states is played by the {\em action functional} 
\be \label{som} 
S_\om : = \ E + \f \om 2 M \ ,
\ee
or explicitly
\be
S_\om (\psi) = \  \f 1 2 \| \psi^\prime \|^2 + \f \om 2   \| \psi
\|^2 - \f 1 {2\gamma}  | \psi (0+) - \psi (0-) |^2 - \f {\lambda} {2 \mu +
  2} \| \psi \|_{2 \mu + 2}^{2 \mu + 2}\ .
\ee
Again, $S_{\omega}\in C^1(Q, \RE)$ and for every $\psi\in Q$
$$
S_\omega^\prime (\psi) = H_\gamma \psi +\omega\psi -\lambda|\psi|^{2\mu}\psi \in Q^*\ .
$$
Stationary states $\psi_{\omega}$ satisfy $S_{\omega}^{\prime}(\psi_{\omega})=0\ .$\par
Moreover, it is useful to define the so-called {\em Nehari functional}
\be \label{iom} 
\begin{split}
I_\om (\psi) \  :&=  \langle S_{\om}^\prime(\psi),\psi\rangle \ =
  \| \psi^\prime \|^2 +  \om    \| \psi
\|^2 - \f 1 {\gamma}  | \psi (0+) - \psi (0-) |^2 -  {\lambda} 
   \| \psi \|_{2 \mu + 2}^{2 \mu + 2}\\ 
   & \  = \ \ 2 S_\om (\psi) -  \f
   {\lambda \mu} {\mu + 1} 
\| \psi \|^{2 \mu + 2}_{2 \mu + 2}  
.
\end{split}
\ee
It is immediately seen by its definition that the zero-level set of the Nehari
functional contains the stationary states associated to
the action $S_\om$. \p
The  action $S_\om$ restricted to the Nehari manifold gives the last auxiliary functional
\be \label{stilde} 
\widetilde S (\psi) : = \ 
\f {\lambda \mu} {2(\mu + 1)} 
\| \psi \|^{2 \mu + 2}_{2 \mu + 2}\  = \ S_\om(\psi) - \f 1 2 I_\om(\psi)\ .
\ee
In many physical contexts, the search for the ground states can be
formulated as follows:
\begin{prob} \label{natural}
Given $m > 0$,
find the minimum and  the minimizers
of the functional $E$ in the energy space $Q$ under
the constraint $M = m$.
\end{prob}
\p
Nevertheless, in the investigation of orbital stability of stationary
states  it proves useful to study another variational problem, namely
\begin{prob} \label{prob1}
Find the minimum and  the non vanishing minimizers
of the functional $S_\om$
in the energy space $Q$ under the constraint $I_\om \ = \ 0$.
\end{prob}
\p
This is the problem studied in the present paper. It can be rephrased
to the issue of characterizing, as explicitly as possible, the function
\be \label{dom}
d (\om) \ : = \  \inf \{S_\om (\psi),  
\psi \in Q \backslash \{ 0 \}, \ 
I_\om
(\psi) = 0 \}.
\ee
\p
Finally, studying problem \ref{prob1} is equivalent (see Section
\ref{sec:exmin})  to the 
\begin{prob} \label{prob2}
Find the minima and  the non vanishing minimizers
of the functional $\widetilde S$ in the energy
space $Q$ under the constraint $I_\om \leq 0$.
\end{prob}



Problems \ref{natural} and \ref{prob1} are related, but not in an
obvious way. Of course, when \ref{natural} has a solution $\psi$, via
Lagrange multiplier theory it turns out that there exists a real
multiplier $\omega$ such that $E^\prime(\psi)+\omega
M^\prime(\psi)=0$, which coincides with $S_{\omega}^\prime (\psi)=0$,
meaning that $\psi$ is a stationary point of $S_{\omega}$ and by
definition it belongs to $I_{\omega}$, so it is a solution of
\ref{prob1}. Nevertheless the two problems are not equivalent and a
complete analysis will be given elsewhere. See however the remarks at
the end of Section $4$ and Section $6$.

\section{\label{sec:wp}Well-posedness and conservation laws}

Here we treat the problem of the existence and uniqueness of the
solution to equation \eqref{nlsdelta'}. In the present paper we are
mainly interested in solution lying in the energy space $Q$ and in the
operator domain $D(H_\gamma)$, i.e. weak or strong solutions
respectively.
Let us stress that it is possible to obtain 
local well-posedness in
$L^2$ by proving suitable Strichartz estimates and then following the
traditional line (see e.g. \cite{[C]}); this route is followed for a general point
interaction and a cubic nonlinearity in \cite{[AN]}, to which we refer.\par\noindent
We begin with weak solutions and to this end, instead of equation \eqref{nlsdelta'} we study its integral form, i.e.
\be \label{intform}
\psi (t) \ = \ e^{-i H_\gamma t} \psi_0 + i \lambda \int_0^t  e^{-i
  H_\gamma (t-s)} | \psi (s) |^{2 \mu} \psi (s) \, ds ,
\ee
where the initial data $\psi_0$ belong to $Q$.


\n
For an exhaustive treatment of the problem given by
a general point interaction at the
origin and a cubic nonlinearity, see \cite{[AN]}. In particular, we
recall that 
the dual $Q^\star$ of the energy space $Q$, i.e. the space of the
bounded  linear functionals
on $Q$,
can be represented as
\be  
Q^\star \ = \ H^{-1} (\erre) \oplus {\rm{Span}} ( \delta(0+), \delta (0-)),
\ee
where 
the action of the functionals $\delta (0\pm)$ on a function
$\varphi \in Q$ reads
$$
\langle \delta (0\pm) , \varphi \rangle \ = \ \varphi (0\pm) .
$$

\n
As usual, exploiting formula \eqref{m-rep}, one can extend the  action
of $H_\gamma$ to the space
$Q$, with values in $Q^\star$, by
\be \label{staraction}
\langle H_\gamma \psi_1 , \psi_2 \rangle \ : = \ ( (H_\gamma + m)^{\f 1 2} \psi_1,
 (H_\gamma + m)^{\f 1 2} \psi_2 ) - m (\psi_1, \psi_2),
\ee
where $m > - \om_0$. The continuity of $H_\gamma \psi_1$ 
as a functional of the space $Q$ is immediately proved by
Cauchy-Schwarz inequality,
\eqref{staraction} and \eqref{m-rep}, that together give
\be \label{contq}
| \langle H_\gamma \psi_1 , \psi_2 \rangle | \ \leq \ C \| \psi_1 \|_Q
\| \psi_2 \|_Q.
\ee 


Now we are ready to state the following lemma.
\begin{lemma} \label{lemma31}
For any $\psi \in Q$
the identity
\begin{equation}
\label{extder}
\f d {dt} e^{-i H_\gamma t} \psi \ = \ - i H_\gamma
e^{-i H_\gamma t} \psi
\end{equation}
 holds in $Q^\star$. 
\end{lemma}

\begin{proof}
The time derivative of the functional $e^{-i H_\gamma t} \psi$ is
defined in the weak sense, namely
$$
\left\langle \f d {dt} e^{-i
  H_\gamma t} \psi, \cdot \right\rangle: = \lim_{h \to 0} \f 1 h \left[ (e^{-i H_\gamma (t
  + h)} \psi, \cdot) - (e^{-i H_\gamma t} \psi, \cdot) \right].
$$
Now, fix
  $\xi$ in the operator domain $D (H_\gamma)$
  defined in \eqref{domgamma}. Then,
\begin{equation*} \nonumber
\left\langle \f d {dt} e^{-i
  H_\gamma t} \psi, \xi \right\rangle \ = \ \lim_{h \to 0} \left( \psi, \f {e^{i
  H_\gamma (t+h)} \xi - e^{i H t} \xi} h \right)\ = \ ( \psi, i H_\gamma e^{i H_\gamma
t} \xi) \ = \
\langle -i H_\gamma e^{-i H_\gamma t} \psi, \xi \rangle,
\end{equation*}
where we used \eqref{staraction}.

Then, the result can be extended to $\xi \in Q$
  by a density argument, and by \eqref{contq} we get the continuity of
  the functional $ \f d {dt} e^{-i
  H_\gamma t} \psi$ on $Q$, so the result is proven.
\end{proof}

\begin{cor} \label{diffform}
By \eqref{extder}, one immediately has that
the formulation \eqref{nlsdelta'} of the
Schr\"odinger equation holds in $Q^\star$.
\end{cor}

In order to prove a well-posedness result we
recall from Section 4 in \cite{[AN]} that the  one-dimensional Gagliardo-Nirenberg
estimates holds for any $\psi$ in $Q$, i.e.
\begin{equation}
\label{gajardo}
\| \psi \|_{p} \ \leqslant \ C \| \psi^\prime \|^{\f 1 2 - \f 1 p}_{L^2}
\| \psi \|^{\f 1 2 + \f 1 p}_{L^2},\ \quad\quad \psi \in Q
\end{equation}
where the $C > 0$ is a positive constant which depends on the index $p$ only.

\n
Notice that from inequality \eqref{gajardo} one immediately obtains
the Sobolev-type estimate
\be \label{zobbolef}
\| \psi \|_{2 \mu + 2} \ \leq \ C \| \psi \|_Q.
\ee


\begin{prop}[Local well-posedness in $Q$]
\label{loch2}
Fixed $\psi_0 \in Q$, there exists $T > 0$ such that 
equation \eqref{intform} has a unique solution $\psi \in C^0 ([0,T),
Q)
\cap C^1 ([0,T), Q^\star)$.

\n
Moreover, eq. \eqref{intform} has a maximal solution $\psi^{\rm{max}}$
defined on an interval of the form $[0, T^\star)$, and the following ``blow-up
alternative''
holds: either $T^\star = \infty$ or
$$
\lim_{t \to T^\star} \| \psi^{\rm{max}} (t)\|_{Q}
\ = \ + \infty.
$$
\end{prop}
\begin{proof}
We denote by
$
{\mathcal X}$ the space $L^\infty ([0,T), Q),$
endowed with the norm
$
\| \psi \|_{\mathcal X} \ : = \ \sup_{t \in [0,T)} \| \psi (t)
\|_{Q}.
$
Given $\psi_0 \in Q$, we define the map $G : {\mathcal X}
\longrightarrow {\mathcal X}$ as
$$
G \phi : = e^{- i H_\gamma \cdot} \psi_0 + i \lambda \int_0^\cdot
e^{- i H_\gamma (\cdot - s)} | \phi (s) |^{2\mu} \phi (s) \, ds.
$$
Notice first that the nonlinearity preserves the space $Q$. Indeed,
referring to the decomposition \eqref{domgamma_2},
since both $\psi_\pm$ belong to $H^1 (\erre)$, the functions
$|\psi_\pm|^2 \psi_\pm$ belong to $H^1 (\erre)$ too, and so the energy
space
is preserved.

Now, estimate \eqref{gajardo} with $p = \infty$
yields
$$
\| | \phi (s) |^{2 \mu} \phi (s) \|_{Q} \ \leq \ C \| \phi (s) \|_{Q}^{2 \mu
  + 1},
$$
so
\begin{equation}
\label{contraz1}
\begin{split}
\| G \phi \|_{\mathcal X} \ \leq \ & \| \psi_0 \|_{Q}
  + C \int_0^T
\| \phi (s) \|_{Q}^{2 \mu + 1} \, ds \
\leq \ \| \psi_0 \|_{Q} + C T \| \phi \|_{\mathcal X}^{2 \mu + 1}.
\end{split}
\end{equation}
Analogously, given $\phi, \xi \in Q$,
\begin{equation}
\label{contraz2}
\begin{split}
\| G \phi - G \xi \|_{\mathcal X} \ \leqslant \ & C T
\left( \| \phi \|_{\mathcal X}^{2 \mu} + \| \xi \|_{\mathcal X}^{2\mu} \right)
\| \phi - \xi \|_{\mathcal X}\,.
\end{split}
\end{equation}
We point out that the constant $C$ appearing in \eqref{contraz1} and
\eqref{contraz2} is independent of $\psi_0$, $\phi$, and $\xi$.
Now let us restrict the map $G$ to elements $\phi$ such that $\| \phi
\|_{\mathcal X} \leqslant 2 \| \psi_0 \|_{Q}$.
    From \eqref{contraz1} and \eqref{contraz2}, if
$T$ is chosen to be strictly less than $(8C \| \psi_0 \|_{Q}^2)^{-1}$, then
 $G$ is a contraction of the ball in ${\mathcal X}$ of centre zero and
 radius
$ 2 \| \psi_0 \|_Q$, and so,
by the contraction lemma,
there exists a unique solution to \eqref{intform} in the
time
interval $[0, T)$. By a standard one-step bootstrap argument one
has that the solution actually belongs to $C^0 ([0,T),
Q)$ and,
due to the validity of \eqref{nlsdelta'} in the space
$Q^\star$ (see Corollary \ref{diffform}), we immediately have that the solution
$\psi$ actually belongs to $C^0 ([0,T), Q) \cap
C^1 ([0,T),Q^\star)$.

The proof of the
existence of a maximal solution is standard, while
the blow-up alternative  is a consequence of the fact that,
whenever the $Q$-norm of the solution is
finite, it is possible to extend it for a further time by the same
contraction
argument.
\end{proof}

The next step consists in the proof of the conservation laws.
\begin{prop} \label{conslaws}
For any solution $\psi \in C^0 ([0,T), Q)$
to
the equation \eqref{intform}, the following conservation laws hold at
any time $t$:
\begin{equation*}
\| \psi (t) \| \ = \ \| \psi_0 \|, \qquad
{E} ( \psi (t) ) \ = \ {E} ( \psi_0 ),
\end{equation*}
where the symbol $E$ denotes the energy functional introduced in 
\eqref{energy}.
\end{prop}

\begin{proof}
The conservation of the $L^2$-norm can be immediately obtained 
using Lemma \ref{lemma31} and Corollary \ref{diffform}. So,
$$
\f d {dt} \| \psi(t) \|^2 \ = \ 2 \, {\rm{Re}} \, \left\langle \psi (t) ,
\f d {dt} \psi(t) \right\rangle \ = \ 2 \, {\rm{Im}} \, \langle \psi(t) ,
H_\gamma \psi(t) \rangle \ = \ 0
$$
by the self-adjointness of $H_\gamma$.
In order to prove the conservation of the energy, notice that
$\langle \psi(t), H_\gamma \psi(t) \rangle$ is differentiable as a function of time. Indeed,
\begin{equation*}
\begin{split} &
\f 1 h \left[ \langle \psi(t+h), H_\gamma \psi(t+h) \rangle -
\langle \psi(t), H_\gamma \psi(t) \rangle \right] 
\ = \ \left\langle \f{\psi(t+h) - \psi(t)} h, H_\gamma \psi(t+h)
\right\rangle + \left\langle H_\gamma \psi(t) ,\f{\psi(t+h) - \psi(t)} h
\right\rangle
\end{split}
\end{equation*}
and then, passing to the limit $h\to0$,
\begin{equation}
\label{previous}
\f d {dt} (\psi(t) ,
H_\gamma \psi(t))\ = \ 2 \, {\rm{Re}} \, \left\langle \f d {dt} \psi(t) ,
H_\gamma \psi(t) \right\rangle \ = \ 2 \, {\rm{Im}} \, \langle | \psi(t) |^2
\psi(t) , H_\gamma \psi(t) \rangle,
\end{equation}
where we used the self-adjointness
of $H_\gamma$ and corollary \ref{diffform}.
Furthermore,
\begin{equation}
\label{naechst}
\f d {dt} (\psi(t) ,
| \psi(t) |^{2\mu} \psi(t)) \ = \ \f d {dt} (\overline{\psi^\mu(t)}
\psi (t) , \overline{\psi^\mu(t)} \psi(t))
 \ = \ ( 2 \mu + 2) \, {\rm{Im}} \,
\langle | \psi(t) |^{2\mu}
\psi(t) , H_\gamma \psi(t) \rangle.
\end{equation}
From \eqref{previous} and \eqref{naechst} one then obtains
$$
\f d {dt} {E} (\psi(t)) \ = \ \f 1 2 \f d {dt} \langle \psi(t) ,
H_\gamma \psi(t) \rangle - \f 1 {2 \mu + 2} \f d {dt} (\psi(t) ,
| \psi(t) |^{2\mu} \psi(t)) \ = \ 0
$$
and the proposition is proved.
\end{proof}
\begin{cor}
For $\mu < 2$, all
solutions to \eqref{intform} are globally defined in time.
\end{cor}

\begin{proof}
By estimate \eqref{gajardo} with $p = \infty$ and conservation of
the $L^2$-norm, there exists
a constant $M$, that depends on $\psi_0$ only, such that
$$
{E} (\psi_0) \ = \ {E} (\psi(t)) \ \geq \
\f 1 2 \| \psi^\prime (t)\|^2 - M \| \psi^\prime (t)\|
$$
Therefore a uniform (in $t$) bound on
$ \| \psi^\prime (t)\|^2$ is obtained. As a consequence,
one has that no blow-up in finite
time can occur, and therefore, by the blow-up alternative proved in
theorem \eqref{loch2}, the
solution
is global in time.
\end{proof}
\vskip10pt
Using a well established analysis which we will not repeat here (see again \cite{[AN]} for a detailed study in the case $\mu=1)$), one can get well posedness for strong solutions of \eqref{intform}, i.e. in the operator domain.
\begin{theorem}[Local and global well posedness in the operator domain] 
\label{int-strong}
\par\noindent
For any $\mu>0$ and initial datum $\psi_0 \in D(H_\gamma)$ there exists $T \in (0, + \infty)$ such that the
equation \eqref{intform}
has a unique solution $\psi \in C ( [0,T), D(H)) \cap C^1 ( [0,T), 
L^2 (\erre))$. Moreover, equation \eqref{intform} has a maximal solution $\psi (t)$
defined on the interval $[0, T^\star)$. For such a solution, the 
following alternative holds: either $T^\star = + \infty$ or 
$$ \lim_{t \to T^\star} \| \psi (t) \|_{H} \ = \ + \infty.$$ 
Moreover, if $0<\mu<2$ and $\psi_0 \in D(H_\gamma)$ the
equation \eqref{intform} has a unique global solution $\psi \in C ( \erre , D(H_\gamma)) \cap C^1 ( \erre , 
L^2 (\erre))$. \par\noindent
Finally, mass and energy are conserved quantities for both weak and strong solutions.

\end{theorem}

\section{Existence of a ground state} \label{sec:exmin}
In this section we show the existence of a solution to
problem \ref{prob1} for any  $\om > \om_0$.
More precisely, we prove the following theorem:

\begin{theorem} \label{teo:exmin}
Let $\omega > \frac4{\gamma^2}\ .$ There exists $\psi \in Q 
\backslash \{ 0 \}$, with $I_\om (\psi) = 0$, 
that solves problem \ref{prob2}, namely
\be \nnn 
S_\om (\psi) \ = \ d (\om).
\ee
\end{theorem}

\n
In order to prove theorem \ref{teo:exmin} we need four preliminary lemmas.
In the first lemma we show that problems \ref{prob1} and \ref{prob2}
share the same solutions.

\begin{lemma} \label{lemmauno}
For the functionals $S_\om$ and $\widetilde S$, defined respectively
in \eqref{som} and \eqref{stilde}, the following equalities hold:
\be \label{equivalenza}
d(\om) \ : = \ \inf \, \{ S_\om (\psi), \, \psi \in Q \backslash
\{0\}, \,  I_\om (\psi) = 0 \} \ = \ 
\inf \{ \widetilde S (\psi), \, \psi \in Q\backslash
\{0\} , \, I_\om (\psi) \leq 0 \}.
\ee
Furthermore, a function $\phi \in Q\backslash
\{0\}$ satisfies $\widetilde S (\phi) = d
(\om)$ and $I_\om (\phi) \leq 0$ if and only if $S_\om (\phi) = d
(\om)$ and $I_\om (\phi) = 0$.
\end{lemma}

\begin{proof}
Let $\phi$ be a nonzero element of $Q$ such that $I_\om (\phi) = 0$. Then, by
\eqref{iom}, $S_\om (\phi) = \widetilde S (\phi)$, therefore
\be \label{un}
 \inf \{ S_\om (\psi), \, \psi \in Q\backslash
\{0\} , \, I_\om (\psi) = 0 \} \ \geq \ 
\inf \{ \widetilde S (\psi), \, \psi \in Q\backslash
\{0\}, \, I_\om (\psi) \leq 0 \}.
\ee
On the other hand, 
let $\phi$ be an element of Q such that
$I_\om (\phi) < 0$.
Defined
\be \label{alfafi}
\alpha (\phi): =  \f{ [F_\gamma (
  \phi) + \om 
\| \phi \|^2]^{\f 1 {2 \mu}}} {\lambda^{\f 1 {2 \mu}} \| \phi \|_{2
  \mu + 2}^{1 + \f 1 \mu}}, 
\ee one can directly verify that $\alpha (\phi) < 1$, $I_\om (\alpha (\phi) \phi) =
0$, and then,
by \eqref{iom},
$$
S_\om (\alpha (\phi) \phi) \ = \ \widetilde S (\alpha (\phi) 
\phi) \ = \ \alpha(\phi)^{2 \mu + 2} \widetilde S
(\phi) \ < \  \widetilde S
(\phi), 
$$
so
\be \label{deux}
\inf \{ S_\om (\psi) , \, \psi \in Q \backslash
\{0\}, \, I_\om
(\psi) = 0 \} \ \leq \ 
\inf \{ \widetilde S (\psi) ), \, \psi \in Q \backslash
\{0\}, \,
I_\om (\psi) \leq 0 \}.
\ee
From \eqref{un} and \eqref{deux}, identity \eqref{equivalenza} is proven.

From \eqref{equivalenza}, it is obvious that, if $\phi$ minimizes
$S_\om$ on the set $I_\om = 0$, then it minimizes $\widetilde S$ on
the set $I_\om \leq 0$ too. 

\noindent
Suppose now that $\widetilde S (\phi) = d (\omega)$ (then 
$\phi$ minimizes $\widetilde S$ on
the set $I_\om \leq 0$) and $I_\om (\phi) < 0$. Defining $\alpha
(\phi)$
like in \eqref{alfafi} one finds $\alpha (\phi) < 1$, $I_\om (\alpha (\phi) \phi) =
0$, and 
$
S_\om (\alpha (\phi) \phi)  
=  \alpha (\phi)^{2 \mu + 2} \widetilde S
(\phi)  < d (\om)
$ again,
that contradicts the definition of $d(\om)$. So the lemma is proven.
\end{proof}

\begin{lemma} \label{lemmadue}
The quantity $d (\om)$ defined in \eqref{dom} is strictly positive.
\end{lemma}
\begin{proof}
First notice that, for any $ a > 0$,
\be \begin{split}
| \psi (0+) - \psi (0-) |^2 \ \leq & \ 2 | \psi (0+) |^2 + 2 | \psi (0-) 
|^2 \ \leq \  2 \| \psi_+ \| \| \psi^\prime_+
\| +  2 \| \psi_- \| \| \psi^\prime_- \| \\
 \ \leq & \ a ( \| \psi_+ \|^2 +
\| \psi_- \|^2 ) +  a^{-1} ( \| \psi_+^\prime \|^2 +
\| \psi_-^\prime \|^2 ) 
\ = \ \ 2 a \| \psi \|^2 + 2 a^{-1} \| \psi^\prime \|^2 
\end{split}
\ee
where we used decomposition \eqref{domgamma_2}, estimate
\eqref{gajardo} with $p = \infty$,  and Cauchy-Schwarz
inequality. Therefore,
\be \begin{split}
F_\gamma (\psi) + \om \| \psi \|^2  \ & \geq \ \left( 1 - \f 2 {\gamma
    a} \right) \| \psi^\prime \|^2 + \left( \om - \f {2 a} \gamma
\right) \| \psi \|^2.
\end{split}
\ee
Since $\om > \f 4 {\gamma^2}$, we can fix the parameter $a$ in such a
way that
$$ \f 2 \gamma \ < \ a \ < \ \f {\gamma \om} 2, $$
so it is proven that
\be
 F_\gamma (\psi) + \om \| \psi \|^2
\ \geq \ C \| \psi \|^2_Q. 
\label{normeq}
\ee
Then, by the estimate \eqref{zobbolef}
\be \begin{split} \nnn 
I_\om (\psi ) \ & \geq \ C \| \psi \|^2_Q - \lambda \| \psi \|_{2 \mu +
  2}^{2 \mu + 2} \
 \geq \ C_1 \| \psi \|^2_{2 \mu + 2} - \lambda  \| \psi \|^{2 \mu + 2}_{2 \mu + 2}.
\end{split}
\ee
It appears that, if $I_\om (\psi) \leq 0$, then either $\psi = 0$ or 
$$
\| \psi \|_{2 \mu + 2} \ \geq \ 
\left( \f {C_1} \lambda \right)^{\f 1 {2 \mu}} \ > \ 0.
$$
Therefore, since in problem \ref{prob2} 
we are looking for a non vanishing minimizer, and owing to the fact
that on the Nehari manifold $S_\om = \widetilde S$, it must be
$d (\om) > 0$.

\end{proof}

In the third lemma we consider a pair of functionals 
$S_\om^0$, $I_\om^0$, that correspond to the functionals
$S_\om$, $I_\om$ in the absence of the point
interaction:
\be \begin{split} \nnn 
S_\om^0 ( \psi ) \ & = \ \f 1 2 \| \psi^\prime \|^2 + \f \om 2   \| \psi
\|^2 - \f {\lambda} {2 \mu +
  2} \| \psi \|_{2 \mu + 2}^{2 \mu + 2}\\
I_\om^0 ( \psi ) \ & = \ \| \psi^\prime \|^2 + \om    \| \psi
\|^2 - {\lambda}
  \| \psi \|_{2 \mu + 2}^{2 \mu + 2}.
\end{split} 
\ee

\begin{lemma}\label{min0} 
The set of the minimizers of the functional 
$S_\om^0$ among the functions in $Q \backslash \{ 0 \}$,
satisfying
the constraint $I_\om^0  \ = \ 0$,
is given by
\be \nnn 
\{ e^{i \theta} \chi_+ \phi^0_\om , \ e^{i
  \theta} \chi_- \phi^0_\om, \, \theta \in [0, 2 \pi) \},
\ee
where the function $\phi^0_\om$ has been defined in \eqref{solitons}.
\end{lemma}

\begin{proof}
First notice that, 
reasoning like in the proof of lemma \ref{lemmauno}, one can prove
that the search for the minimizers of $S_\om^0$ among the nonzero functions in
$Q$ that satisfy $I_\om^0 = 0$, is equivalent to the search for 
the minimizers of $\widetilde S$ among the nonzero functions in $Q$
that satisfy
$I_\om^0 \ \leq \ 0$. 

Let us define the real function of a real variable
$$
d^0 (\omega) \ := \ \inf
\{
\widetilde S (\psi), \ \psi  \in Q \backslash \{ 0 \}, \ 
I_\om^0 (\psi) \leq 0 \}.
$$
Proceeding like in lemma \ref{lemmadue} one can show that $d^0
(\omega) > 0$.

Besides, we recall that $\phi_\om^0$ minimizes the functional
$\widetilde S$ among all functions in $H^1 (\erre) \backslash \{ 0 \}$ such that
$I_\om^0 = 0$.
Now, making resort to the representation \eqref{domgamma_2}, let us
consider a  generic function of $Q  \backslash \{ 0 \}$ 
supported on $\erre^+$, call it $\chi_+ \psi_+$, 
with $\psi_+ \in H^1 (\erre)$ and even,  and suppose that $I_\om^0
(\chi_+ \psi_+) \leq 0$. 
One immediately has
$$\widetilde S (\chi_+ \psi_+) \ = \ \f
1 2  \widetilde S (\psi_+) \ \geq \ \f 1 2  \widetilde S (\phi^0_\om) \ = \
  \widetilde S (\chi_+ \phi^0_\om), 
$$
so $\chi_+ \phi^0_\om$ is a minimizer of $\wt S$ among the 
functions of $Q \backslash \{ 0 \}$, supported on
$\erre^+$ and satisfying $I_\om^0 \leq 0$. 

\n
Notice that the equal sign holds if and only if $\psi_+ =
\phi^0_\om$. Otherwise, $\psi_+$ would not belong to the family
\eqref{solitons}, nevertheless, as $S_\om^0 (\psi_+) = S_\om^0
(\phi_\om^0)$,  it would be a minimizer of $S_\om^0$
among the nonzero functions in $H^1 (\erre)$ that satisfy $I_\om^0 =
0$, which is impossible.

\n
Thus, for any function $\psi \in Q  \backslash \{ 0 \}$, with 
$\psi \ = \ \chi_+ \psi_+ + \chi_- \psi_-$, and $I_\om^0 (\psi) \leq 0$, the following alternative
holds: either $I_\om^0 (\chi_+ \psi_+) \leq 0$ and so
\be
\label{egym}
 \wt S (\psi) \ \geq \  \wt S (\chi_+ \psi_+) \ \geq \ \wt S (\chi_+
 \phi^0_\om),
\ee
or $I_\om^0 (\chi_- \psi_-) \leq 0$ and so
\be \label{ketto}
 \wt S (\psi) \ \geq \ \wt S (\chi_- \psi_-) \ \geq \ \wt S (\chi_-
 \phi^0_\om),
\ee
and the equality in the last step of \eqref{egym} and \eqref{ketto}
holds if and only if $|\psi_+| =  \phi^0_\om$, or $|\psi_-| =
\phi^0_\om$, respectively.

Taking into account the  $U(1)-$symmetry of the problem, the proof is
complete. 
\end{proof}

\begin{lemma} \label{ebbravo}
For the infimum of problem \ref{prob2} the following inequality holds
\be \label{minoracao}
d(\om) \ < \ \widetilde S (e^{i \theta} \chi_+ \phi^0_\om ) \ = \ \f 1 2
\left( \f {\mu + 1} \lambda \right)^{ \f 1 \mu}  {\om^{\f 1 2 + \f 1 \mu}}
\int_0^{1} (1 - u^2)^{\f 1 \mu} \, du.
\ee
\end{lemma}

\begin{proof}
We notice that the last identity in \eqref{minoracao} can be obtained by
direct computation.
Furthermore,
$$
I_\om (\chi_+ \phi_0^\om) \ = \ I_\om^0 (\chi_+ \phi_0^\om) 
- \f 1 \gamma \left( \f {(\mu + 1) \om } \lambda \right)^{\f 1 \mu}
 \ = \ - \f 1 \gamma \left( \f {(\mu + 1) \om } \lambda \right)^{\f 1 \mu}
\ < \ 0.
$$
Following the proof of lemma \ref{lemmauno} we define
\be \nonumber \begin{split}
\alpha : = & \
 \f{ [F_\gamma ( \chi_+
  \phi_\om^0) + \om 
\| \chi_+ \phi_\om^0 \|^2]^{\f 1 {2 \mu}}} {\lambda^{\f 1 {2 \mu}} \|
\chi_+ \phi_\om^0 \|_{2 \mu + 2}^{1 + \f 1 \mu}}
\ = \ \left( 1 - \f 1 {\gamma \lambda} \f{2 | \phi^0_\om (0+) |^2}
{ \| \phi_\om^0 \|_{2 \mu + 2}^{2 \mu + 2}}
\right)^{\f 1 {2 \mu}} \ = \  \left( 1 - \f \mu {\gamma (\mu + 1) \om^{\f 1 2}
\int_0^{1}  (1 - u^2)^{\f 1 \mu} \, du }\right)^{\f 1 {2 \mu}}
\\ < \ 1.
\end{split}
\ee
Therefore, 
$I_\om^0 (\alpha \chi_+ \phi_0^\om) = 0$ and 
$$
\widetilde S (\alpha \chi_+ \phi_0^\om) \ = \ \f{\alpha^{2 \mu + 2}} 2
\left( \f {\mu + 1}  \lambda \right)^{ \f 1 \mu}  {\om^{\f 1 2 + \f 1 \mu}}
\int_0^{+\infty} (1 - u^2)^{\f 1 \mu} \, du,
$$
and, since $\alpha < 1$, we get
$$
d (\om) \leq \ \f{\alpha^{2 \mu + 2}} 2 \left( \f {\mu + 1} \lambda \right)^{ \f 1 \mu}  
{\om^{\f 1 2 + \f 1 \mu}}
\int_0^{+\infty} (1 - u^2)^{\f 1 \mu} \, du \ < \ \f 1 2 
\left( \f {\mu + 1} \lambda \right)^{ \f 1 \mu}  {\om^{\f 1 2 + \f 1 \mu}}
\int_0^{+\infty} (1 - u^2)^{\f 1 \mu} \, du,
$$
and the proof is complete.
\end{proof}

Now we can prove theorem  \ref{teo:exmin}.

\begin{proof}

\vspace{.3cm}
Let $\{ \psi_n \}$ be a minimizing sequence for the functional $\wt S$
on the set $I_\om \leq 0$. We show that there exists a
subsequence of $\{ \psi_n \}$ that converges weakly in $Q$.
First, notice that $\| \psi_n \|_Q$
is bounded. Indeed, the sequence
$\| \psi_n \|_{2\mu + 2}$ is bounded as it converges. 
Furthermore, by the lower boundedness of the form $F_\gamma$, and
recalling that $I_\om (\psi_n) \leq 0$, we have
$$ 0 \ \leq \ \left( \om - \f 4 {\gamma^2} \right) \| \psi_n \|^2 \ \leq \
F_\gamma (\psi_n) + \om \| \psi_n \|^2 \ \leq \ \lambda \ 
\| \psi_n \|_{2\mu + 2}^{2\mu + 2} \ \leq \ C,
$$
so $\|  \psi_n \|  \leq C$. 

Then, using $I_\om (\psi_n) \leq 0$ again, by the decomposition
\eqref{domgamma_2} and estimate \eqref{gajardo},
\be \begin{split}
\| \psi_n^\prime \|^2 \ & \leq \ \lambda \ 
\| \psi_n \|_{2\mu + 2}^{2\mu + 2} - \om \| \psi_n \|^2 + \f 1 \gamma
| \psi_n (0+) - \psi_n (0-) |^2 \ \leq \ C + \f 2 \gamma (|\psi_{n,+}
(0)|^2 + |\psi_{n,-}(0)|^2 ) \\ & \ \leq \
C + C \| \psi_{n,+} \| \| \psi_{n,+}^\prime \|
+ C \| \psi_{n,-} \| \| \psi_{n,-}^\prime \| \
 \leq \  C + C \left( \f 1 \ve \| \psi_n \|^2 + \ve \| \psi_n^\prime \|^2
\right).
\end{split}
\ee
Choosing $\ve$ sufficiently small,
we obtain that $\| \psi_n^\prime \|^2$ is bounded, so the sequence $\{ \psi_n \}$
is bounded in $Q$, and then,
by Banach-Alaoglu
theorem, there exists a converging subsequence, that we call $\{ \psi_n
\}$ again, in the weak topology of $Q$. We call $\psi_\infty$ the weak
limit of the sequence $\{ \psi_n \}$. 

\vspace{.3cm}
We prove that $\psi_\infty \neq 0$. To this aim, we 
show, first, that  
the sequences $\{ \psi_n (0\pm) \}$ converge to  $\psi_\infty (0\pm)
$, and, second, that
$\lim_{n \to \infty} I_\om (\psi_n) = 0$.

Let us define the functions $\varphi_\pm (x) : = \chi_\pm (x) e^{\mp x}$.
Then, integrating by parts,
$$
 (\varphi_+, \psi_n)_Q \ = \ \int_0^{+ \infty} e^{-x} \psi_n (x) \, dx -
\int_0^{+ \infty} e^{-x} \psi_n^\prime (x) \, dx \ = \ \psi_n (0+).
$$
Analogously, $(\varphi_-, \psi_n)_Q \ = \ \psi_n (0-)$.
Therefore, by weak convergence,
\be \label{halflimit}
\psi_n (0 \pm) \ = \ (\varphi_\pm, \psi_n)_Q \ \rightarrow \ (\varphi_\pm,
\psi_\infty)_Q \ = \ \psi_\infty (0 \pm),
\ee
and the first preliminary claim is proven. We prove the second claim by
contradiction, i.e., supposing that
$I_\om (\psi_n) \rightarrow 0$ is false. Then, there
must be a subsequence of $\{ \psi_n \}$, denoted by $\{ \psi_n \}$ too, such
that
$$ \lim_{n \to \infty} I_\om (\psi_n) \ = \ - \beta \ < \ 0. $$
We define the sequence $\zeta_n : = \nu_n \psi_n$, with
$$\nu_n : =  \f{ [F_\gamma (
  \psi_n) + \om 
\| \psi_n \|^2]^{\f 1 {2 \mu}}} {\lambda^{\f 1 {2 \mu}} \| \psi_n \|_{2 \mu + 2}^{1 + \f 1 \mu}}
\ < \ 1.
$$
Since
$$
\lim_{n \to \infty} \nu_n \ = \ \lim_{n \to \infty} \left[ 1 + \f{I_\om (\psi_n)} 
{\lambda \| \psi_n \|_{2 \mu + 2}^{2 \mu + 2}}  \right]^{\f 1 {2 \mu}}
\ = \  \left[ 1 - \f {\beta \mu}{2 (\mu + 1) d (\om)} \right]^{\f 1 {2 \mu}}
\ < \ 1,
$$
we obtain
$$\lim_{n \to \infty} \widetilde S (\zeta_n ) \ = \
\nu^{2 \mu + 2}  \widetilde S (\psi_n ) \ < \ 
 \widetilde S (\psi_n )
$$
and, since $
I_\om (\zeta_n)  =  0,$
it follows that the assumption that $\{ \psi_n \}$ is a minimizing sequence
is false. Therefore, it must be 
\be \label{secondpre}
\lim_{n \to \infty} I_\om (\psi_n) = 0.
\ee

To prove that $\psi_\infty \neq 0$ we proceed by contradiction again. 
Assume that $\psi_\infty = 0$. Define the sequence
$\eta_n : = \rho_n \psi_n$ with
\be \label{rhon}
\rho_n : = \f {\left[ \| \psi_n^\prime \|^2 + \om \| \psi_n \|^2
  \right]^{\f 1 {2 \mu}} }{\lambda^{\f 1 {2 \mu}}  
\| \psi_n \|_{2\mu + 2}^{1 + \f 1
    \mu}}
\ee
Using \eqref{halflimit} and \eqref{secondpre} we obtain
$$
\lim_{n \to \infty} \rho_n \ = \ \lim_{n \to \infty} \left[ 
1 + \f{I_\om (\psi_n) + \gamma^{-1} | \psi_n (0+) - \psi_n (0-) |^2}
{\lambda
\| \psi_n \|_{2\mu + 2}^{2 \mu + 2}}
    \right]^{\f 1 {2 \mu}} \ = \ 1, 
$$
and therefore
\be \nnn 
\lim_{n \to \infty} \widetilde S (\eta_n) \ = \ \lim_{n \to \infty}
\rho_n^{2 \mu + 2} \widetilde S (\psi_n) \ = \ d(\omega).
\ee
Moreover, owing to definition \eqref{rhon},
\be \nnn 
I_\om^0 (\eta_n) \ = \ I_\om^0 (\rho_n \psi_n) \ = \
\rho_n^2 \left( \| \psi_n^\prime \|^2 + \om \| \psi_n \|^2 - \lambda \rho_n^{2\mu}
\| \psi_n \|^{2 \mu + 2}_{2 \mu + 2} \right) \ = \ 0,
\ee
so, due to Lemma \ref{min0},
\be \label{emanuele}
 d (\om) \ \geq \ S_\om^0 (\chi_+ \phi_0).
\ee
On the other hand, by Lemma \ref{ebbravo} we conclude
\be \begin{split}
\nnn 
d (\om) \ & \ < \
 \widetilde S (\chi_+ \phi_0) \ \leq \ \widetilde S (\eta_n)
.
\end{split}
\ee
that contradicts \eqref{emanuele}. So the hypothesis $\psi_\infty \ = \ 0$
cannot hold.

Now we prove that $I_\om (\psi_\infty) \leq 0$. To this purpose, we 
follow 
the last lines in the proof of proposition 2 in  \cite{reika}. First,
we recall an inequality due to Brezis and Lieb (\cite{breli}): if $u_n$ 
converges to $u_\infty$ weakly in $L^p$, then
\begin{eqnarray}
\| u_n \|_{p}^{p} - \| u_n - u_\infty \|_{p}^{p} 
- \| u_\infty \|_{p}^{p} &
\longrightarrow & 0, \quad \forall \, 1 < p < \infty. \label{breli1}
\end{eqnarray}
First, we notice that if $u_n = \psi_n$ and $p = 2 \mu + 2$,
then \eqref{breli1}
yields
\be
\widetilde S (\psi_n) -\widetilde S( \psi_n - \psi_\infty)
  - \widetilde S
( \psi_\infty) \
\longrightarrow \ 0 \label{breli3}.
\ee
Further applying \eqref{breli1} to the sequence $\{\psi_n\}$  and 
to the sequence $\{\psi_n^\prime \}$ with 
$p=2$, and using \eqref{halflimit}, yields
\begin{eqnarray}
I_\om (\psi_n) - I_\om( \psi_n - \psi_\infty)  - I_\om ( \psi_\infty) &
\longrightarrow & 0 \label{breli2}.
\end{eqnarray}
Suppose $I_\om (\psi_\infty) > 0$. Then, by \eqref{breli2} and \eqref{secondpre}, 
$$ \lim_{n \to \infty} I_\om (\psi_n - \psi_\infty) \ 
= \ \lim_{n \to \infty} I_\om (\psi_n) - I_\om (\psi_\infty) 
\ = \  - I_\om (\psi_\infty) \ < \ 0.$$ 
Choose $\bar n$ such that $I_\om (\psi_{n} - \psi_\infty) <
0$ for any $n > \bar n$. Then, by definition of $d (\omega)$ we have
\be \label{partial1}
d (\omega) \ \leq \ \widetilde S (\psi_{n} - \psi_\infty), 
\quad \forall n > \bar n,
\ee
but, on the other hand, $\psi_\infty \neq 0$ implies $\widetilde S (\psi_\infty) > 0$, and,
together with
\eqref{breli3}, this yields
$$ \lim_{n \to \infty} \widetilde S (\psi_n - \psi_\infty ) \ = \  \lim_{n \to \infty} \widetilde S (\psi_n )
- \widetilde 
S (\psi_\infty) \ = \ d (\omega) -  S (\psi_\infty) \ < \ d (\omega),$$ 
that contradicts \eqref{partial1}, and so it must be $I_\om (\psi_\infty) \ \leq \ 0$. 
As a consequence, by definition of $d (\om)$,
\be
\nnn 
\widetilde  S (\psi_\infty) \ \geq \ d (\omega). 
\ee
Now, since $\psi_\infty$ is the weak limit of $\{ \psi_n \}$ in $L^{2 \mu +
  2}$, we must have
$$ \widetilde S (\psi_\infty) \ = \ \f {\lambda \mu} {2 (\mu + 1)} \| \psi_\infty
\|^{2 \mu + 2}_{2 \mu + 2} \ \leq \ \lim_{n \to \infty}  \f {\lambda \mu} {2 (\mu + 1)} 
\| \psi_n
\|^{2 \mu + 2}_{2 \mu + 2} \ = \ d (\omega)
$$
which implies
\be \label{minimizza}
\widetilde S (\psi_\infty) = d (\omega),
\ee
 and so $\psi_\infty$ is a 
solution to the minimization problem \ref{prob2}, and therefore, to the
minimization problem \ref{prob1} too. The proof is complete.


\end{proof}

\begin{cor}[Strong convergence]
If a minimizing sequence $\{ \psi_n \}$ for the problem \ref{prob2}
converges weakly in $Q$, then it converges strongly in $Q$.
\end{cor}

\begin{proof}
Formulas \eqref{breli1} and
\eqref{minimizza} prove that
$\{ \psi_n \}$ converges strongly in $L^{2 \mu + 2}$. 
As a consequence, 
$$ F_\gamma (\psi_n) + \om \| \psi_n \|^2 
\ = \ 2 \f \mu {\mu + 1}  \widetilde S(\psi_n) + I_\om (\psi_n)
\ \longrightarrow  \  \frac{\mu}{\mu + 1}  \widetilde S(\psi_\infty)
 \ \longrightarrow  \
F_\gamma (\psi_\infty) + \om \| \psi_\infty \|^2,
$$
and by
\eqref{normeq} 
this is equivalent to the strong convergence in $Q$. 
\end{proof}
 
\par\noindent We end this section by adding some remarks on the
variational problem \ref{natural}, i.e. to minimize the energy at
fixed norm. \par\noindent To be precise, let us define the manifold
\begin{equation}
\Gamma_m =\left\{\psi\in Q: \| \psi \|^2=m \right\}
\end{equation}
and 
\begin{equation}
-{\mathcal E}_m ={\rm inf} \left\{E(\psi)\ |\ \ \psi\in \Gamma_m\right\}\ .
\end{equation}
In the following results it is shown that in the supercritical regime
the constrained energy is unbounded from below and in the subcritical
regime its infimum is finite and negative. Moreover the energy is
controlled from below by $\|\cdot\|_Q$ norm.

\begin{lemma}{\rm (Behaviour of the constrained energy).}\par\noindent
1) Let $\mu>2$; then ${\mathcal E}_m=+\infty$ and the energy $E$ is
unbounded from below in $\Gamma_m$; \par\noindent 2) Let $\mu<2$; then
$0<{\mathcal E}_m<+\infty$; moreover there exist positive and finite
constants $C_1$, $C_2$ (depending on $\mu,\gamma, m$) such that
\begin{equation}\label{energybound}
E(\psi)>C_1 \|\psi\|^2_Q - C_2 \qquad \quad \forall\ \psi\in \Gamma_m;
\end{equation}
 \par\noindent 3) Let $\mu = 2$; then, there exists $m^* > 0$ such that 
for $m < m^*$ inequality \eqref{energybound} holds true.
\end{lemma}

\begin{proof}
To show 1) let us consider the trial function 
$$
\Phi(\sigma,x)=\frac{\sqrt{m}}{(2\pi{\sigma}^2)^{\frac{1}{4}}}\ e^{- \frac{|x|^2}{4{\sigma}^2} }\ .
$$
A direct calculation shows that $\Phi(\sigma,x)\in \Gamma_m\ , $ and 
$$
E(\Phi(\sigma, x))=\frac{1}{2\sigma^2}\|\Phi^\prime(1,\cdot)\|^2-
\frac{{\sigma}^{-\mu}} {2\mu+2}\|\Phi(1,\cdot)\|_{2\mu+2}^{2\mu+2}\ .
$$ This proves that for $\mu>2$ the energy is unbounded from
below. Moreover, for $\mu<2$ and $\sigma $ big enough, $E(\Phi(\sigma,
\cdot))<0$. \par\noindent Now, let us show the bound
\eqref{energybound}. The nonlinear term and the point interaction term
in the energy are dominated by the kinetic energy.  Let us consider
first the nonlinear term.  \par\noindent Gagliardo-Nirenberg estimate \eqref{gajardo}
jointly with the condition $\psi\in \Gamma_m$ give
$$ \| \psi \|_{2\mu+2}^{2\mu+2}\leq C \| \psi^\prime \|^{\mu}
\| \psi \|^{\mu + 2} = C(\|\psi^\prime \|^2)^{\frac{\mu}{2}}
(\| \psi \|^2)^{1-\frac{\mu}{2}+\mu}=C
m^{\mu}\ (\|\psi^\prime \|^2)^{\frac{\mu}{2}}\ (\|\psi \|^2)^{1-\frac{\mu}{2}}
\equiv *\ .
 $$
 
With the use of the classical elementary inequality 
$$ xy\leq \frac{x^p}{p} + \frac{y^q}{q} \ , \qquad \left(\ \frac{1}{p}
+\frac{1}{q} = 1\ \right)
$$ one obtains, for any $\varepsilon > 0$, 
\be\begin{split} * &= C
m^{\mu}\ \left(\| \psi^\prime \|^2 \ \varepsilon \right)^{\frac{\mu}{2}}\ 
 \ \left( \|\psi \|^2 \varepsilon^{-\frac{\mu}{2-\mu}} \right)^{1-\frac{\mu}{2}}
\leq C
m^{\mu}\ \left[\frac{\| \psi^\prime \|^2 \varepsilon}{\frac{2}{\mu}} +
  \frac{\| \psi \|^2\ \varepsilon^{-\frac{\mu}{2-\mu}}}{\frac{2}{2-\mu}}\right]
\\ &= C m^{\mu} \frac{\mu}{2} \varepsilon \|\psi^\prime\|^2 + C
m^{\mu+1}\frac{2-\mu}{2}\varepsilon^{-\frac{\mu}{2-\mu}}\ .
\end{split}
\ee from which it follows 
\be 
\|\psi\|_{2\mu+2}^{2\mu+2}\leq \tilde
C_1 \varepsilon \|\psi^\prime\|^2_Q +\tilde C_2 
\ee 
In an analogous way
one can treat the point interaction part of the energy. Again taking
in account that $\psi \in \Gamma_m$ and by use of Sobolev embedding in
one dimension and elementary inequalities, one has 
\be
\frac{1}{\gamma} |\psi(0^+)-\psi(0^-)|^2 \leq \frac{1}{\gamma}\left[
  \varepsilon \|\psi^\prime \|^2 +\frac{ \|\psi\|^2
  }{\varepsilon}\right]\leq \frac{1}{\gamma}\varepsilon \|\psi\|^2_Q +\delta
\ee 
Collecting the estimates for the nonlinear part and for the point
interaction part of the energy and choosing $\varepsilon$ small enough
one gets points 2).

To prove $3)$, it suffices to notice that from \eqref{gajardo} one
immediately has
$$ E (\psi) \geq \f 1 2 \| \psi' \|^2 - \f{C \lambda m} 6 \| \psi' \|^2 
- \f{\sqrt m} \gamma \| \psi' \|, $$
for any $\psi \in \Gamma_m$, and the proof is complete.
\end{proof}
Note that the constrained action attains its minimum for every
positive value of $\mu$, at variance with the constrained energy,
which is unbounded from below for $\mu>2$. Even for $\mu<2$ it is not
guaranteed that the energy constrained on $\Gamma_m$ has a minimum,
i.e. that there exist a solution to the variational problem
\be
\nnn 
-{\mathcal E}_m ={\rm min}
\left\{E(\psi)\ |\ \ \psi\in \Gamma_m\right\}\ .  \ee An analysis of
this problem for NLS with point interactions will be given
in \cite{anv}. However, let us note that if a minimum exists at
$\psi_m$, and the constraint $\Gamma_m$ is regular at $\psi_m$, there
exists a Lagrange multiplier $\Lambda_m$ such that $E^\prime(\psi_m) +
\Lambda_m M^\prime(\psi_m)=0$. This means that $\psi_m$ is a
stationary point for $S_\omega$ with $\omega=\Lambda_m$ and so
$\psi_m\in I_\om$.

\section{Identification of the ground state: bifurcation}
\label{sec:bif}

\begin{prop} \label{propsystem}
Any solution $\psi$ to Problem \ref{prob1} 
has the form
\be \label{formsol}
\psi^{x_1,x_2,\theta}_\om(x) = \left\{             
\begin{array}{ll}                   
 - e^{i \theta} \phi_\om^{x_1} (x), & \qquad x<0\\                   
   e^{i \theta} \phi_\om^{x_2} (x), & \qquad x>0              
\end{array}       \right. , 
\ee
where the functions $\phi_\om^{x_i}$ have been defined in
\eqref{solitons}, $\theta$ can be arbitarily chosen, and the couple
$(x_1, x_2)$ is determined in the following way:
denoted $t_i = \tah (\mu \sqrt \om |x_i|)$,
then $(t_1, t_2)$ solves the system
\be \label{tsystem}
 \left\{
\begin{array}{ccc}
t_1^{2 \mu} - t_1^{2 \mu + 2} & = &  t_2^{2 \mu} - t_2^{2 \mu + 2} \\ 
t_1^{-1} + t_2^{-1} & = & \gamma \sqrt \om
\end{array}
\right. .
\ee
\end{prop}

\begin{proof}
Consider the functional $J_{\om, \nu} \ = \ S_\om + \nu I_\om$, with $\nu$ a
Lagrange multiplier. Any solution of problem \ref{prob1} must then be a
stationary point for $J_{\om, \nu}$. Let $\psi$ be one of such
solutions. Then,
$$ S_\om^\prime (\psi) \psi \ = \ I_\om (\psi) \ = \ 0. $$
Furthermore,
$$ I_\om^\prime (\psi) \psi \ = \ - 2 \mu \lambda \| \psi \|_{2 \mu +
  2}^{2 \mu + 2}, $$
and therefore $J^\prime_{\om, \nu} (\psi) \psi \ = \ - 2 \nu \mu \lambda   
\| \psi \|_{2 \mu +
  2}^{2 \mu + 2}$. Thus, 
for nontrivial solutions 
it must be $\nu = 0$. We 
conclude that any non vanishing minimizer $\psi$ of the functional
$J_{\mu, \nu}$ must fulfil $S^\prime (\psi) = 0$, i.e.
\be \label{eins}
S_\om^\prime (\psi) \eta \ = \ 0, \quad \forall \eta \in Q.
\ee
Applying \eqref{eins} first to $\eta$, then to $\xi = - i \eta$, and
summing the two expressions, we find
\be \label{zwei}
B_\gamma (\psi, \eta) - \lambda ( |\psi |^{2 \mu} \psi, \eta) + \om (\psi, \eta) \ = \ 0  
,
\ee
where we used the shorthand notation
\be \begin{split} \nnn 
B_\gamma (\psi, \eta) \ : = \ ( \psi^\prime, \eta^\prime) - \f 1
\gamma ( \ov{\psi (0+)} -  \ov{\psi (0-)}) (\eta (0+) - \eta (0-)).
\end{split}
\ee
So, from \eqref{zwei} the following estimate holds.
\be \begin{split} \label{klmn}
| B_\gamma (\psi, \eta) | & \ \leq \ \lambda \| \psi \|^{2 \mu}_\infty
\| \psi \| \| \eta \| \ \leq \ C_\psi \| \eta \|, \qquad \forall \,
\eta \in Q.
\end{split} \ee
Notice that, letting $\eta$ vary among the functions vanishing in  a
neighbourhood of zero, we conclude from \eqref{klmn} that $\psi \in H^2 (\erre \backslash
\{0 \})$. Thus, for a generic $\eta \in Q$ a straightforward
computation gives
\be \label{bigamma2}
B_\gamma (\psi, \eta) \  = \ ( - \psi^{\prime \prime}, \eta ) - ( \eta
(0+) - \eta (0-) ) \left( \f {\ov{\psi (0+)} - \ov{\psi (0-)}} \gamma
  + \ov{\psi^\prime (0-)} \right), 
\ee
where we used the notation $\psi'' : = \chi_+ \psi_+'' + \chi_- \psi_-''$.
So, from \eqref{klmn} and \eqref{bigamma2}, we conclude that $\psi$ belongs to the domain
$D(H_\gamma)$ (see definition \eqref{domgamma}). As a consequence, the function
$ H_\gamma \psi - \lambda | \psi |^{2 \mu} \psi + \om \psi$ belongs to
$L^2 (\erre)$.




\noindent
Furthermore, from \eqref{zwei}
$$ (H_\gamma \psi - \lambda | \psi |^{2 \mu} \psi + \om \psi, \eta)  \
= \ 0, \qquad \forall \eta \in Q, $$
and, since $Q$ is dense in $L^2(\erre)$, 
\be \label{zero}
H_\gamma \psi - \lambda | \psi |^{2\mu} \psi + \om \psi \ = \ 0 \quad
{\rm{in}} \ L^2(\erre).  
\ee
Since $\psi$ lies in the domain of $H_\gamma$, equation \eqref{zero}
can be rewritten as 
\be \label{eqdiff} \left\{ \begin{array}{ccc}
- \psi^{\prime \prime} - \lambda | \psi |^{2 \mu} \psi + \om \psi & =
& 0, \quad x \neq 0, \quad \psi \in H^2 (\erre \backslash \{ 0 \}) \\
\psi^\prime (0+) & = & \psi^\prime (0-) \\
\psi (0+) - \psi (0-) & = & - \gamma \psi^\prime (0+)
\end{array} \right. .
\ee
Consider first the case of a real $\psi$. Then, the first equation can
be interpreted as the law of motion of a point particle with unitary
mass, moving on the
line under the action of the double-well potential $V (x) = \f \lambda {2 \mu + 2}
x^{2 \mu + 2} - \f \om 2 x^2$. 
By
standard methods of classical mechanics (see e.g. \cite{goldstein})
one immediately sees that the only solutions that vanish at infinity
correspond to the zero-energy orbits, whose shape is given by
\eqref{solitons}, where $x_0$ is a free parameter that, in the
mechanical problem, embodies the invariance under time translation. 

\n
Consider now the possibility of complex solutions. Writing $\psi (x) =
e^{i \theta (x)} \rho (x)$, the first equation in \eqref{eqdiff}
yields
$$ - \rho'' - 2 i \theta' \rho' - \lambda \rho^{2 \mu + 1} + (\om +
\theta'') \rho \ = \ 0,$$
thus, in order to make the imaginary part vanish, either $\rho'$ or
$\theta'$ must be identically equal to zero. If $\rho' = 0$, then $\psi$
either vanishes or is not an element of $L^2 (\erre)$. So it
must be $\theta' = 0$, and since $\erre \backslash \{ 0 \}$ is not
connected, one can choose a value for the phase in the positive
halfline and another value in the negative halfline. 
One then obtains that all possible solutions to \eqref{eqdiff} must be
given by
\be \label{solu}
\psi^{x_1,x_2,\theta_1, \theta_2}_\om(x) = \left\{             
\begin{array}{ll}                  
 e^{i \theta_1} \phi^{x_1}_\om (x), & \qquad x<0\\                   
e^{i \theta_2}  \phi^{x_2}_\om (x), & \qquad x>0              
\end{array}       \right. ,
\ee
where $x_1$ and $x_2$ are to be chosen in order to satisfy the
matching conditions at zero.

\par\noindent
We remark that, among the functions in \eqref{solu}, once fixed $x_1$
and
$x_2$ 
the minimum of $S_\om$ is accomplished if the
condition $  e^{i \theta_1} = -  e^{i \theta_2} $ is fulfilled. Indeed, it is clear that
such a condition minimizes the quantity $-{2\gamma}^{-1} | \psi (0+) - \psi (0-) |^2$,
while the other terms in the definition \eqref{som} of the functional
$S_\om$ are the same. This explains the phase factor in \eqref{formsol}.

Owing to the phase invariance of the problem, without losing generality we can
choose $\theta_1 = \pi$, $\theta_2 = 0$,
so
the matching 
conditions in \eqref{domgamma}
yield the following system for the unknowns $x_1$, $x_2$, and $\omega$:
\be \label{sistem1}
\left\{ \begin{array}{ccc}
\frac{\tah(\mu \sqrt \om x_1)}{\cosh^{\f 1 \mu}(\mu \sqrt \om
  x_1)} +  \frac{\tah(\mu \sqrt \om x_2)}{\cosh^{\f 1 \mu}
(\mu \sqrt \om x_2)}&=
&0  \\ & & \\
\frac{1}{\cosh^{\f 1 \mu}(\mu \sqrt \om x_2)} + 
\frac{1}{\cosh^{\f 1 \mu}(\mu \sqrt \om x_1)}& =&
\gamma \sqrt\om \frac{\tah(\mu \sqrt \om x_1)}{\cosh^{\f 1 \mu}
(\mu \sqrt \om x_1)}. 
\end{array}
\right.
\ee\par\noindent
\vskip 5pt
\par\noindent
By the first equation of system \eqref{sistem1},
$x_1$ and $x_2$ must have opposite sign. 
Furthermore, the second equation gives $x_1 > 0$.
So it is proven that $x_2 < 0 < x_1$.

Denoting $t_i = \tah(\mu \sqrt \om |x_i|)$, 
and exploiting elementary relations between hyperbolic functions,
system \eqref{sistem1} gives \eqref{tsystem}
and the proof is complete.

\end{proof}

Before explicitly showing the solutions to the problem \eqref{formsol}, 
\eqref{tsystem}, we prove a preliminary lemma.

\begin{lemma} \label{unicity}
For any $\mu > 0$, $a > 2 \sqrt{\f {\mu+1} \mu}$, there
exists a unique $\bar x \in \left( \f 2 a , 1 \right]$ such that
\be \nnn 
\f {(a^2 -1 ) \bar x^2 - 2 a  \bar x + 1}{(a \bar x - 1)^{2 \mu + 2}}
+ \bar x ^2 - 1 \ = \ 0
\ee
\end{lemma} 

\begin{proof}
Let us denote $w (x) = 
\f {(a^2 -1 ) x^2 - 2 a  x + 1}{(a x - 1)^{2 \mu + 2}} + x^2 -1$. 
First, notice
that $w \left( \f 2 a\right) = 0$.
Furthermore,
$$
w^\prime  \left( \f 2 a\right) \ = \ \f 2 a ( 4 (\mu + 1)
- \mu a^2) 
$$
as $a > 2 \sqrt{\f \mu {\mu+1}}$.
Therefore $w (x) < 0$ in some right neighbourhood of $\f 2 a$.
On the other hand, $w (1) > 0$, so  the set $\Xi$ whose element are the 
zeroes of $w$ in $\left( \f 2 a, 1 \right]$, is not empty.
Let us denote $x_1 := \min \Xi$. Then, since $w$ is regular in $\left(
  \f 2 a, + \infty \right)$,
it must be either $w(x) < 0$ or  $w(x) > 0$ in some
right neighbourhood of $x_1$. In the first case, $x_1$ is a local maximum
for $w$. Besides, since $w(1) > 0$, there exists $x_2 > x_1$ such that
$w (x_2) = 0$.
As a consequence, there are two local minima $y_1 \in \left( \f 2 a,
  x_1 \right)$, $y_2 \in (x_1, x_2)$. Owing to the mean value lemma,
there exist three points $z_1$, $z_2$ and $z_3$, lying respectively in
a neighbourhood of $y_1$, $x_1$ and $x_2$, such that $w^{\prime
  \prime} (z_1) > 0$, $w^{\prime
  \prime} (z_2) < 0$ and $w^{\prime
  \prime} (z_3) > 0$. Owing to the mean value lemma again, there exist
$s_1 \in (z_1, z_2)$ and $s_2 \in (z_2, z_3)$ such that $w^{\prime
\prime  \prime} (s_1) < 0$ and $w^{\prime
\prime  \prime} (s_2) > 0$. From the explicit expression
\be \label{wtre}
\begin{split}
w^{\prime
\prime  \prime} (x) \ = \ & [- 4 \mu a^3 (a^2 - 1)(2 \mu + 1)(\mu +
1) x^2 + 4 a^2 (\mu + 1) (2 \mu + 1) (2 \mu a^2 + 3) x \\ & + 4 a (5 \mu^2
a^2 - 2 \mu^3 a^2 + \mu + 3 \mu a^2 - 1)] {(ax - 1)^{-2 \mu - 5}}
\end{split}
\ee
it is clear that $w^{\prime
\prime  \prime} (x) < 0$ for large $x$.
It follows that $w^{\prime
\prime  \prime}$ undergoes at least two 
changes of sign in the interval $\left( \f 2 a, + \infty \right)$, but the
expression \eqref{wtre} shows that in the interval $\left( \f 1 a, +
  \infty\right)$
there is a single change of sign
only. As a consequence, our starting assumption is false and
it must be $w(x) > 0$ in some neighbourhood of $x_1$. 

\noindent
Let us suppose
that there is a point $x_2 > x_1$ such that $w(x_2) = 0$. Following
the same reasoning as before, we conclude
that $w^{\prime
\prime  \prime}$ must change sign at least twice in $\left( \f 1 a, +
\infty \right)$,  that
contradicts \eqref{wtre}.

\noindent
As a consequence, there is only one zero (i.e. $x_1$) of $w$ in $\left(\f 2 a, +
\infty\right)$, so the lemma is proven. 
\end{proof}

\begin{theorem} \label{theo:bifurcation}
If $\om_0 < \om \leq \om^*$, 
then the 
solutions to Problem \ref{prob1} are given  by $\psi^{y, -y,
  \theta}_\om $
(see definition (\ref{formsol})),
with $\theta \in \erre$ and 
\be \label{bifurcation1}
y = \f 1 {2 \mu \sqrt \om} \log \f {\gamma \sqrt \om + 2} {\gamma \sqrt \om  - 2}. 
\ee
If $ \om > \om^*$, then the solutions
to Problem \ref{prob1} are given by $\psi^{y_1, -y_2, \theta}_\om $
and $\psi^{y_2, -y_1, \theta}_\om $, with $\theta \in \erre$ and 
\begin{equation} \label{bifurcation2}
\begin{split}
y_1 \ = \ \f 1 {2 \mu \sqrt \om} \log \left| \f {1 + t_1} {1 - t_1} \right|, &
\qquad
y_2 \ = \  \f 1 {2 \mu \sqrt \om} \log \left| \f {1 + t_2} {1 - t_2} \right|,
\end{split}
\end{equation}
where the couple $(t_1, t_2)$, with $t_1 < t_2$, solves the system \eqref{tsystem}.


\end{theorem}
 
\begin{proof}
The function 
\be \label{effe}
f (t) : = t^{2\mu} - t^{2\mu+2}
\ee
 vanishes
at the points $0$ and $1$, and is strictly positive in the interval $(0,1)$.  
Furthermore, in the interval $(0, 1)$ its only stationary point is 
$\bar t : = \sqrt {\f \mu
{\mu + 1}}$, where the function $f$ has a local maximum and takes the value 
$m : = \f {\mu^\mu}{(\mu + 1)^{\mu +1}}$.

\noindent
As a consequence,
given $a > 0$, the system 
\be \label{asys}
a = f (t_1) = f (t_2)
\ee
 in the unknowns $t_1$ and $t_2$,
has no solutions for $a > m$, the unique
solution $t_1 = t_2 = \bar t$ for $a = m$, and,
imposing $t_1 < t_2$,
 three solutions for $0 \leq a < m$:
 indeed, there exists a unique couple $t_1, t_2$, with $t_1 \in [0, \bar t)$,
$t_2 \in (\bar t , 1]$, such that $f(t_1) = f(t_2) = a$. Therefore, the 
three couples
$(t_1, t_1)$ $(t_2, t_2)$, $(t_1,t_2)$, solve \eqref{asys}.

So the set of the solutions
to the first equation in \eqref{tsystem} with $t_1 \leq t_2$ consists of the union of 
\be \nnn 
T_1 : = \{ 0 \leq t_1 = t_2 \leq 1 \}
\ee 
and 
\be \label{tidue}
T_2 : = \{ (t_1, t_2), \ 0 \leq t_1 <
\bar t < t_2 < 1, \ f(t_1) = f(t_2) \}.
\ee
Due to the regularity of $f$, $T_2$ is a regular curve
(see Figure 1).

\begin{figure}
\begin{center}
{}{}\scalebox{0.43}{\includegraphics{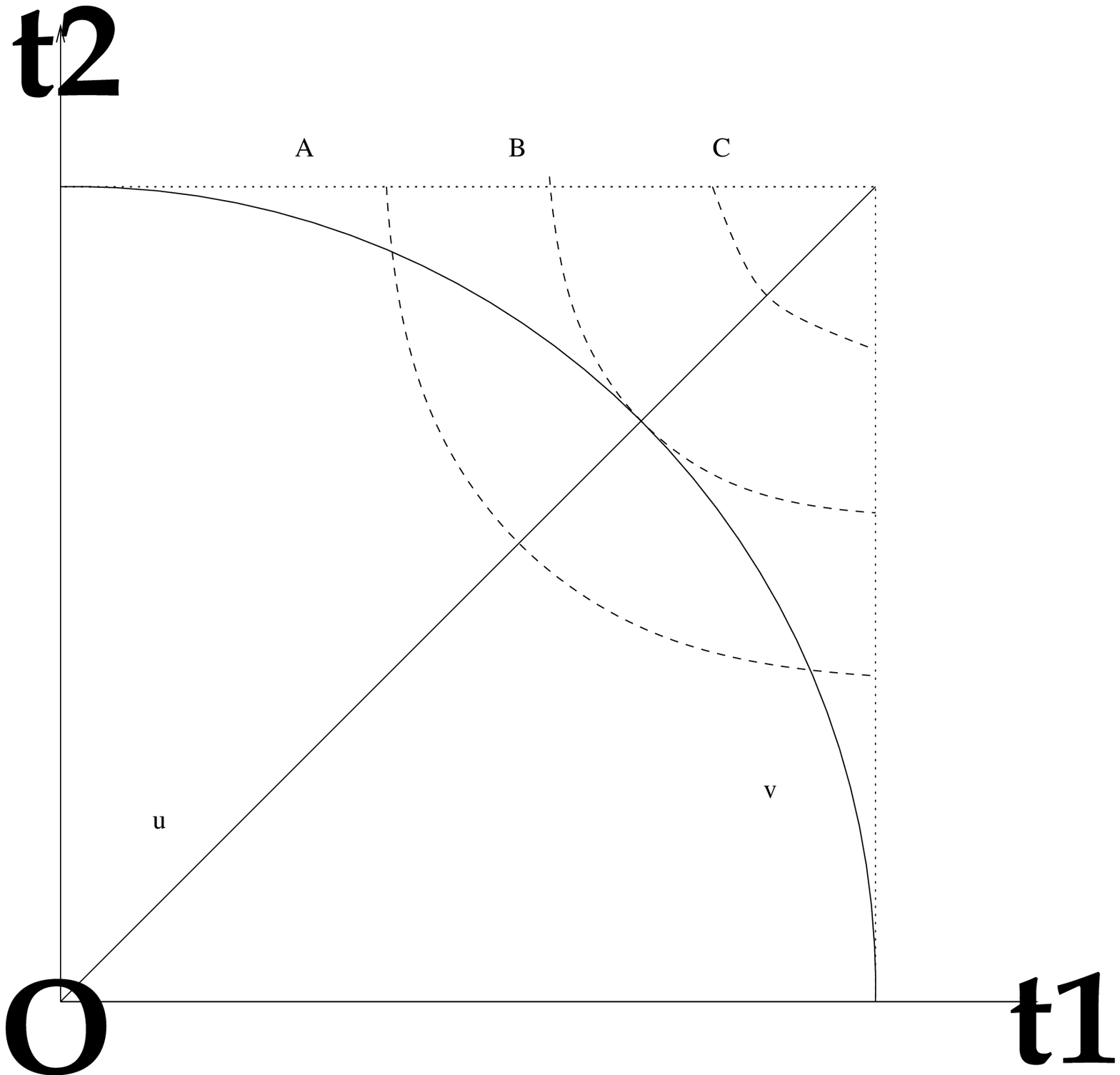}}
\caption{{\bf The system \eqref{tsystem}}. Full lines represent
the solutions to the first equation, while dashed lines represent the
family of hyperboles given by the second equation. The concavity of
the curve represented $T_2$ refers to a sufficiently large value of $\mu$
(for instance, $\mu > 1/2$). For $\mu$ near zero, $T_2$ can exhibit some
changes in convexity, but the result on bifurcation and on the number of
solutions still holds.}
\end{center}
\end{figure}

We consider the second equation in \eqref{tsystem}. Varying the
parameter $\om$, it describes a family of hyperboles in the plane $(t_1, t_2)$, whose
intersections
with $T_1$ and $T_2$ provide the required solutions to the
system \eqref{tsystem}.

First, observe that
\be \label{mint1}
\min_{t_1, t_2 \in T_1} ( t_1^{-1} +  t_2^{-1}) \ = \ 2,
\ee
and such a minimum is attained at $t_1 = t_2 = 1$.

\noindent
Second,
we claim that 
\be \label{minhyp}
\inf_{t_1, t_2 \in T_2} ( t_1^{-1} +  t_2^{-1}) \ = 2 \sqrt{\f{\mu +1}
\mu} ,
\ee
and such a value is attained at $t_1 = t_2 = \bar t$. To show this, 
we use the Lagrange multiplier method, and find that any
stationary point of the function $t_1^{-1} + t_2^{-1}$ constrained on $f(t_1) = f(t_2)$
must satisfy the condition 
$$
t_1^2 f^\prime (t_1) = - t_2^2 f^\prime (t_2), \quad t_1 {\mbox{ and
  }} t_2 \neq \bar t.
$$ 
Let us define $g(t) : = t^2 f^\prime (t)$. Notice that $g > 0$ in
$(0, \bar t)$, and $g < 0$ in $(\bar t, 1]$. Therefore, the
condition
$g (t_1) = - g (t_2)$ with $0 < t_1 < \bar t < t_2 < 1$
is equivalent to
$g^2 (t_1) = g^2 (t_2)$, $t_1 < t_2$.

Observe that
\be \label{precambio}
 g^2(t) \ = 
\int_{\bar t}^{t} \f d {ds} g^2 (s)  \, ds \ = \ 
\  \int^{t}_{\bar t}  ( 4 s^3 (f^{\prime} (s))^2 + 2 s^4  f^{\prime}
(s) f^{\prime \prime}
(s)) \, ds.
\ee
In the interval $[0, \bar t]$ the function $f$ is monotone, so it is
possible
to perform the change of variable $t \to y = f(t)$. Then
\be \label{links}
g^2 (t_1) \ = \  \int_{m}^{f(t_1)}   p ((\tau(y)) \, dy  
\ee
where we introduced the function
\be \label{py}
p (t) \ : = \ 4 t^3 f^\prime (t) + 2 t^4  f^{\prime \prime}
(t)
\ = \
(8 \mu^2 + 4 \mu) t^{2 \mu + 2}  - (8 \mu^2 + 20 \mu + 12)  t^{2 \mu + 4} 
\ee
and the function $\tau$, that 
is the inverse of $f$ in the
interval
$[0, \bar t]$.

Analogously, exploiting the monotonicity of $f$ in the interval $[\bar
t, 1]$, one can change the variable in the integral in
\eqref{precambio} and obtain
\be \label{linksbis}
g^2 (t_2) \ = \  \int_{m}^{f(t_2)}   p (\sigma(y)) \, dy
\ee
where 
the function $\sigma$ is the inverse of $f$ in the
interval
$[\bar t,1]$.

%


The function $p$ is non negative in $\left[0, \sqrt{\f {2 \mu^2 + \mu} {2
    \mu^2 + 5 \mu + 3}} \right]$, and vanishes at the endpoints of the same
interval. Its only stationary point in the interval $(0, 1)$ 
is located at $\widetilde t : = \sqrt{\f{2\mu^2
    + \mu} {2 \mu^2 + 7 \mu + 6}}$ and  is a local
maximum,
where 
\be \label{pmax}
p \left( \widetilde t \right) \ = \ \f 4 {(2 \mu + 3)^{\mu + 1}} \left( \f{2\mu^2 + \mu}
      {\mu + 2} \right)^{\mu + 2}.
\ee
In the interval $( \widetilde t, 1]$ the function $p$ is 
negative and monotonically decreasing. As a consequence, from
\eqref{linksbis} one gets
\be \label{links2}
g^2 (t_2) \ = \  \int^{m}_{f(t_2)}   | p (\sigma(y))| \, dy.
\ee
Besides, since
\be \label{pbar}
p (\bar t) \ = \ -8 \f  {\mu^{\mu + 2}}{(\mu + 1)^{\mu + 1}}.
   \ee 
one immediately has $\bar t > \sqrt{\f {2 \mu^2 + \mu}{2 \mu^2 + 5 \mu + 3}} > \wt t$.

Comparing equations \eqref{pmax} and \eqref{pbar}, 
$p \left( \widetilde t \right) < | p (\bar t) |$ if and only if
$$
( \mu + 1)^{\mu + 1} \left( \mu + \f 1 2 \right)^{\mu + 2} \ < \ 
 \left( \mu + \f 3 2 \right)^{\mu + 1} ( \mu + 2)^{\mu + 2}
$$
which holds for any value of $\mu$. Therefore, $\max_{t \in [0, \bar
  t]} | p (t) | = |p (\bar t) |$. 

\n
As a consequence, if $s_1 < \bar t < s_2$, then
$$ | p (s_1) | \ < \ | p (\bar t) | \ < \ | p (s_2) |, $$
and recalling that $t_1$ and $t_2$ are defined to fulfil $f(t_1) = f(t_2)$,
\be \label{525}
g^2 (t_1) \ = \ 2 \int_{m}^{f(t_1)} p (\tau (y)) \, dy \ < \
  2 \int^{m}_{f(t_2)} |p (\sigma (y))| \, dy \ = \ g^ 2 (t_2).
\ee
where in the last identity we used \eqref{links2}.

\noindent
It follows that the only point where $g^2 (t_1) = g^2 (t_2)$ is given by
$t_1 = t_2 = \bar t$, so there are no stationary points of
the function $t_1^{-1} + t_2^{-1}$ on the set $T_2$. By comparison
with
the endpoints $(0,1)$ and $(1,0)$ one immediately has that $t_1 = t_2
= \bar t$
corresponds to a minimum, so by direct computation our claim \eqref{minhyp} is proved.

As a consequence,
from
\eqref{mint1} and \eqref{minhyp} we have that:
\begin{itemize}
\item if $\om \leq \om_0$, then the
system \eqref{tsystem} has no solutions;
\item if $\om_0 < \om \leq \om^*$, then the only solution 
to \eqref{tsystem} lies in $T_1$ and reads $t_1 = t_2 = \f \gamma {2 \sqrt \om}$;
\item if $\om > \om^*$, then the system \eqref{tsystem}
  exhibits  three solutions: the first one lies in 
  $T_1$ and is given by $t_1 = t_2 = \f \gamma {2 \sqrt \om}$. 

\n
Furthermore,
in the set $T_2$ consider the region $t_1 <  \f \gamma {2 \sqrt \om} < t_2$;
expressing $t_1$ from the second equation of \eqref{tsystem} and plugging
it into the first one, one obtains the equation $w(t_2) = 0$, with $w$
defined as in the proof of lemma \ref{unicity} and $a = \gamma \sqrt \om$.
By virtue of the same lemma, such equation has a unique solution $t_2^*$ 
in the considered
interval. The unique value of $t_1$ such that $(t_1, t_2^*)$ is a solution
to \eqref{tsystem} is given by $(t_1^*)^{-1} = \gamma \sqrt \om - (t_2^*)^{-1}$.

\n
Due to the symmetry
  of \eqref{tsystem} under exchange of $t_1$ and $t_2$, the third and
last solution is
  given by $(t_2^*, t_1^*)$. 
\end{itemize}

\noindent
Any solution to \eqref{tsystem} singles out
a stationary point of the functional $S_\om$ on the Nehari manifold,
that
is unique up to multiplication by a phase. Obviously, the value of
$S_\om$ on such functions is independent of the phase.

Defining $y$ like in \eqref{bifurcation1}, $y_1$ and $y_2$ like in
\eqref{bifurcation2},  and owing to \eqref{propsystem}, we conclude
that:
\begin{itemize}
\item if $\om_0 < \om \leq \om^*$,
then the only stationary point (up to a phase)
for the functional $S_\om$ is given by $\psi_\om^{y, -y, 0}$.
Due to its uniqueness, it must be
the minimizer for $S_\om$ whose existence is established by theorem \ref{teo:exmin}.
The explicit expression for $y$
given in \eqref{bifurcation1} is found imposing $t_1 = t_2$ in the second
equation of \eqref{tsystem}.

\item
For $\om > \om^*$
two further solutions appear. 
Keeping into account the signs of $x_1$ and $x_2$ established in 
proposition \eqref{propsystem},
the two related families of solutions 
can be denoted by $\psi_\om^{y_1, - y_2, \theta}$, $\psi_\om^{y_2,
  -y_1, \theta}$, with $y_1$ and $y_2$ positive numbers.
Obviously, the functional $S_\om$ takes the same value on
them. In order to establish which stationary point
is the minimizer we must compare $S_\om (\psi_\om^{y_1, -y_2, \theta})$ with $S_\om (\psi_\om^{y, -y, \theta})$.
\end{itemize}
Let us proceed with such a comparison.
From \eqref{stilde} we know that the functional $S_\om$
reduces to $\widetilde S$ when evaluated on stationary states.
We have
\be \nnn 
\begin{split}
S_\om (\psi_\om^{y, - y,\theta}) \ = 
\ & \f {\om^{\f 1 2 + \f 1 \mu}
( \mu + 1 )^{\f 1 \mu }} { 2 \lambda^{\f 1 \mu }} \left[ \int_{-1}^1
( 1 - t^2)^ {\f 1 \mu} \, dt \ - \ \int_{-\f 2 {\gamma \sqrt \om}}^{\f 2 {\gamma \sqrt \om}}
( 1 - t^2)^ {\f 1 \mu} \, dt 
\right],  
\end{split} \ee
and 
\be \nnn 
\begin{split}
S_\om (\psi_\om^{y_1, -y_2, \theta}) \ = 
\ & \f {\om^{\f 1 2 + \f 1 \mu}
( \mu + 1 )^{\f 1 \mu }} { 2 \lambda^{\f 1 \mu }} \left[ \int_{-1}^1
( 1 - t^2)^{\f 1 \mu} \, dt \ - \ \int_{-t_1}^{t_2} ( 1 - t^2)^{\f 1 \mu} \, dt
\right]. 
\end{split} \ee
Introducing the function $\varphi (t) : = - \f {t} {\gamma \sqrt \om t - 1}$, we obtain
\be \nnn 
\begin{split}
S_\om (\psi_\om^{y, - y,\theta}) \ = 
\ & \f {\om^{\f 1 2 + \f 1 \mu}
( \mu + 1 )^{\f 1 \mu }} { 2 \lambda^{\f 1 \mu }} \left[ \int_{-1}^1
( 1 - t^2)^ {\f 1 \mu} \, dt \ - \ \int_{\varphi \left(\f 2 {\gamma
      \sqrt \om} \right)}^{\f 2 {\gamma \sqrt \om}}
( 1 - t^2)^ {\f 1 \mu} \, dt 
\right] \\
S_\om (\psi_\om^{y_1, -y_2, \theta}) \ = 
\ & \f {\om^{\f 1 2 + \f 1 \mu}
( \mu + 1 )^{\f 1 \mu }} { 2 \lambda^{\f 1 \mu }} \left[ \int_{-1}^1
( 1 - t^2)^ {\f 1 \mu} \, dt \ -  \int_{\varphi(t_2)}^{t_2} ( 1 - t^2)^{\f 1 \mu} \, dt
\right].
\end{split} \ee
We define the function 
\be \label{qut}
 q (t) \ : = \  \int_{\varphi (t)}^{t}
( 1 - \nu^2)^ {\f 1 \mu} \, d\nu, 
\ee
thus
\be \nnn 
\begin{split}
S_\om (\psi_\om^{y, - y,\theta}) \ = 
\ & \f {\om^{\f 1 2 + \f 1 \mu}
( \mu + 1 )^{\f 1 \mu }} { 2 \lambda^{\f 1 \mu }} \left[ \int_{-1}^1
( 1 - t^2)^ {\f 1 \mu} \, dt \ - \ q \left( \f 2 {\gamma \sqrt \om} \right)
\right] \\
S_\om (\psi_\om^{y_1, -y_2, \theta}) \ = 
\ & \f {\om^{\f 1 2 + \f 1 \mu}
( \mu + 1 )^{\f 1 \mu }} { 2 \lambda^{\f 1 \mu }} \left[ \int_{-1}^1
( 1 - t^2)^ {\f 1 \mu} \, dt \ - \ q (t_2)
\right]. 
\end{split} \ee


\noindent
Since 
$$ q^\prime (t) \ = \ 
( 1 - t^2)^{\f 1 \mu} - ( 1 - \varphi^2 (t))^{\f 1 \mu} 
\f {\varphi^2 (t)} {t^2}, $$
the stationary points of $q$ must solve the equation
\be \label{premiere}
t^2 (1 - t^2)^{\f 1 \mu} - (1 - \varphi^2 (t))^{\f 1 \mu} \varphi^2
(t) \ =
\ 0,
\ee
which is equivalent to the first equation of \eqref{tsystem} in the unknowns $(t, - \varphi (t))$. 
Furthermore, from \eqref{qut},
\be \label{deuxieme}
 t^{-1} - \varphi^{-1} (t) \ = \ \gamma \sqrt \om.
\ee
Thus, from \eqref{premiere} and \eqref{deuxieme} we conclude that
the couple $(t, - \varphi (t))$ solves the system \eqref{tsystem}.
This
implies that the functional $\widetilde S$, restricted to functions of
the kind \eqref{formsol}, has stationary points only in
correspondence with the solutions of \eqref{tsystem}. In other words,
for $\om > \om^*$ there are three stationary points for $q$, and they coincide with
$y, y_1, y_2$.

\noindent
A straightforward computation provides
$$
q^{\prime \prime} \left( \f 2 {\gamma \sqrt \om} \right) \ = \ \f 2
{\gamma \sqrt \om}
\left( 1 - \f 4 {\gamma^2 \om} \right) \left( \gamma^2 \om - 4 \, \f {\mu
    + 1} \mu \right)
$$ 
which is positive if $\om > \om^*$.
Then, $\f 2 {\gamma \sqrt \om}$ is a minimum for $q$, and since $q$ is regular
and there are no other stationary points, 
it must be
$$
S_\om (\psi_{\om}^{y_1, -y_2,\theta}) \ < \ S_\om (\psi_{\om}^{y, - y,\theta}).
$$
Thus we conclude that
$\psi_{\om}^{y_1,  -y_2,\theta}$ and $\psi_{\om}^{y_2, - y_1,\theta}$ are the
minimizers, and the proof is complete.
\end{proof}

\vskip 10pt
We end this section showing that the branch of nonlinear standing waves bifurcates from the trivial (vanishing) stationary state for $\omega>\omega_0=\frac{4}{\gamma^2}\ .$
\begin{prop}
Let us consider the branch of symmetric standing waves 
$$(\omega_0, \omega^*)\ni \omega \mapsto \psi_{\om}^{y, -y,\theta}\in L^2(\RE)\cap H^2(\erre \backslash \{ 0 \})\ . $$ The following relations hold true
$$
\lim_{\omega\rightarrow \omega_0}\ \|\psi_{\om}^{y, -y,\theta}\|\ =\
\lim_{\omega\rightarrow \omega_0}\ \|\psi_{\om}^{y, -y,\theta}\|_Q =  
\lim_{\omega\rightarrow \omega_0}\ \|\psi_{\om}^{y,
  -y,\theta}\|_{H^2(\erre \backslash \{ 0 \})} =\ 0 
$$
\end{prop}
\begin{proof}
The result immediately follows by observing that the solution
$\left(\f 2 {\gamma \sqrt \om},\f 2 {\gamma \sqrt \om} \right)$
 to \eqref{tsystem} tends to $(1,1)$ as $\om \to \om_0 + 0$. So, in
 the same limit $y$
 tends to $+ \infty$, and the result immediately follows from the
 explicit expression (\eqref{solu}, \eqref{bifurcation1})
for $\psi_\om^{y,-y,0}$.
\end{proof}


\section{Stability and instability of the ground states}

\subsection{The linearized evolution around a stationary solution}
We study the second variation of the functional $S_\om$. 
Indeed, for any $w \in Q$
$$
\f d {ds} S_\om (\psi + s w) \ = \ {\mbox{Re}} \ 
( H_\gamma \psi + \om \psi - \lambda | \psi |^{2 \mu} \psi , w ),
$$
which is often referred to as the fact that 
$S^\prime_\om [ \psi ] \ = \ H_\gamma \psi + \om \psi - \lambda 
| \psi |^{2 \mu} \psi$ in a weak sense. Therefore, linearizing the equation
\eqref{nlsdelta'} around the {\em real} function $\psi$ can be made by computing the
second variation of the functional $S_\om$. Given $w \in Q$, with 
$u = {\mbox {Re}} \, w$ and $v = {\mbox {Im}} \, w$
 one has
\be \begin{split}
\f {d^2} {ds^2} S_\om ( \psi + s w) \ = & \ \| w^\prime \|^2
+ \om \| w \|^2 - \f 1 \gamma | w (0+) - w (0-) |^2 - \lambda
\int_\erre | w (x)|^2 | \psi (x) |^{2 \mu}  dx \\
& \ - 2 \mu \lambda 
\int_\erre  u^2  (x) | \psi (x) |^{2 \mu} dx \\
\ = & \ \| u^\prime \|^2
+ \om \| u \|^2 - \f 1 \gamma | u (0+) - u (0-) |^2 - \lambda ( 2 \mu + 1)
\int_\erre  u (x)^2 | \psi (x) |^{2 \mu} dx  
 \\
\  & \ + \| v^\prime \|^2
+ \om \| v \|^2 - \f 1 \gamma | v (0+) - v (0-) |^2 - \lambda 
\int_\erre  v (x)^2 | \psi (x) |^{2 \mu} dx .  
\end{split} \ee 
Defining the operators
\be \label{elleunoeddue}
\begin{split}
L_1^{\gamma,\om}u\  & = H_{\gamma}u + \om u - (2\mu +1) \lambda |\psi|^{2\mu} u \\   
L_2^{\gamma,\om}v\ & = H_{\gamma}v + \om v -\lambda |\psi|^{2\mu} v 
\end{split}
\ee
on the domain 
$D(L_1^{\gamma,\om})=D(L_2^{\gamma,\om})=D(H_{\gamma})$ (see \eqref{domgamma}),
we get
\be \begin{split}
\f {d^2} {ds^2} S_\om ( \psi + s w) \ = & \
( u, L_1^{\gamma, \om} u ) + ( v, L_2^{\gamma, \om} v ).
\end{split} \ee 

Now we  derive the general spectral properties of the operators 
$L_1^{\gamma,\om}$ and $L_2^{\gamma,\om}$,
 needed to prove stability or instability of the stationary states. 

It is easy to show that $L_1^{\gamma,\om}$ and $L_2^{\gamma,\om}$ are
self-adjoint operators in $L^2(\RE)$. In fact they are abstract
Schr\"odinger operators of the form $(H_\gamma 
+ \om) + V_{i}(x) $ , where the perturbation $V_i(x)={\rm c}_i
|\psi|^{2\mu}(x) \ $  is given by a bounded and rapidly 
decaying function, and ${\rm c}_1=2\mu+1\ $, ${\rm c}_2=1\ .$  Let us
consider the couple of  
operators $L_2^{\gamma,\om}$ and $- \frac{d^2}{dx^2} + \om -\lambda
|\psi|^{2\mu} $. Both are self-adjoint extensions of the same closed
symmetric operator with defect indices $(2,2)$; so their resolvents
differ for a finite rank operator.  As a consequence, 
thanks to the Weyl's theorem (see \cite{rs4}, Theorem XIII.4),
the essential spectra
of $L_2^{\gamma,\om}$ and  $- \frac{d^2}{dx^2} + \om -\lambda
|\psi|^{2\mu} $ coincide. Moreover,  $\sigma_{e}(- \frac{d^2}{dx^2} + \om -\lambda
|\psi|^{2\mu})=\sigma_{e}(- \frac{d^2}{dx^2} + \om)=[\omega +\infty)$,
because the potential  $V_2$ is $(- \frac{d^2}{dx^2} + \om)$-compact.
The same reasoning holds for the operator  $L_1^{\gamma,\om}$, so we can
conclude
$$\sigma_e(L_1^{\gamma,\om})=\sigma_e(L_2^{\gamma,\om})=[\om,+\infty)
  .$$ \p 
Moreover, the fact that $L_1^{\gamma,\om}$ and $L_2^{\gamma,\om}$ are symmetric and relatively compact
  perturbations of the self-adjoint nonnegative operator $H_\gamma + \om$ allows to conclude
  (see for example \cite{[HT]}, Theorem 6.32) that the
  possible discrete spectrum is finite or accumulates at the border of
  the essential spectrum, which in our case is positive. So the negative spectrum
  is finite. \par\noindent
We will often use the previous remarks without repeating the
 argument.\par\noindent
We need more detailed spectral information on the operators $L_1^{\gamma,\om}$ and $L_2^{\gamma,\om}$, in particular concerning the number of negative eigenvalues.
A standard  technique to deal with this sort of problems 
in the case of operators with domain elements which are regular enough
 (typically Schr\"odinger operator with a smooth enough potential) relies 
on the Sturm oscillation theorem which relates the number of nodes of
an  eigenfunction to the ordering of the corresponding eigenvalue.
So, if $\psi$ is positive, then it  
coincides with the first eigenfunction, which is simple and corresponds
 to the ground state. This reasoning is not applicable in our case, due to the singular character of $H_{\gamma}$, with possibly discontinuous domain elements. 
By the way, the problem is not completely settled neither for the case of the 
milder $\delta$ interaction, so we give an independent proof of the
relevant spectral properties for this case too.\p
The results are based on a generalization of a ground state
transformation for the singular operator $H_{\gamma}$. 
\begin{prop} \label{diecitre}
Let $e^{i \om t} \psi (x)$ be a stationary solution to problem
\eqref{intform} with $\psi (0+) \psi (0-) < 0$.
Then, for the operator $L_2^{\gamma, \om}$ defined in
\eqref{elleunoeddue}, the following statements hold:
\p
a) ${\rm Ker}\ L_2^{\gamma, \om} \ =\ {\rm{Span}} \left\{\psi \right\}$\p
b) $L_2^{\gamma, \om} \ \geq \ 0$

\end{prop}
\begin{proof} 
Along this proof we denote the operator $L_2^{\gamma, \om}$ by $L_2$.
Proceeding like in the proof of Proposition \ref{propsystem} up to formula \eqref{zero},
one obtains that the function $\psi$ is regular in $\erre \backslash \{0 \}$ and
fulfils the boundary conditions defined by the
$\delta^\prime$-interaction. So,
by the equation for the stationary states \eqref{zero} again, one immediately verifies
that $L_2 \psi = 0$ and point $a)$ is proven.

\n
To prove $b)$, notice that for any $\phi \in D(L_2)$ the following
identity holds at any point $x \neq 0$:
$$
-\phi^{\prime\prime} + \om \phi -|\psi|^{2\mu} \phi = -\frac{1}{\psi}
\frac{d}{dx}\left(\psi^2\frac{d}{dx}\left(\frac{\phi}{\psi}\right)\right)\ ;
$$
integrating by parts,
\be \label{gsr}
( \phi, L_2 \phi) \ = \ \int_{-\infty}^0
\psi^2\left|\frac{d}{dx}\left( \frac{\phi}{\psi} \right)  \right|^2 dx +
\int^{\infty}_0
\psi^2\left|\frac{d}{dx}\left( \frac{\phi}{\psi} \right) \right|^2 dx +
\lim_{\varepsilon\to 0} \left[\phi^\prime \overline{\phi}
  -\frac{\psi^\prime}{\psi}|\phi|^2\right]_{-\varepsilon}^{+\varepsilon}. 
\ee
The integral terms in \eqref{gsr} are non negative and equal zero if
and only if $\phi = \psi$. Let us focus on the contribution of the
boundary, that consists of two terms. Using boundary conditions, the first term gives
\be \label{arifirst} 
\lim_{\ve \to 0} \left[ \phi^\prime \bar \phi \right]_{- \ve}^\ve \ = \
- \gamma | \phi^\prime (0+) |^2
\ee
For the second term we immediately get
\be \label{arisecond} 
- \lim_{\ve \to 0} \left[ \f{\psi^\prime}{\psi} | \phi |^2
\right]_{- \ve}^{\ve} \ = \ \f { \psi^\prime (0-) \psi (0+) | \phi (0-) |^2
-  \psi^\prime (0+) \psi (0-) | \phi (0+) |^2} {\psi (0+) \psi (0-)}
\ee
Summing \eqref{arifirst} and \eqref{arisecond}, and using the matching condition for both $\psi$ and
$\phi$ we finally get
\be \nnn 
\begin{split}
- \lim_{\ve \to 0} \left[ \f{\psi^\prime}{\psi} | \phi |^2
\right]_{- \ve}^{\ve} \ = \
& -  \f{  \psi^2 (0+) | \phi (0-) |^2 +  \psi^2 (0-) | \phi (0+) |^2
- 2 \psi (0+) \psi (0-) {\mbox{Re}} (  \phi (0+) \overline{\phi (0-)} )} 
{\gamma  \psi (0+)  \psi (0-)} \\
= \ & -  \f{  \left| \psi (0+)  \phi (0-) -  \psi (0-)  \phi (0+) \right|^2 }
{\gamma  \psi (0+)  \psi (0-)} 
\end{split}\ee
Due to the hypothesis $ \psi (0+) \psi (0-) < 0$, verified by all
ground states, we conclude that the boundary term in \eqref{gsr} is non negative, and
this completes the proof.
\end{proof}
\p
\begin{rem}{\em
An analogous proposition holds in the case, treated in \cite{reika}, \cite{fukujean},
\cite{lacozza} of a $\delta$ 
interaction with strength $\alpha$ 
(no matter if attractive or repulsive). 
One has that (the meaning of the symbols is the obvious one)\p 
a) ${\rm Ker}\ L_2^{\alpha,\om}={\mbox{Span}} \left\{\psi \right\}$\p
b) $L_2^{\alpha,\om}\geq 0$\p
In this case, the boundary term in \eqref{gsr} vanishes.}
\end{rem}

Now we prove the spectral features of interest for the operator $L_1^{\gamma, \om}$
defined in \eqref{elleunoeddue}.

\n
Consider first the case of the symmetric stationary state 
$\psi = \psi_\om^{y, -y, \theta}$. We define
\be \label{l1explicit}
L_{1,\rm{sym}}^{\gamma, \om} \ = \ H_\gamma + \om - \f {\om (\mu + 1) (2 \mu + 1)}
{\cosh^2 ( \mu \sqrt \om ( | x | + y ) )}
\ee
and
\be 
L_{2,\rm{sym}}^{\gamma, \om} \ = \ H_\gamma + \om - \f {\om (\mu + 1)}
{\cosh^2 ( \mu \sqrt \om ( | x | + y ) )}, 
\ee
where $\tanh (\mu \sqrt \om y ) = \f 2
{\gamma \sqrt \om}.$
\begin{prop} \label{l1simm}
Fixed $\mu > 0$, the operator $L_{1,{\rm{sym}}}^{\gamma, \om}$ has:

\noindent
a) A trivial kernel and one simple negative eigenvalue, if $\om < \om^*$;

\noindent
b) A one-dimensional kernel, spanned by the function
$$\xi_{-1} (x) \ = \ \f {\sinh(\mu \sqrt \om (|x| + y))} {\cosh^{1 +
    \f 1 \mu}(\mu \sqrt \om (|x| + y))} , $$
where $y$ has been defined in \eqref{bifurcation1},
 and one simple negative eigenvalue, if $\om = \om^*$;

\noindent
c) A trivial kernel and two simple negative eigenvalues or a double negative
eigenvalue, if $\om > \om^*$.
\end{prop}
\begin{proof}
For shorthand, in this proof we will denote the operators 
$L_{1,\rm{sym}}^{\gamma, \om}$ and $L_{2,\rm{sym}}^{\gamma, \om}$ 
by $L_1$ and $L_2$, respectively. Furthermore, the function $\psi_\om^{y, -y, 0}$
will be denoted by $\psi$.

Consider first the case
$\om \leq \om^*$.  By stationarity of $\psi$, the 
following identity must hold up to higher order terms in $w$: 
\be 
S_\om (\psi + w ) \ = \ 
 S_\om (\psi) + \f 1 2 ( u,  L_1 u)
 + \f 1 2 (v, L_2 v),
\ee
for any $w = u + i v$ in $Q$, with $u$ and $v$ real. 

\noindent
By Proposition \ref{diecitre}, the operator $L_2$ is
non-negative, and by Weyl's theorem on the stability of the essential
spectrum (see for example \cite {[HS]} or Theorem XIII.4 in \cite{rs4}), one has  
$\sigma_{ess}(L_1)=\sigma_{ess}(H_{\gamma} + \omega)=
[\om, + \infty)$. Furthermore, Theorem \ref{theo:bifurcation}
guarantees that  $ \psi$
minimizes the functional $S_\om$ on the 
Nehari manifold.
Thus, since the Nehari
manifold has codimension one, the operator $L_1$ has at most one negative
eigenvalue.

\n
On the other hand,
\be \label{notice}
(\psi, L_1 \psi) \ = \ - 2 \mu \lambda 
\|  \psi \|_{2 \mu + 2}^{2 \mu + 2} \ < \ 0,
\ee
so we can conclude that for $\om \leq \om^*$ the operator $L_1$ 
has exactly one negative eigenvalue.

Concerning the kernel of $L_1$,
we recall that the only square-integrable solution to the linear differential equation
\be \label{eqkernel}
- \xi^{\prime \prime} + \om \xi  - \f {\om (\mu + 1) (2 \mu + 1)}
{\cosh^2 ( \mu \sqrt \om  \ \cdot )} \xi \ = \ 0
\ee
is given, up to a factor, by
\be \label{solkernel}
\xi (x) \ = \ \f {\sinh (\mu \sqrt \om x)}{\cosh^{1 + \f 1 \mu}  (\mu \sqrt \om x)}.
\ee
Furthermore, there cannot be a solution $\zeta \notin {\mbox{Span}} (\xi)$ 
to equation \eqref{eqkernel} such that $\int_a^\infty | \zeta (x)|^2
dx < \infty$ for some finite $a$,
otherwise, by invariance under reflection of \eqref{eqkernel}, the function $\zeta (-x)$ would
be a solution to \eqref{solkernel} too, satisfying
$\int^{-a}_{-\infty} | \zeta (x)|^2 dx < \infty$,
so we would obtain three linearly independent solutions to
\eqref{eqkernel}. As a consequence, the possible solutions to the equation
\be
\label{eqkergam}  L_1 \xi
 + \om \xi  - \f {\om (\mu + 1) (2 \mu + 1)}
{\cosh^2 ( \mu \sqrt \om \  ( |\cdot | + y) )} \xi \ = \ 0
\ee
are given by $\xi_a (x) \ = \ \chi_+ (x) \xi (x + y) + a \chi_- (x)  \xi (x -
y)$, with $a \in \comple$,
provided that they
fulfil the matching condition of the $\delta^\prime$ interaction. 
Such conditions prescribe the identity of the left and the right
derivative at zero, namely $\xi^\prime (y)  = a \xi^\prime
(-y)$. Since $\xi^\prime$ is even, this implies either $a = 1$ or
$\xi^\prime (y) = 0$. In the first case, imposing the  boundary condition 
$\xi_1 (0+) - \xi_1 (0-) = - \gamma \xi_1^\prime (0+)$ leads to the equation
$ 2 \sinh (\mu \sqrt \om y) \cosh (\mu \sqrt \om y) = \gamma \sqrt \om [
\sinh^2 (\mu \sqrt \om y) - \mu]$, that cannot be solved in $y$ for any $\mu > 0$.
In the second case, one has 
$\xi_a^\prime (0) = 0$, which is
fulfilled, 
as $\xi_a^\prime (0+) = \sqrt \om \cosh^{-2-\f 1 \mu} (\mu \sqrt \om y) [ \mu -
\sinh^2 (\mu \sqrt \om y) ]$,
 if and only if $\sinh^2 (\mu \sqrt \om y) = \mu$. 
This is equivalent to $\tanh (\mu \sqrt \om y) = \sqrt{\f \mu {\mu +
    1}}$,
that, owing to
definition of $y$ in \eqref{l1explicit},
is verified only for $\om = \om^*$. Furthermore, since zero is a
stationary point, $\xi_a$ must be continuous at zero, so $a = -1$.

Thus we proved points $a)$ and  $b)$, and the case 
$\om \leq \om^*$ is exhausted.  

In order to prove point $c)$, let us write the spectrum of $L_1$ as
$$
\sigma (L_1) \  : = \ \{ \nu_1, \dots, \nu_n \} \cup
 \{ \tau_1, \dots, \tau_m \} \cup [\om, + \infty), 
$$
where $\nu_i < \nu_j < 0$ and $\tau_l > \tau_m \geq 0$ for any $i < j$
and any $l > m$. The boundedness and fast decay in $x$ of the last term in the
l.h.s. of \eqref{eqkergam} ensures that both $m$ and $n$ are finite. Moreover the essential spectrum coincides with the one of $L_1$, thanks to the Weyl's theorem again.

By \eqref{notice}, we know that $n> 0$. 
Denoted by $P_\alpha$ the orthogonal projection in $L^2$ on the
eigenspace associated to the eigenvalue $\alpha$, we define the following operators:

\n
- $P_-$ is the projection on the space $\bigoplus_j P_{\nu_j}$;

\n
- $P_+$ is the projection on the space $\bigoplus_j P_{\tau_j}$;

\n
- $P_c$ is the projection on the space associated to the 
essential spectrum $[ \om, + \infty)$ of $L_1$.

Let us suppose that $n=1$. Then, denoted by $F_1$ the
quadratic form associated to the operator $L_1$,
there exists at least a non vanishing combination
$\eta$
of $\psi$ and $\xi_{-1}$ that satisfies $F_1 (\eta) \geq 0$.
Indeed, denoted by $\psi_1$ the only (up to a phase) normalized
eigenfunction associated to the eigenvalue $- \nu_1$, we define the
function
$$ \eta \ : = \ - \f {( \psi_1, \xi_{-1})}{( \psi_1,
  \psi )}   \psi + \xi_{-1}. $$
Notice first that, since by \eqref{notice} $F_1 (\psi) < 0$, $\psi$
and $\psi_1$ cannot be orthogonal, therefore $\eta$ is well defined. Furthermore,
$\psi$ and $\xi_{-1}$ are linearly independent as they have different
parity, so it must be $\eta \neq 0$. 
Since $( \psi_1,  \eta ) = 0$, $\eta$ has no components in
the negative part of the spectrum of $L_1$, so 
$F_1 (\eta) \geq 0$.

\noindent
But this is not the case. Indeed, for a generic combination $\phi =
\alpha \psi + \beta \xi_-$, we have
\be
 F_1 (\phi) \ = \ | \alpha |^2  F_1
(\psi) + | \beta |^2  F_1
(\xi_{-1}) + 2 \mbox{Re} \bar \alpha \beta \langle L_1 \psi, \xi_{-1}
\rangle \ = \ | \alpha |^2  F_1
(\psi) + | \beta |^2  F_1^{\gamma, \om}
(\xi_{-1}), 
\label{negativo}
\ee
as the mixed term vanishes, being the scalar product of an even and an
odd function. 


\noindent
Now we compute $F_1 (\xi_{-1})$. We notice that, due to the continuity of $\xi_{-1}$,
 the term related to the point interaction
vanishes, so,
after integrating by parts, we get
\be \begin{split}
F_1 (\xi_{-1}) \ = & \ \lim_{\ve \to 0} \left[
\int_{-\infty}^{-\ve} \xi_{-1} (x) \left( - \xi_{-1}^{\prime \prime} (x) +
\om \xi_{-1} - \f {\om (2 \mu + 1) (\mu + 1)} {\cosh^2 (\mu \sqrt \om (x
- y)) } \xi_{-1} (x) \right) \, dx \right. \\
& \left. \ + \int^{\infty}_{\ve} \xi_{-1} (x) \left( - \xi_{-1}^{\prime \prime} (x) +
\om \xi_{-1} - \f {\om (2 \mu + 1) (\mu + 1)} {\cosh^2 (\mu \sqrt \om (x
+ y)) } \xi_- (x) \right) \, dx \right. \\ & \left.
+ \xi_{-1} (0) ( \xi_{-1}^\prime (0-) -  \xi_{-1}^\prime (0+)) 
\right] \\
\ = & \ \xi_{-1} (0) ( \xi_-^\prime (0-) -  \xi_{-1}^\prime (0+)),  
\end{split} \ee
where we used the fact that, by definition of $\xi_{-1}$, 
$$  
- \xi_{-1}^{\prime \prime} (x) +
\om \xi_{-1} (x) - \f {\om (2 \mu + 1) (\mu + 1)} {\cosh^2 (\mu \sqrt \om (|x|
+ y)) } \xi_{-1} (x) \ = \ 0, \qquad \forall x \neq 0.
$$
Then, one can directly compute
$$
F_1 ( \xi_{-1}) \ = \ 
- \f 4 \gamma \left( 1 - \f 4 {\gamma^2 \om} \right)^{\f 1 \mu} 
\left[ \mu - (\mu + 1) \f 4 {\gamma^2 \om} \right]
$$
which is negative if and only if $\om > \om^*$. Thus, as a consequence
of \eqref{negativo},
beyond the bifurcation frequency $\om^*$, for any linear combination 
$\phi$ of $\xi_{-1}$ and $\psi$ we have $F_1 (
\phi) < 0$,
that contradicts the hypothesis of having only one simple
negative eigenvalue in the spectrum of $L_1$.

In order to prove that actually either $n=2$ or $n=1$ and $\nu_1$ is
a double eigenvalue, we prove that $\psi$ minimizes $S_\om$
on the Nehari manifold with the additional constraint $\varphi (0+) = -
\varphi (0-)$. To this aim, we first observe that, if $\varphi \in Q$
fulfils $\varphi (0+) = - \varphi (0-)$, then
\be \begin{split}
S_\om (\varphi) \ = & \ \f 1 2 \| \varphi^\prime \|^2 + \f \om 2 
\| \varphi
\|^2 - \f \lambda {2 \mu + 2} \| \varphi \|_{2 \mu + 2}^{2 \mu + 2} - \f
2 \gamma | \varphi (0+)|^2 \\
I_\om (\varphi) \ = & \ \| \varphi^\prime \|^2 + \om  \| \varphi
\|^2 -  \lambda  \| \varphi \|_{2 \mu + 2}^{2 \mu + 2} - \f
4 \gamma | \varphi (0+)|^2 \ = \ 0
\end{split} \ee 
Consider the unitary transformation $U_\sharp$ of the space $Q$, defined by
\be \nnn 
\varphi_\sharp (x) : = (U_\sharp \varphi) (x) \ : = 
\ \epsilon(x) \varphi (x)
\ee
and notice that, if $\varphi (0+) = - \varphi (0-)$, then 
$\varphi_\sharp$ belongs
to $H^1(\erre)$, so the minimization problem is mapped into the
problem of minimizing the functional 
\be \label{mappinginto}
S_{\om,\sharp} (\varphi_\sharp) \ = \ 
\f 1 2 \| \varphi^\prime_\sharp \|^2 + \f \om 2 \| \varphi_\sharp
\|^2 - \f \lambda {2 \mu + 2} \| \varphi_\sharp \|_{2 \mu + 2}^{2 \mu + 2} - \f
2 \gamma | \varphi_\sharp (0)|^2
\ee
among the functions in $H^1 (\erre)$ that satisfy the constraint
\be \nnn
I_{\om,\sharp} ( \varphi_\sharp) \ =  \ \|  \varphi^\prime_\sharp \|^2 
+ \om  \|  \varphi_\sharp 
\|^2 -  \lambda  \|  \varphi_\sharp  \|_{2 \mu + 2}^{2 \mu + 2} - \f
4 \gamma |  \varphi_\sharp  (0)|^2 \ = \ 0.
 \ee 
Problem \eqref{mappinginto}
corresponds to the issue of finding the ground state for a nonlinear
Schr\"odinger equation in the presence of a $\delta$-type defect of
strength $- 4 / \gamma$. By
\cite{reika} and \cite{fukujean}, we know that the solution reads
\be \nnn 
\phi_{\om,\sharp} \ : = \ \left[ \f{(\mu + 1) \om} \lambda \right]^{\f 1 {2 \mu}}
 \cosh^{-\f
  1 \mu} \left( |x| + y \right),
\ee
so we
obtain, still up to a phase,
\be \nnn 
\phi_{\om,\sharp} \ = \ U_\sharp \psi
\ee
and, as a consequence,
$\psi$ minimizes the action on the
Nehari manifold with the additional condition $\psi (0+) = - \psi (0-)$.

\noindent
It remains to prove that such a constraint has codimension two.
We denote the constraint by 
\be
{\mc M} \ : = \ \{ \varphi \in Q, \ I_\om (\varphi ) = 0, \, \ 
\varphi (0+) = - \varphi (0-)\}.
\ee
By the operator $U_\sharp$ such a constraint is mapped to
\be
{\mc M_\sharp} \ : = \ \{ \varphi \in H^1 (\erre), \ I_{\om,\sharp} (\varphi) = 0 \}.
\ee
It is well-known (\cite{reika}, \cite{fukujean}) that ${\mc M_\sharp}$ has codimension
one as a subspace of $H^1 (\erre)$. On the other hand, since 
any function $\zeta$ in $Q$ can be decomposed as
$$ \zeta \ = \ \f 1 2 (\zeta (0+) - \zeta (0-)) \epsilon (\cdot) e^{-
  | \cdot |} + \widetilde \zeta, $$
with $\widetilde \zeta \in H^1 (\erre)$, we have
$$ Q \ = \ H^1 (\erre) \oplus {\mbox{Span}} \{ \epsilon (\cdot) e^{-| \cdot |} \}, $$ 
so $H^1(\erre)$ has codimension one as a subspace of $Q$.  
Therefore, ${\mc M_\sharp}$ has codimension
two as a subspace of $Q$. Thus, by unitarity of $U_\sharp$, ${\mc M}$ has codimension
two as a subspace of $Q$ too.

We then proved that the negative space of the operator $L_1$ has dimension
at most two and that for $\om > \om^*$ it equals exactly two. The
proof is concluded.
\end{proof}

Now we consider the case of an asymmetric ground state $\psi^{y_1, -y_2, \theta}_\om$.
We define
\be \nnn 
L_{1,\rm{asym}}^{\gamma, \om} \ = \ 
H_\gamma + \om - \f {\om (\mu + 1) (2 \mu + 1)}
{\cosh^2 \left( \mu \sqrt \om \left(  x  + \chi_+ (x) y_2 - \chi_- (x)
      y_1
\right) \right)}
\ee
and
\be \nnn 
L_{2,\rm{asym}}^{\gamma, \om} \ = \ H_\gamma + \om - \f {\om (\mu + 1)}
{\cosh^2 \left( \mu \sqrt \om \left(  x  + \chi_+ (x) y_2 - \chi_- (x)
      y_1
\right) \right)}
\ee
where $\tanh (\mu \sqrt \om y_j ) = t_j$, and $t_1, t_2$ are
the unique  positive solutions to \eqref{tsystem}
with $t_1 < t_2$.

\begin{prop} \label{l1asimm}
Fixed $\mu > 0$,  for any $\om > \om^*$ the operator $L_{1,\rm{asym}}^{\gamma, \om}$ has trivial
kernel and one simple negative eigenvalue.
\end{prop}
\begin{proof}
For shorthand, in this proof we will denote the operator
$L_{1,\rm{asym}}^{\gamma, \om}$ 
by $L_1$. 
Furthermore, the function $\psi_\om^{y_1, -y_2, 0}$
will be denoted by $\psi$.

By Theorem \ref{theo:bifurcation}, if $\om > \om^* $ then $ \psi$
is a local minimizer for the functional $S_\om$ on the 
Nehari manifold.  As a consequence, one can prove that $L_1$ has
one simple negative eigenvalue by following the proof of Proposition
\ref{l1simm} up to \eqref{notice}.

Concerning the kernel of $L_1$, one can follow the reasoning carried out
in the proof of \ref{l1simm} through \eqref{eqkernel},  \eqref{solkernel}
and conclude that the only solutions to the equation $L_1 \xi = 0$
can be given by $\xi_a (x) \ = \ \chi_+  \xi (x + y_2) + a \chi_-  \xi
(x - y_1)$,
where $a$ is a complex number, 
provided that $\xi_a$
fulfils the matching conditions at zero, that translate into the system
\be
\left\{ \begin{array}{ccc}
\f{ \mu - \sinh^2 (\mu \sqrt \om y_2)}{ \cosh^{2+ \f 1 \mu} (\mu \sqrt
  \om y_2)} & = & a \f{ \mu - \sinh^2 (\mu \sqrt \om y_1)}
{ \cosh^{2+ \f 1 \mu} (\mu \sqrt
  \om y_1)} \\ & & \\
\f{ \sinh (\mu \sqrt \om y_2)}{ \cosh^{1 + \f 1 \mu} (\mu \sqrt
  \om y_2)} + a \f{ \sinh (\mu \sqrt \om y_1)}{ \cosh^{1 + \f 1 \mu} (\mu \sqrt
  \om y_1)} & = & - \gamma \sqrt \om \f{ \sinh (\mu \sqrt \om y_2)} 
{ \cosh^{2 + \f 1 \mu} (\mu \sqrt
  \om y_2)}
\end{array}
\right.
\ee
Expliciting $a$ from the first equation, plugging it into the second,
and denoting as customary $t_i = \tanh( \mu \sqrt \om y_i)$, we get
the equation
$$
\f {t_1} { \mu - (\mu + 1) t_1^2 }+ \f {t_2} { \mu - (\mu + 1) t_2^2 }
\ = \ - \gamma \sqrt \om,
$$
that, using the second equation in \eqref{tsystem}, gives
$$
\f {1 - t_1^2} {t_1 (\mu - (\mu + 1) t_1^2 )}+ 
\f {1 - t_2^2} {t_2 (\mu - (\mu + 1) t_2^2 )}
\ = \ 0.
$$
Finally, by the first equation in \eqref{tsystem} one gets
$$
\f {1} {t_1^{2 \mu + 1} (\mu - (\mu + 1) t_1^2 )}+ 
\f {1} {t_2^{2 \mu + 1} (\mu - (\mu + 1) t_2^2 )}
\ = \ 0.
$$
Such a  problem translates to the problem of finding $t_1$ and $t_2$
such that $0 \leq t_1 < \bar t < t_2 \leq 1$, and $g^2 (t_1) = g^2
(t_2)$, where $g(t) : = t^2 f^\prime (t)$ and $f(t) = t^{2 \mu} - t^{2
  \mu + 2}$.

\n
The problem was treated in the last form in the proof of Theorem
\ref{theo:bifurcation}, from formula \eqref{precambio} up to formula
\eqref{525}. The conclusion is that it
has no solutions, so none among the functions $\xi_a$ lies in the
kernel of $L_1$, which is therefore trivial. This concludes the proof.
\end{proof}

\subsection{The sign of $d^{\prime
\prime} (\om)$}

\begin{prop} \label{propdisecondo}
Given $\mu > 0$, 
the sign of the second derivative of the function $d$,
defined in \eqref{dom},
is determined as follows:
\begin{enumerate}
\item if $0 < \mu \leq 2$, then 
$d^{\prime
\prime} (\om) > 0$ for any $\om \in (\om_0, \om^*) \cup 
(\om^*, + \infty)$.
Furthermore, $0 < d^{\prime \prime} (\om^* + 0)  < 
d^{\prime \prime} (\om^* - 0)$; 
\item if $\mu > 2$, then $d^{\prime \prime} (\om) > 0$ for $\om \in (\om_0, \om^*)$. 
Besides, $d^{\prime \prime} (\om^*-0) > 0$;
\item there exists $\mu^\star \in (2,2.5)$ such that 
\begin{enumerate}
\item If $\mu < \mu^\star$, then   $d^{\prime \prime} (\om^* + 0)> 0$, so there exists
$\om_1 (\mu) \geq \om^*$ such that $d^{\prime \prime} (\om) > 0$ for any $\om \in
(\om^*, \om_1(\mu))$;
\item if $\mu = \mu^\star$, then   $d^{\prime \prime} (\om^* + 0)= 0$;
\item if $\mu > \mu^\star$, then   $d^{\prime \prime} (\om^* + 0)< 0$;
\end{enumerate}
\item if
$\mu > 2$, then
there exists $\om_2 (\mu) \geq \om^*$ such that, if $\om > \om_2(\mu)$, then
$d^{\prime \prime} (\om) < 0$.
\end{enumerate} 
\end{prop}
\begin{proof}
First, we notice that, given $\mu > 0$, $\om > \om_0$,
and denoted by $\psi_\om$ a 
solution to the problem \ref{prob1} corresponding to the chosen value of $\mu$ and $\om$, 
one has
\be \nnn 
d^\prime (\om) \ = \  \f 1 2 \| \psi_\om \|^2.
\ee
Indeed, from definition \eqref{dom}, using
the stationarity  of $\psi_\om$ we get
\be \nonumber \begin{split}
d^\prime (\om) \ = \ &  \f d {d\om} S_\om ( \psi_\om )
  \ = \  \f d {d\om} E ( \psi_\om ) + \f 1 2 M ( \psi_\om )
+ \f \om 2  \f d {d\om} M ( \psi_\om )
 \ = \  S_\om^\prime [ \psi_\om ]  \f d {d \om} \psi_\om +  \f 1 2 M ( \psi_\om ) 
 \ = \  \f 1 2 M ( \psi_\om ).
\end{split}
\ee
Expliciting $M (\psi_\om)$ one obtains
\be
d^\prime (\om)
\ = \  \label{diprimo}
 \left( \f {\mu + 1} \lambda \right)^{\f 1 \mu}
\f {\om^{\f 1 \mu - \f 1 2}} {2 \mu} 
 \left[ \int_{\zeta_1(\om)}^1 (1-t^2)^{\f 1 \mu -1} \,  dt 
+  \int_{\zeta_2(\om)}^1 (1-t^2)^{\f 1 \mu -1} \,  dt 
\right],
\ee
provided that the change of variable $t = \tanh (\mu \sqrt \om (x -
x_1))$  $(t = \tanh (\mu \sqrt \om (x -
x_2)))$ has been performed in the first (second) integral. The lower
bounds in the intervals of integration are defined by
\be \label{zetai}
\zeta_i (\om) \ = \ \left\{ \begin{array}{ccc}
\f 2 {\gamma \sqrt{\om}}, & & \om \in \left( \om_0,
  \om^* \right] \\ & & \\
t_i, & &  \om \in \left( \om^*,
+ \infty \right) 
\end{array} \right. , \quad i=1,2
\ee
where the couple $t_1, t_2$ is the unique solution to the system \eqref{tsystem}
such that $t_1 < \bar t < t_2$, where $\bar t = \sqrt{\f \mu {\mu + 1}}$.

Differentiating \eqref{diprimo} yields
\be \label{disecondo} \begin{split}
d^{\prime \prime} (\om) \ = \ & \left( \f {\mu + 1} \lambda \right)^{\f 1 \mu}
\f {\om^{\f 1 \mu - \f 3 2}} {2 \mu} \left\{ \left( \f 1 \mu - \f 1 2
\right) \left[ \int_{\zeta_1(\om)}^1 (1-t^2)^{\f 1 \mu -1} \,  dt 
+  \int_{\zeta_2(\om)}^1 (1-t^2)^{\f 1 \mu -1} \,  dt 
\right] \right. \\ & \left.
- \om \left[ \zeta_1^\prime (\om) (1 - \zeta_1^2(\om))^{\f 1 \mu -1} +
 \zeta_2^\prime (\om)  (1 - \zeta_2^2 (\om))^{\f 1 \mu -1}
\right] \right\}.
\end{split}
\ee
Let us denote
\be \label{unoddue} \begin{split}
(I) \ : = \ &  \left( \f 1 \mu - \f 1 2
\right) \left[ \int_{\zeta_1(\om)}^1 (1-t^2)^{\f 1 \mu -1} \,  dt 
+  \int_{\zeta_2(\om)}^1 (1-t^2)^{\f 1 \mu -1} \,  dt 
\right] \\
(II) \ : = \ & - \om \left[ \zeta_1^\prime (\om) (1 - \zeta_1^2(\om))^{\f 1 \mu -1} +
 \zeta_2^\prime (\om)  (1 - \zeta_2^2 (\om))^{\f 1 \mu -1}
\right].
\end{split}
\ee

\bigskip

\noindent
{\em 1. $0 < \mu \leq 2$.}

\medskip

\noindent
The quantity $(I)$ is
positive for any $\om > \om_0$. Moreover, by \eqref{zetai},
for $\om \in (\om_0, \om^*)$ the quantity 
$(II)$ can be explicitly evaluated as
\be \label{ii} (II)
\ = \
\f 2 {\gamma \sqrt \om} \left( 1 - \f 4 {\gamma^2 \om} \right)^{\f 1 \mu - 1}
\ee
which is positive too. Therefore, by \eqref{disecondo} and \eqref{unoddue},
we conclude that
$d^{\prime \prime} (\om) > 0$, if $\om_0 < \om < \om^*$.

\n
To determine the sign of $(II)$ for $\om > \om^*$, it is convenient to
distinguish between the cases $\mu \leq 2$ and $\mu > 2$. Let us start with
$\mu \leq 2$.
Rewriting the first equation in \eqref{tsystem}
as
$$
t_1^2 (1 - t_1^2)^{\f 1 \mu} \ = \ t_2^2 (1 - t_2^2)^{\f 1 \mu}
$$
and differentiating with respect to $\om$, we obtain
$$
t_1^\prime \ = \ \f {\mu t_2 - (\mu + 1) t_2^3} {\mu t_1 - (\mu + 1) t_1^3}
\ \left( \f{1 - t_2^2}   {1 - t_1^2} \right)^{\f 1 \mu - 1} \ t_2^\prime,
$$
where we neglected in the notation the dependence on $\omega$.
Therefore 
$$ (II)
\ = \
- \om 
t_2^\prime (1 - t_2^2)^{\f 1 \mu - 1} \left( 1 +  
\f {\mu t_2 - (\mu + 1) t_2^3} {\mu t_1 - (\mu + 1) t_1^3}\right).
$$
We prove that such a quantity is positive for $\mu < 2$. As $t_1 < \bar t < t_2$,
this reduces to prove
\be \label{toprove}
{(\mu + 1) t_2^3 - \mu t_2} \ > \ {\mu t_1 - (\mu + 1) t_1^3}. 
\ee
To this aim, we define the function
\be \nnn 
\Gamma (t) \ : = \ [(\mu + 1) t^3 - \mu t]^2 \ = \ \f 1 4 t^{4 - 4 \mu} 
[f^\prime (t)]^2,  
\ee
where $f$ is defined as in \eqref{effe}, namely by $f(t) := t^{2 \mu} - t^{2 \mu + 2}$.
By the fundamental theorem of the integral calculus, for any $t \in [0,1]$
\be \label{fundint}
\Gamma (t) \ = \ \int_{\bar t}^t 
\left[ (1 - \mu) s^{3 - 4 \mu}
  [f^\prime (s)]^2 + \f 1 2 s^{4 - 4 \mu} f^\prime (s)  f^{\prime \prime}
  (s) \right] \, ds,
\ee
where $\bar t = \sqrt{\f \mu {\mu + 1}}$.
Now we proceed like in the proof of theorem \ref{theo:bifurcation},
formulas \eqref{links}-\eqref{pbar}. First, we define the function
$s_1$ as the inverse of $f$ in the interval $[0, \bar t]$, as well as
the function $s_2$ that is the inverse of $f$ in $[\bar t, 1]$. Then,
performing the change of variable $u = f(s)$, \eqref{fundint} gives
\be
\Gamma (t) \ = \  \int_{m}^{f(t)}
 \Sigma (s_i (u)) \, du
\ee
where $i = 1$ if $t \in [0, \bar t]$, $i = 2$ if $t \in (\bar t, 1]$,
$m := f (\bar t) = \f {\mu^{\mu}} {(\mu + 1)^{\mu + 1}}$,
and 
\be \nnn 
\Sigma (s) \ : = \
 (1 - \mu) s^{3 - 4 \mu}
  f^\prime (s ) + \f 1 2 s^{4 - 4 \mu}  f^{\prime \prime}
  (s) \ = \ \mu s^{2 - 2 \mu} - 3 (\mu + 1) s^{4 - 2 \mu}.
\ee
Consider the case $t_1 \in [\bar t/ \sqrt 3, \bar t)$. 
The study of the sign of $\Sigma$ and $\Sigma^\prime$ shows that
$\Sigma$ is negative
and strictly decreasing in $(\bar t / \sqrt 3, 1)$,
for any $u \in [f(t_1),m)$ one has
$$
0 \ > \ \Sigma (s_1 (u)) \ > \ \Sigma (\bar t) \ = \ -2 \f {\mu^{2 -
    \mu}} {(1 + \mu)^{1 - \mu}} \ > \ \Sigma (s_2 (u)),
$$
and therefore, denoting $a = f(t_1) = f(t_2)$,
\be \label{firstcase}
\Gamma (t_1) \ = \ - \int_a^m \Sigma (s_1 (u)) \, du \ < \ -  \int_a^m
\Sigma (s_2 (u)) \, du \ = \ \Gamma (t_2).
\ee
Second, consider the case $t_1 < \bar t / \sqrt 3$. 
Write $\Gamma
(t_1)$
as
$$
\Gamma (t_1) \ = \ \int_m^{f(\bar t / \sqrt 3))}  \Sigma (s_1 (u)) \,
du +  \int_{f(\bar t / \sqrt 3)}^{f(t_1)}  
   \Sigma (s_1 (u)) \, du.
$$
Notice that the first integral in the r.h.s. is positive, while the
second is negative, owing to the facts that $\Sigma$ is positive in
$(0, \bar t /\sqrt 3)$ and that $f(t_1) < f(\bar t / \sqrt 3)$. 
Then, denoting by $\bar t_2$ the only point in $(t_1, 1]$ such that 
$f (\bar t / \sqrt 3 ) = f (\bar t _2)$, by \eqref{firstcase},
\be \label{tibar2}
\Gamma (t_1) \ < \Gamma (\bar t /\sqrt 3) \ \leq \ \Gamma (\bar t_2).
\ee
Furthermore, since $t_1 < \bar t /\sqrt 3 $, it must be $\bar t_2 < t_2$, and
since $\Sigma$ is negative in the interval $(\bar t_2, t_2)$,
we obtain $
\Gamma (t_2) = - \int_a^{f(\bar t_2)} \Sigma (s_2 (u)) du + \Gamma (\bar t_2)
>  \Gamma (\bar t_2)
$
that, together with \eqref{tibar2} and \eqref{firstcase} yields
$
\Gamma (t_1) < \Gamma (t_2)
$
for any $t_1 \in [0, \bar t)$, which is equivalent to \eqref{toprove}.
So we proved
$(II) > 0$ if $\mu \leq 2$, that, together with the
positivity of $(I)$,
proves  that $d^{\prime \prime} (\om) > 0$ for any $\om > \om_0$, $\om \neq \om^*$.

In order to complete the proof of
point {\em 1.}, we need to evaluate $d^{\prime \prime} (\om^* \pm 0)$, and to compare
them. By \eqref{disecondo} and \eqref{unoddue}, this amounts to compare the value of
the two limits $\lim_{\om \to \om^* \pm 0} [ (I) + (II)]$. 
From \eqref{unoddue} it is clear that $(I)$ is continuous at $\om^*$, where it takes
the value $\left( \f 2 \mu -  1 
\right) \int_{\sqrt{\f \mu {\mu + 1}}}^1 (1-t^2)^{\f 1 \mu -1} \,  dt $,
so we reduce to study the limits of the term $(II)$ only.
The left limit
is immediately given by \eqref{unoddue} and \eqref{ii},
and reads
\be \label{leftl}
 \lim_{\om \to \om^* - 0} (II) \ = \
  \f {\sqrt \mu} {(\mu + 1)^{\f 1 \mu - \f 1 2}} \ > \ 0.
\ee
Computing the right limit is more complicated. 
Differentiating both equations in 
\eqref{tsystem} one can express the couple $(t_1^\prime, t_2^\prime)$ as a function
of the couple $(t_1, t_2)$, namely
\be \begin{split}
\label{tiprimi}
t_1^\prime \  = \ - \f {\gamma f^\prime (t_2) t_1^2 t_2^2} 
{2 \sqrt \om \left( t_1^2  f^\prime (t_1) +   t_2^2  f^\prime (t_2)
  \right)}, \ & \ \quad
t_2^\prime \  = \ - \f {\gamma f^\prime (t_1) t_1^2 t_2^2} 
{2 \sqrt \om \left( t_1^2  f^\prime (t_1) +   t_2^2  f^\prime (t_2) \right)} ,
\end{split}
\ee
where we neglected in notation the dependence of $t_1$ and $t_2$
on the variable $\om$ and denoted $f(t) = t^{2 \mu} -  t^{2 \mu+2}$.
From \eqref{tiprimi}, \eqref{zetai}, \eqref{unoddue}, and since, by \eqref{tsystem},
$f(t_1) = f(t_2)$,
one gets
\be \label{du}
(II) \ = \ \f{\gamma \sqrt \om} 2 t_1^{2 \mu} t_2^{2 \mu} f^{\f 1 \mu - 1}(t_1)
\f { t_1^{2-2 \mu}   f^\prime (t_1) +   t_2^{2-2 \mu}  f^\prime (t_2)} 
{ t_1^2  f^\prime (t_1) +   t_2^2  f^\prime (t_2)}.
\ee
Here, the only non trivial factor is the last one, in which both numerator
and denominator vanish as $\om$ goes to $\om^*$ from the right. In order to compute
such a limit we pass from the variable $\om$ to the variable $t_1$. In other words, we
consider $t_2$ as a function of $t_1$. 

\noindent
We define the functions
\be \begin{split}
N (t_1) \ & : = \  t_1^{2-2 \mu}   f^\prime (t_1) +   t_2^{2-2 \mu}  (t_1)
 f^\prime (t_2 (t_1)) \\
D (t_1)  \ & : = \  t_1^{2}   f^\prime (t_1) +   t_2^{2}  (t_1)
 f^\prime (t_2 (t_1))
\end{split} \ee
and provide a Taylor expansion near $t_1 = \bar t = \f 2 {\gamma \sqrt \om}$ for
both of them. One immediately gets
\be \begin{split} \nnn 
N^\prime (t_1) \ & : = \ (2 - 2 \mu) \left[ t_1^{1-2 \mu}   f^\prime (t_1) +
 t_2^{1-2 \mu}   f^\prime (t_2) \dot t_2  
\right] +   t_1^{2-2 \mu}   f^{\prime \prime} (t_1) +   
t_2^{2-2 \mu}   f^{\prime \prime} (t_2)  \dot t_2  \\
D^\prime (t_1) \ & : = \ 2 t_1  f^\prime (t_1) + t_1^2  f^{\prime \prime} (t_1)
+ 2  t_2  f^\prime (t_2)  \dot t_2  +  t_2^2  f^{\prime \prime} (t_2) \dot t_2  
\end{split} \ee
where we used the notation $\dot t_2  : = \f {dt_2}{dt_1}(t_1)$. To evaluate
$N^\prime (\bar t)$ and $D^\prime (\bar t)$ we must then compute 
$ \f {dt_2}{dt_1}(\bar t)$. From $f(t_1) = f(t_2)$ it follows 
\be \label{dotti}
\dot t_2 = \f {f^\prime (t_1)}{f^\prime (t_2)}.
\ee
By de l'H\^{o}pital's Theorem,
$$
\lim_{t_1 \to \bar t-0} \dot t_2 \ = \ \lim_{t_1 \to \bar t -0}  
\f {f^{\prime \prime} (t_1)}{f^{\prime \prime} (t_2) \dot t_2} ,
$$
from which one immediately has $\left(\lim_{t_1 \to \bar t} \dot t_2 \right)^2
= 1$. Now, since $f^\prime (t_1) > 0$ and $f^\prime (t_2) < 0$, it must be
$$
\lim_{t_1 \to \bar t-0} \dot t_2 \ = \ - 1.
$$
As a consequence, $N^\prime (\bar t) = D^\prime (\bar t) = 0$. Further
differentiating $N$ and $D$, and recalling that $f^\prime (\bar t) = 0$, one obtains
\be \begin{split} \nnn 
N^{\prime \prime} (\bar t) \ & : = \ 8 (1 -  \mu) \bar t^{1-2 \mu}   
f^{\prime \prime} (\bar t) + 2
 \bar t^{2-2 \mu}   f^{\prime \prime \prime} (\bar t) +
t^{2-2 \mu}   f^{\prime \prime} (\bar t)  \f {d^2t_2}{dt_1^2} (\bar t) \\
D^{\prime \prime} (\bar t) \ & : = \ 8 \bar t
f^{\prime \prime} (\bar t) + 2
 \bar t^{2}   f^{\prime \prime \prime} (\bar t) +
t^{2}   f^{\prime \prime} (\bar t)  \f {d^2t_2}{dt_1^2} (\bar t).
\end{split} \ee
By \eqref{dotti}
$$
 \f {d^2t_2}{dt_1^2} (t_1) \ = \ \f { f^{\prime \prime} (t_1) (f^{\prime})^2 (t_2
(t_1)) -  (f^{\prime})^2 (t_1)  f^{\prime \prime} (t_2(t_1))}
{ (f^{\prime})^3 (t_2
(t_1))},
$$
so, again using de l'H\^{o}pital's theorem, 
\be
\nnn 
\lim_{t_1 \to \bar t - 0} \ = \ - \f{ 2  f^{\prime \prime \prime} (\bar t)}
{ 3  f^{\prime \prime} (\bar t)}.
\ee
It then follows
\be \begin{split} \nnn 
N^{\prime \prime} (\bar t) \ & : = \ 8 (1 -  \mu) \bar t^{1-2 \mu}   
f^{\prime \prime} (\bar t) + \f 4 3 
 \bar t^{2-2 \mu}   f^{\prime \prime \prime} (\bar t) \\
D^{\prime \prime} (\bar t) \ & : = \ 8 \bar t
f^{\prime \prime} (\bar t) + \f 4 3 
 \bar t^{2}   f^{\prime \prime \prime} (\bar t).
\end{split} \ee
So, going back to \eqref{du}, we get
$$
\lim_{\om \to \om^* + 0} (II) \ = \ 
\f{\gamma \sqrt \om^*} 2 \bar t^{4 \mu}  f^{\f 1 \mu - 1}(\bar t)
\f { 6 (1 -  \mu) \bar t^{1-2 \mu}   
f^{\prime \prime} (\bar t) +  \bar t^{2-2 \mu}   f^{\prime \prime \prime} (\bar t)}
{ 6 \bar t
f^{\prime \prime} (\bar t) + 
 \bar t^{2}   f^{\prime \prime \prime} (\bar t)},
$$
which, recalling that $\bar t = \sqrt{\f \mu {\mu + 1}}$
and the definition of the function $f$, gives
\be \label{sdraio}
\lim_{\om \to \om^* + 0} (II) \ = \ \f {\sqrt \mu} {(\mu + 1)^{\f 1 \mu - \f 1 2}}
\f {5 - 2 \mu} {4 \mu + 5} \ > \ 0. 
\ee
Comparing \eqref{leftl} and \eqref{sdraio}, and observing that the
existence of such limits implies the existence of the left and right
derivative and gives their values, one completes the proof of
point {\em 1.}

\bigskip

\noindent
{\em 2. $\mu > 2$, $\om_0 < \om \leq \om^*$.}

\medskip

\noindent
In this case, term $(I)$ is negative, so one must compare its size to
the size of 
$(II)$.

From \eqref{zetai} and  \eqref{disecondo} we know 
that
\be \label{prerre}
d^{\prime \prime} (\om) \ = \ \left( \f {\mu + 1} \lambda \right)^{\f 1 \mu}
\f{\om^{\f 1 \mu - \f 3 2}} \mu r (\om), 
\ee
where
\be \label{erreom}
r (\om) \ = \ \f {2 - \mu} {2 \mu} \int_{\f 2 {\gamma \sqrt \om}}^1 (1
- t^2)^{\f 1 \mu - 1} dt + \f 1 {\gamma \sqrt \om} \left( 1 - \f 4
  {\gamma^2 \om} \right)^{\f 1 \mu - 1} .
\ee 
Then,
\be \nnn 
r^\prime (\om) \ = \ - \f 1 {\gamma \om^{\f 3 2}} \left( 1 - \f 1 \mu \right)  
 \left( 1 - \f 4
  {\gamma^2 \om} \right)^{\f 1 \mu - 2} \ < \ 0.
\ee
We estimate the first term in the
r.h.s. of \eqref{erreom}, for $\om = \om^*$, as
\begin{eqnarray} & & 0 \ > \ \nonumber
\f{2 - \mu}{2 \mu} \int_{\sqrt{ \f \mu {\mu + 1}}}^1 (1 - t^2)^{\f 1
  \mu - 1} dt \, > \, \f{2 - \mu}{2 \mu} \f {\int_{\sqrt{ \f \mu {\mu + 1}}}^1 (1 - t)^{\f 1
  \mu - 1} dt}  {\left( 1 + \sqrt{\f \mu
      {\mu + 1}} \right)^{1 - \f 1 \mu}} 
\ = \
 \f{2 - \mu}  {2 \left(1 + \sqrt{\f \mu
      {\mu + 1}} \right) (\mu + 1)^{\f 1 \mu}}.
\end{eqnarray}
So, 
$$
r (\om^*) \ > \ \f 1  {2 (\mu + 1)^{\f 1 \mu}} \left[ \f {2 - \mu}  {1 + \sqrt{\f \mu
      {\mu + 1}}} + \sqrt{\mu (\mu + 1)} \right] \ > \ \f {2 - \mu +
\sqrt{\mu (\mu + 1)}}  {2 (\mu + 1)^{\f 1 \mu}} \ > \ 0.
$$
Thus, since $r$ is monotonically decreasing, we have $r > 0$ for any
$r \in (\om_0, \om^*]$. 

Finally, since $\lim_{\om \to \om^* - 0} d^{\prime \prime} (\om)$
exists and can be recovered by putting $\om = \om^*$ in \eqref{prerre},
one obtains that $d^{\prime \prime} (\om^* - 0)$
equals such limit and so point 2. is proven.

\bigskip

\noindent
{\em 3. $\mu > 2$.}

\medskip

\noindent
In order to prove points {\em 3 (a,b,c)} we need to evaluate
$d^{\prime \prime} (\om^* + 0)$. 
We compute $\lim_{\om \to \om^*+0} d^{\prime \prime} (\om)$. Then,
$$
\lim_{\om \to \om^*+0} d^{\prime \prime} (\om) = \lim_{\om \to \om^*+0} (I) +
\lim_{\om \to \om^*+0} (II),
$$
with $(I)$ and $(II)$ defined in \eqref{unoddue},
and, by a direct computation,
$$ 
\lim_{\om \to \om^*+0} (I) \ = \
\left( \f 2 \mu -  1  \right) \int_{\sqrt \f \mu {\mu + 1}}^1
(1 - t^2)^{\f 1 \mu - 1} dt  
 \ < \ 0.
$$
Exactly as in the proof of point {\em 1.}, one finds that 
the right limit of the term $(II)$ is given by formula \eqref{sdraio}.

Now, according to \eqref{disecondo} and \eqref{unoddue}, the sign of
$d^{\prime \prime} (\om^* + 0)$, is given by the sign of  
the function
\be \nonumber
w (\mu) \ := \ \lim_{\om \to \om^* + 0} [ (I) + (II) ] \ = \ 
\left( \f 2 \mu - 1  \right) \int_{\sqrt{\f \mu {\mu + 1}}}^1 
(1 - t^2)^{\f 1 \mu - 1} \, dt + \f {\sqrt \mu}{(\mu + 1)^{\f 1 \mu - \f 1 2}}
\f {5 - 2 \mu} {4 \mu + 5}. 
\ee
Such a sign is obviously negative for $\mu \geq \f 5 2$, so we restrict
to $2 < \mu < \f 5 2$.
Then,
\be  \label{wprimo}
\begin{split}
w^\prime (\mu) \ = \ & - \left( \f 2 {\mu^2} + \f {\mu - 2} {2 \mu^{\f 3 2} (
  \mu + 1 )^{\f 3 2}} \right) \int_{\sqrt{\f \mu {\mu + 1}}}^1 
(1 - t^2)^{\f 1 \mu - 1} \, dt - \f {\mu - 2}{ \mu^3} 
 \int_{\sqrt{\f \mu {\mu + 1}}}^1 (1 - t^2)^{\f 1 \mu - 1} 
\log \left( \f 1 {1 - t^2} \right)\, dt \\
& + \f {2 \mu + 1} {2 \sqrt \mu (\mu + 1)^{1 + \f 1 \mu}} \f {5 - 2
  \mu} {4 \mu + 5} - \f {(\mu + 1)^{\f 1 2 - \f 1 \mu}}{\mu^{\f 3 2}}
{\log{(\mu + 1)}} \f {5 - 2
  \mu} {4 \mu + 5}  - \f 1 {\sqrt \mu (\mu + 1)^{\f 1 2 - \f 1 \mu}}
\f {5 - 2 
  \mu} {4 \mu + 5} \\ & - 30
\f  {\sqrt \mu (\mu + 1)^{\f 1 2 - \f 1 \mu}}{(4 \mu + 5)^2}.
\end{split} \ee
For $2 < \mu < \f 5 2$ the only non negative term in \eqref{wprimo} is 
the third one. By elementary computation we find that the sign of the
sum of such term with the last one is negative if and only if
$$
- 16 \mu^3 + 12 \mu^2 + 25 \ < \ 60 \mu^2 .
$$
 Since for $\mu
= 2$ the inequality is verified, it must be verified for any $\mu >
2$, owing to the fact,
easy to check, that the l.h.s is a monotonically decreasing
function, while the r.h.s is monotonically increasing. 
As a consequence, $w$ is monotonically decreasing for $\mu \in [2, \f
5 2]$, so there exists a unique value of $\mu$ that makes $d^{\prime
  \prime} (\om^*+0)$ vanish. Denoting it by $\mu^\star$, we complete the
proof of points {\em 3 (b,c)}. Point {\em 3 (a)} follows by continuity
from
the fact that, if $ \mu < \mu^\star$, then $d^{\prime \prime} (\om^*+
0) > 0$.

\bigskip

\noindent
{\em 4. $\mu > 2, \om \to \infty.$}

\medskip

\noindent
Point 4. follows  from the asymptotics of $t_1$ and $t_2$
as $\om$ goes to $\infty$, in the region $t_2 > t_1 > 0$:
\begin{eqnarray}
t_1 \ = \ \f 1 {\gamma \sqrt \om} + {\rm{o}} (\om^{-\f 1 2
}) \label{egy} & \ & 
t_2 \ = \ 1 - \f 1 {2 \gamma^{2 \mu} \om^\mu} + {\rm{o}}
(\om^{-\mu})  \\
t_1^\prime \ = \ - \f 1 {2 \gamma \om^{\f 3 2}}  + {\rm{o}} (\om^{-\f
  3 2 }) \label{tres} & \ &
t_2^\prime \ = \ \f \mu {2 \gamma^{2\mu} \om^{\mu + 1}}  + {\rm{o}}
(\om^{-\mu -1 }). 
\end{eqnarray}
Let us prove such asymptotics. 
The condition $t_2 > t_1$ selects, among the
solutions to the system \eqref{tsystem}, those belonging to the set
$T_2$ defined in \eqref{tidue}. It is immediately seen that, in such
region, $t_1 \to 0$ and $t_2 \to 1$ as $\om \to \infty$. As a
consequence, from the second equation in \eqref{tsystem} one gets
$$
\lim_{\om \to \infty} \gamma \sqrt \om t_1 \ = \ \lim_{\om \to \infty}
\f {\gamma \sqrt \om t_2} {\gamma \sqrt \om t_2 -1} \ = \ 1,
$$
so the first formula in \eqref{egy} is proven.

\n
From the first equation in \eqref{tsystem} we have
$$
\lim_{\om \to \infty} t_1^{-2 \mu} (1 - t_2^2) \ = \ \lim_{\om \to
  \infty} t_2^{-2 \mu} (1 - t_1^2) \ = \ 1,
$$  
and by the first of \eqref{egy}
$$
1 - t_2^2 \ = \ \f 1 {\gamma^{2 \mu} \om^\mu}  + {\rm{o}}
(\om^{-\mu})
$$
and the second identity in \eqref{egy} immediately follows.

\noindent
In order to prove \eqref{tres} we differentiate
both equations in \eqref{tsystem}, obtaining
\begin{equation} \label{difftsystem}
\left\{ \begin{array}{ccc}
t_1^\prime t_1^{2 \mu - 1} ( \mu - (\mu + 1) t_1^2) & = & t_2^\prime
t_2^{2 \mu - 1} ( \mu - (\mu + 1) t_2^2) \\ & & \\
{t_1^\prime}{t_1^{-2}} + {t_2^\prime}{t_2^{-2}} & = & - \f \gamma {2
  \sqrt \om}
\end{array}
\right.
\end{equation}
Expressing $t_2^\prime$ from the second equation 
and plugging it into the first one we get
$$
\om^{\f 32}t_1^\prime \ = \ - \f {\gamma \om t_1^2 t_2^{2 \mu + 1}
  (\mu - (\mu + 1) t_2^2)} {2 t_2^{2 \mu + 1}
  (\mu - (\mu + 1) t_2^2) + 2 t_1^{2 \mu + 1}
  (\mu - (\mu + 1) t_1^2)},
$$
that converges to $- \f 1 {2 \gamma}$ as $\om$ goes to infinity, and
so the first formula in
 \eqref{tres} is proved. To prove the second one, it is sufficient to
use the first equation in \eqref{difftsystem} and by \eqref{egy},
\eqref{ketto}, the first in \eqref{tres} one finally obtains
$$
\lim_{\om \to \infty} \om^{\mu + 1} t_2^\prime \ = \ \f {\mu}{2
  \gamma^{2 \mu}}
$$ and the proof is complete.
\end{proof}

\subsection{\label{sec:stab} 
Stability and instability of the ground states. Pitchfork
  bifurcation}

We recall the definition of orbital  
neighbourhood and orbital stability.

\begin{defi}
The set
$$ U_\eta (\phi) : = \{ \psi \in Q, \, {\rm s.t.} \, \inf_{\theta \in [0, 2 \pi)}
\| \psi - e^{i \theta} \phi \|_Q \leq \eta \}$$
is called the {\em orbital spherical neighbourhood with radius $\eta$ of the function $\phi$}.
\end{defi}

\begin{defi}\label{defi:orbstab}
We call {\em orbitally stable} (in the future) any stationary state
$\phi$ such that for any $\ve > 0$ there exists $\delta > 0$ such that
$$
\inf_{\theta \in [0, 2 \pi)}
\| \psi -  e^{i \theta} \phi \|_Q \leq \delta \, 
\Longrightarrow \, \sup_{t \geq 0} \inf_{\theta \in [0, 2 \pi)}
\| \psi(t) - e^{i \theta} \phi \|_Q \leq \ve,
$$
where $\psi (t)$ is the solution to the problem \eqref{intform} with $\psi$
as initial datum.
\end{defi}

\begin{defi}
We call {\em orbitally unstable} any stationary state
that is not orbitally stable.
\end{defi}

\begin{prop} [Stability and instability of ground states] \label{nature}
Consider the ground states of the dynamics described by \eqref{nlsdelta'}, defined as
the solutions to problem \ref{prob2} and explicitly computed in
Theorem \ref{theo:bifurcation}. Then,
\begin{enumerate}
\item If $0 < \mu \leq 2$, then 
for any $\om \in (\om_0, + \infty)$, $\om \neq \om^*$, all ground states are stable.
\item If $\mu > 2$, then
\begin{enumerate}
\item there exists $\om_1 > \om^*$ such that all the ground states 
with frequency $\omega < \om_1$ are orbitally stable.
\item There exists $\om_2 \geq \om_1$ such that if $\om > \om_2$ then
all ground states with frequency $\om > \om_2$ are orbitally unstable.
\end{enumerate} 
\end{enumerate}
\end{prop}
\begin{proof}
Points 1 and 2 (a) follow
 from Theorem 2 in \cite{[GSS1]}. Indeed, notice that 
Assumption 1 in such theorem is proven by Propositions \ref{loch2}
and \ref{conslaws},
while Assumption 2 follows from Proposition \ref{propsystem} and
Theorem \ref{theo:bifurcation}. 
Furthermore, owing to propositions \ref{diecitre}, \ref{l1simm}
a), b), and \ref{l1asimm}, Assumption 3 is verified for all ground
states. Therefore, in order to establish orbital stability, it is
sufficient to remark that,  for the considered cases, Proposition
\ref{propdisecondo} establishes  $d^{\prime
  \prime} (\om) > 0$. 

Case 2 b) follows from Theorem 4.7 in \cite{[GSS1]}, as we know
from Proposition \ref{propdisecondo} that $d^{\prime \prime} (\om) <
0$.
Thus, the theorem is proven. 
\end{proof}
\begin{figure}
\begin{center}
{}{}\scalebox{0.43}{\includegraphics{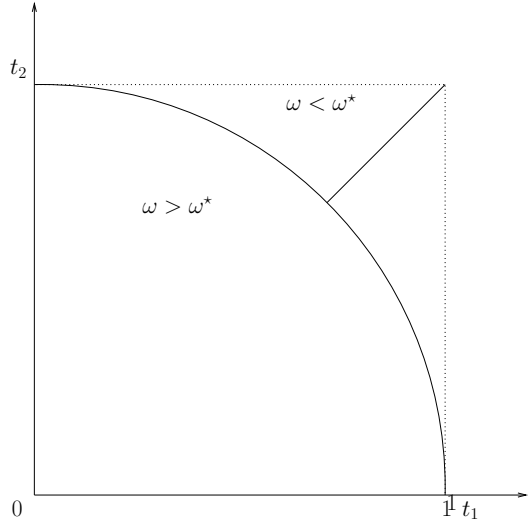}}
\caption{{\bf Hypocritical and critical case (i.e. $\mu \leq 2$)}. All
ground states are stable.}
\end{center}
\end{figure}

\begin{figure}
\begin{center}
{}{}\scalebox{0.43}{\includegraphics{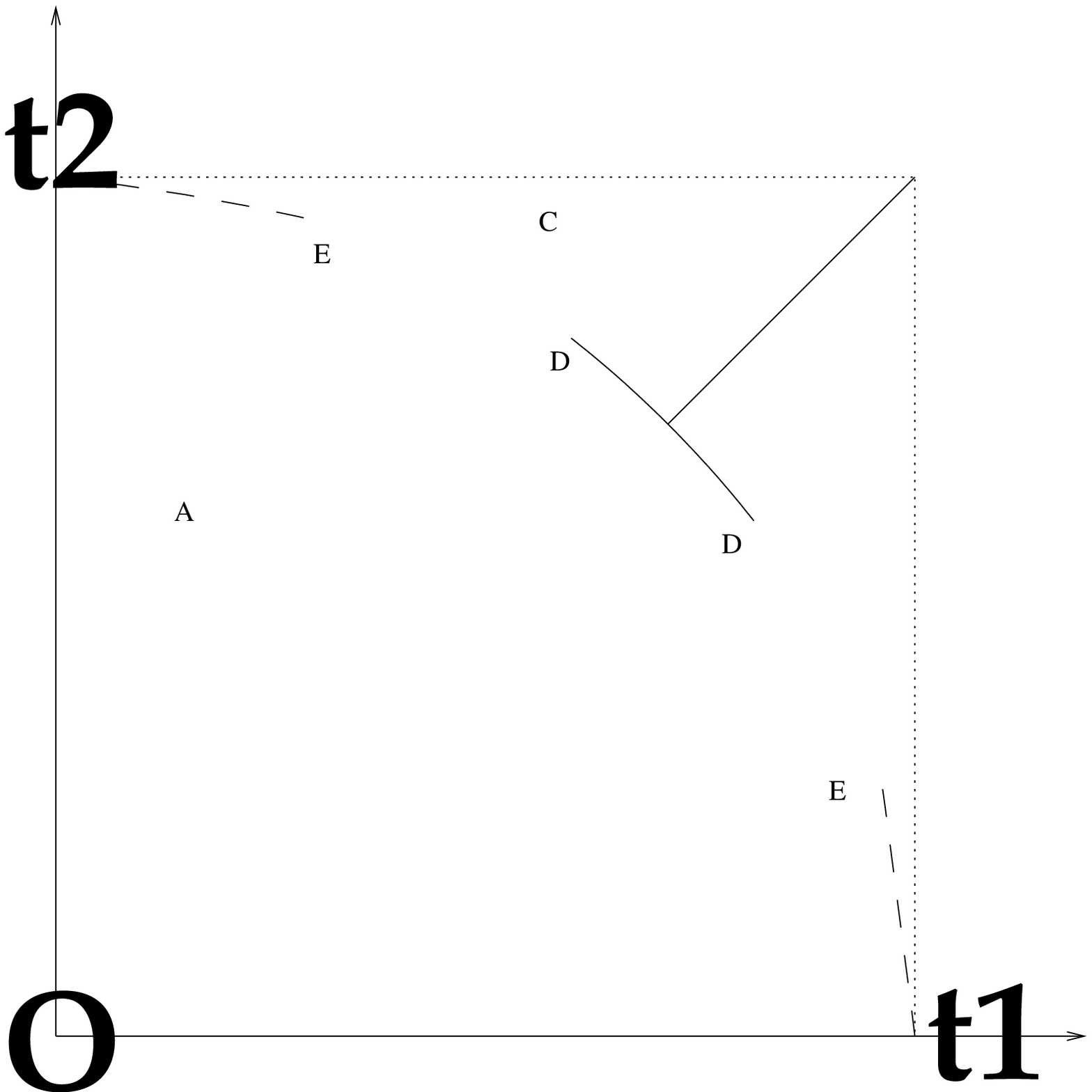}}
\caption{{\bf First hypercritical case (i.e. $2 < \mu < \mu^\star
< 2.5$)}. Symmetric ground states are stable. Immediately after
bifurcation, the two newborn asymmetric states are still stable. At
large frequencies, they become unstable.}
\end{center}
\end{figure}

\begin{figure}
\begin{center}
{}{}\scalebox{0.43}{\includegraphics{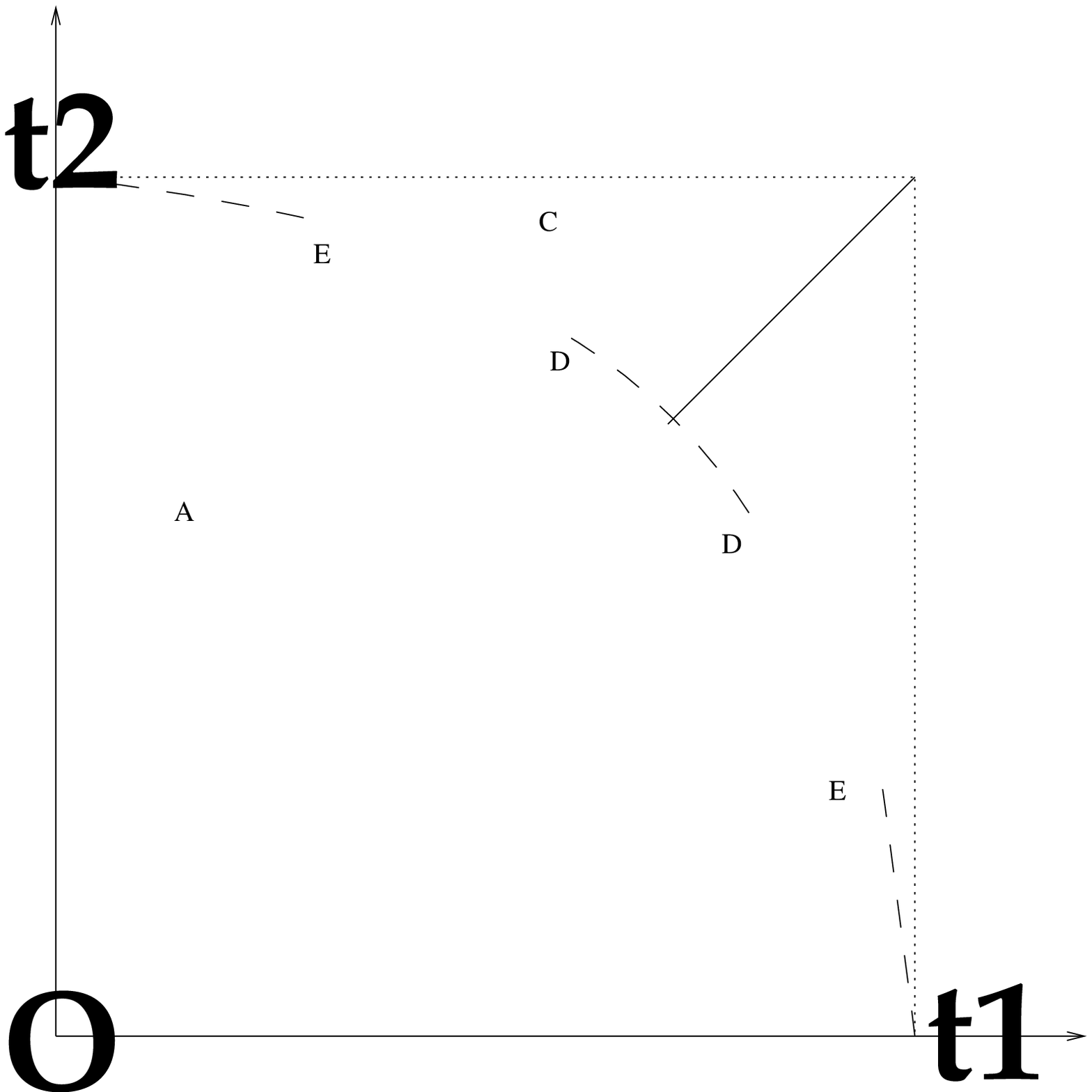}}
\caption{{\bf Second hypercritical case (i.e. $\mu > \mu^\star$)}. 
Symmetric ground states are stable. Change of stability
occurs immediately after bifurcation.}
\end{center}
\end{figure}

\begin{rem}{\em
By Theorem 3 in \cite{[GSS1]} we have that, under the hypotheses of
Theorem \ref{nature}, the ground states solve Problem \ref{natural},
i.e., they minimize {\em at least locally}
the energy \ref{energy} among the functions with
the same $L^2$-norm.}
\end{rem}

\begin{theorem}[Pitchfork bifurcation] \label{pitchfork}
Given $\mu > 0$, if $\om > \om^*$ then
the stationary solutions $\psi_{\om}^{y,-y,\theta}$
defined in Theorem \ref{theo:bifurcation} are orbitally unstable.
\end{theorem}

\begin{proof}
Closely mimicking the computation that led to formula
\eqref{disecondo}, one gets
\be \nnn 
\begin{split}
\f{d^2 S_\om (\psi_{\om}^{y,-y,\theta})}{d \om^2} (\om)
\ = \ & \left( \f {\mu + 1} \lambda \right)^{\f 1 \mu}
\f {\om^{\f 1 \mu - \f 3 2}} {\mu} \left\{ \left( \f 1 \mu - \f 1 2
\right) \int_{\f 2 {\gamma \sqrt \om}}^1 (1-t^2)^{\f 1 \mu -1} \,  dt 
+ \f 2 {\gamma \sqrt \om} \left( 1 - \f 4 {\gamma^2 \om} \right)^{\f 1 \mu - 1}
\right\}. \end{split}
\ee
As $\om > \om^*$, we know by proposition \ref{l1simm} we know that the
number of negative  
eigenvalues of the operator $L^{\gamma, \om}$ equals two. 
If $\f{d^2 S_\om (\psi_{\om}^{y,-y,\theta})}{d \om^2} (\om)$ is
positive (for instance
if $\mu \leq 2$), 
then the result follows
from Theorem 6.2 in \cite{[GSS2]}.

On the other hand, if $\f{d^2 S_\om (\psi_{\om}^{y,-y,\theta})}{d
  \om^2} (\om)$ is negative, then we restrict the problem to the space
$Q_a$ of
antisymmetric functions lying in $Q$. First, from the
explicit knowledge of the propagator of the Schr\"odinger equation
with a $\delta^\prime$ interaction, represented by the
integral kernel (see \cite{[ABD]})
$$
e^{-iH_\gamma t} (x,y) \ = \ \f {e^{i \f {(x-y)^2}{4t}}}{\sqrt{4 \pi i t}}
+ \epsilon(xy)  \f {e^{i \f {(|x|+|y|)^2}{4t}}}{\sqrt{4 \pi i t}}
+ \f{ \epsilon(xy)}{2 \gamma} \int_0^{+\infty} e^{\f{-2u}\gamma}
\f {e^{i \f {(|x|+|y|- u)^2}{4t}}}{\sqrt{4 \pi i t}} +  
\f{2 \epsilon(xy)}{ \gamma}{e^{i \f{4t}{\gamma^2}}} e^{-\f 2 \gamma
|x|+|y|}
$$
it appears that,  denoting  $\widetilde g (x) = g (-x)$,
one has
$$
\widetilde{e^{-iH_\gamma t} \psi_0} = e^{-iH_\gamma t}\widetilde
 \psi_0
$$
Let us consider the problem
\eqref{nlsdelta'}, and initial data
$\psi_0^a$ such that $
\widetilde \psi_0^a = -  \psi_0^a$.
Then, applying Duhamel's formula to 
\eqref{nlsdelta'}, one finds
$$
\psi_t \ = \ e^{-iH_\gamma t}\psi_0 + i \lambda \int_0^t  e^{-iH_\gamma (t-s)}
| \psi_s |^{2 \mu} \psi_s,
$$
so that
\be \nonumber \begin{split}
\widetilde  \psi_t \ & = \ e^{-iH_\gamma t}
\widetilde \psi_0 + i \lambda \int_0^t  e^{-iH_\gamma (t-s)}
| \widetilde \psi_s |^{2 \mu} \widetilde \psi_s 
  \ = \ - e^{-iH_\gamma t}
 \psi_0 + i \lambda \int_0^t  e^{-iH_\gamma (t-s)}
| \widetilde \psi_s |^{2 \mu} \widetilde \psi_s
\end{split} \ee
It follows that $\widetilde \psi_t$ solves \eqref{nlsdelta'} with 
$- \psi_0$ as initial data. Since \eqref{nlsdelta'} is invariant under 
multiplication by a phase factor, and since the solution is unique, it must be
$\widetilde  \psi_t = - \psi_t$, and so we have that \eqref{nlsdelta'}
preserves the antisymmetry, and therefore the evolution problem \eqref{nlsdelta'}
is 
well-defined in $Q_a$. 

Now, as already remarked in the proof of Proposition \ref{l1simm}, the
functions $\psi_{\om}^{y,-y,\theta}$ are the minimizers of $S_\om$
among the antisymmetric functions belonging to the Nehari
manifold. The time-evolution operator, linearized around them, admits
a simple negative eigenvalue, so, by Theorem 4.1 in \cite{[GSS1]}, the
stationary state $\psi_{\om}^{y,-y,\theta}$ is orbitally unstable in $Q_a$, so, 
{\em a fortiori}, it is orbitally unstable in $Q$. 
To complete the proof, we recall that the instability of $\psi_\om^{y,-y,\theta}$
in the case $\f{d^2 S_\om (\psi_{\om}^{y,-y,\theta})}{d \om^2} (\om) = 0$ can be
established by the argument in \cite{cp}.
\end{proof}

We are then in the presence of a pitchfork bifurcation, that can be depicted as
in Figure 5.

\begin{figure}
\begin{center}
{}{}\scalebox{0.43}{\includegraphics{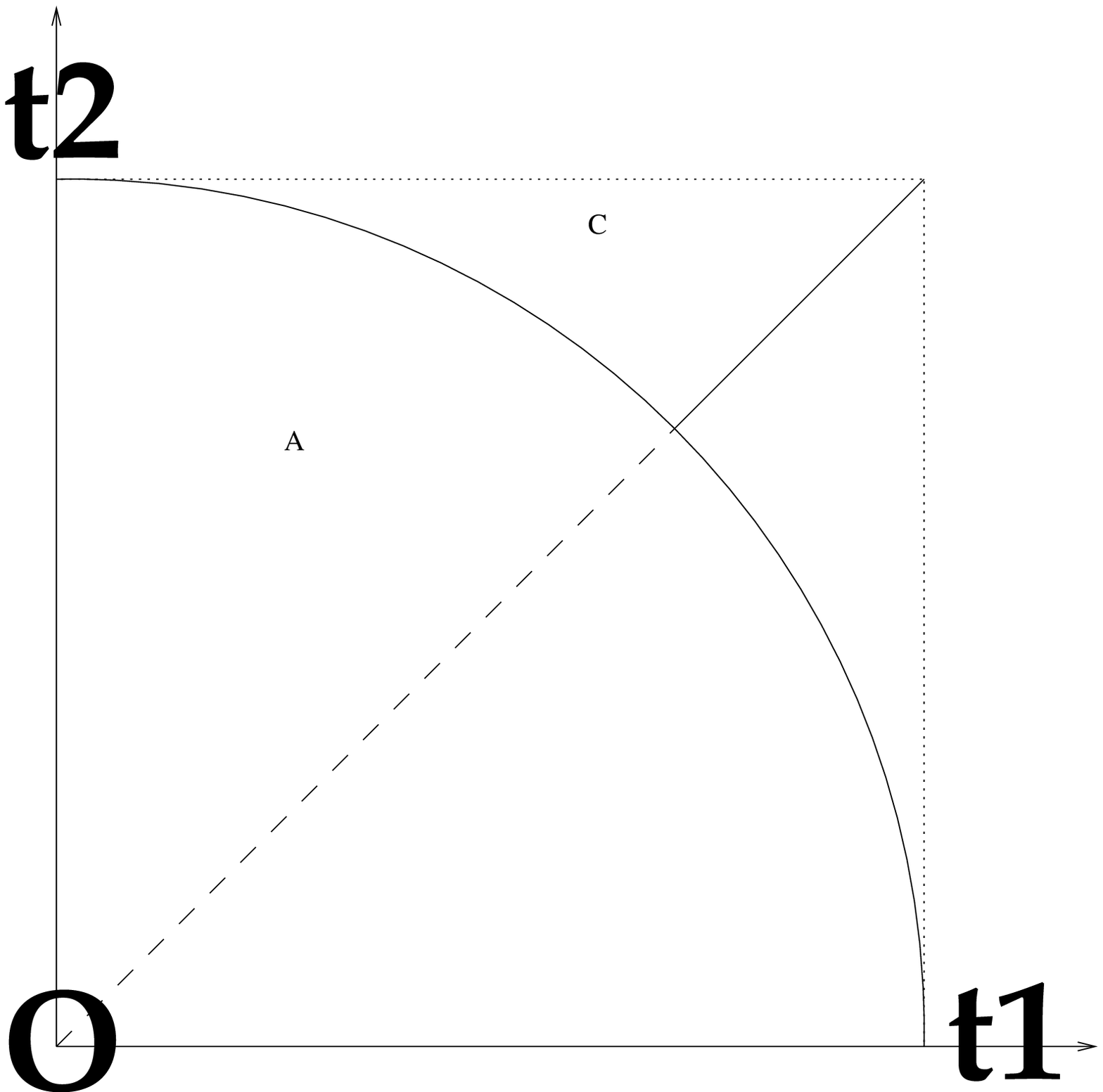}}
\caption{{\bf Pitchfork bifurcation}. The branch of the (anti-)symmetric
ground states can be continued beyond bifurcation, for any value of $\om$,
but the corresponding stationary states are orbitally unstable. By the same
modification, i.e., adding the unstable straight branch from the origin
up to $(\om^*, \om^*)$ to Figures 3 and 4, one obtains the pictures
 that correspond to
the cases $2 < \mu < \mu^\star$ and $\mu > \mu^\star$.}
\end{center}
\end{figure}

We leave open the issue of optimizing $\om_1$ and $\om_2$, and possibly getting
$\om_1 = \om_2$. In other words, we do not know whether, for a frequency beyond
the bifurcation, each ground state undergoes only a change of variable or more.

Finally, we don't treat here the problem of determining the stability features
of the ground states at the bifurcation frequency. This will be the subject of a
forthcoming note.

{\bf Acknowledgements.} 
The authors thank Andrea Sacchetti for many discussions during the
initial step of the project, that helped us to focus the problem.
R. A. is grateful to Cinzia Casagrande for some
explanations on algebraic curves and to Sergio Conti for an illuminating
discussion on the properties of the symmetric stationary states beyond
bifurcation.

R. A. is partially supported by 
 the PRIN2009 grant {\em Critical Point Theory and Perturbative 
Methods for Nonlinear Differential Equations}.


\begin{thebibliography}{99}
\bibitem{[AN]} Adami R., Noja D.: {\em Existence of dynamics for a 1-d NLS
  equation in dimension one}, J. Phys. A, {\bf 42}, 49, 495302,
  19pp (2009). 
\bibitem{anv} Adami R., Noja D., Visciglia N.: {\em Variational structure of standing
waves for NLS with point interactions: a general result with applications}, 
in preparation.
\bibitem{[AKM]} Akhmediev N. N.: {\em Novel class of nonlinear surface waves:
  asymmetric modes in a symmetric layered structure}, Sov. Phys. JETP
  {\bf 56} 299-303 (1982). 
\bibitem{[ABD]} Albeverio S., Brze\'zniak Z., Dabrowski L.: 
 {\em Fundamental solutions of the Heat and Schr\"odinger 
Equations with point interaction},
 {J. Func. An.}, {\bf 130}, 220--254 (1995).
\bibitem{[AGHH]} Albeverio S., Gesztesy F., Hoegh-Krohn R., Holden H.: 
{\it Solvable Models in Quantum Mechanics}, Springer-Verlag, New York (1988).
\bibitem{[AEL]} Avron J. E., Exner P., Last Y.: {\em Periodic Schr\"odinger operators with large
gaps and Wannier�Stark ladders}, Phys. Rev. Lett. 72, 896�899 (1994).
\bibitem{breli} Brezis H., Lieb E. H.: {\em A relation between
    pointwise 
convergence of functions and convergence of functionals}, 
Proc. Amer. Math. Soc., {\bf 88},
486--490 (1983).
\bibitem{[CL]} Cazenave T., Lions P.-L.: {\em Orbital stability of standing
  waves for some nonlinear Schr\"odinger equations},
  Comm. Math. Phys., {\bf 85} 549--561 (1982). 
\bibitem{[CM]} Cao Xiang D., Malomed A. B.: {\em Soliton defect collisions 
in the nonlinear Schr\"odinger equation}, {Phys. Lett. A}, {\bf
  206}, 177--182  (1995).
\bibitem{[C]} Cazenave, T.: {\it Semilinear Schr\"odinger Equations},
  vol. {\bf 10} Courant Lecture Notes in Mathematics {\it AMS},
  Providence (2003). 
\bibitem{[CS]} Cheon T., Shigehara T.:  
{\em Realizing discontinuous wave functions with renormalized short-range
potentials},  Phys. Lett. A, {\bf 243}, 111–--116 (1998).
\bibitem{cp} Comech A., Pelinovsky D.: {\em Purely nonlinear instability of standing waves with
minimal energy}, {Comm. Pure App. Math.}, {\bf 56}, 1565--1607 (2003).
\bibitem{[EG]} Exner P., Grosse P.:
{\em Some properties of the one-dimensional generalized point interactions (a torso)},
mp-arc 99-390, math-ph/9910029  (1999). 
\bibitem{[ENZ]} Exner P., Neidhart H., Zagrebnov V. A.: {\em Potential
  approximations to $\delta^\prime$: an inverse Klauder phenomenon
  with norm-resolvent convergence}, Comm. Math. Phys., {\bf 224}, 593--612
  (2001). 

\bibitem{fibich} Fibich G., Wang X. P.: {\em Stability for solitary waves
  for nonlinear Schr\"odinger equations with inhomogenous
  nonlinearities}, Physica D {\bf 175}, 96--108 (2003).
\bibitem{fukujean} Fukuizumi R., Jeanjean L.: {\em Stability of standing
  waves for a nonlinear Schr\"odinger equation with a repulsive Dirac
  delta potential}, Dis. Cont. Dyn. Syst. (A), {\bf
    21},  129--144 (2008). 
\bibitem{reika} Fukuizumi R., Ohta M, Ozawa T.:  {\em Nonlinear
  Schr\"odinger equation with a point defect}, {
    Ann. Inst. H. Poincar\'e - AN}, {\bf 25}, 837--845 (2008). 
\bibitem{[FS]}Fukuizumi R., Sacchetti A.: {\em Bifurcation and stability
  for nonlinear Schr\"odinger equation with double well potential in
  the semiclassical limit}, arxiv:1104.1511v1, to appear on J. Stat. Phys. (2011). 
\bibitem{goldstein} Goldstein H., Pole C.P., Safko J.L.: {\em
Classical Mechanics}, Addison Wesley, 2002. 
\bibitem{[GHW]} Goodman R. H., Holmes P. J., Weinstein M.I.: {\em Strong
  NLS soliton-defect interactions}, { Physica D}, {\bf 192},
  215--248 (2004). 
\bibitem{[GSS1]} Grillakis M., Shatah J., Strauss W.: {\em Stability theory
  of solitary waves in the presence of symmetry - I}, J. Func. An., {\bf
    74}, 160--197 (1987).  
\bibitem{[GSS2]} Grillakis M., Shatah J., Strauss W.: {\em Stability theory
  of solitary waves in the presence of symmetry - II}, J. Func. An.,
  {\bf 94}, 308--348 (1990).  
\bibitem{[HS]} Hislop P.D., Sigal I.M.: {\em Introduction to spectral
  theory: With applications to Schr\"odinger operators}, Springer, New
  York (1996). 
\bibitem{[HT]} Haroske D.D., Triebel H.: {\em Distributions, Sobolev
  Spaces, Elliptic Equations}, European Mathematical Society (2008). 
\bibitem{[HMZ]} Holmer J., Marzuola J., Zworski M.: Fast soliton
  scattering by delta impurities, {\it Comm. Math. Phys}, {\bf 274},
  187--216 (2007).
\bibitem{[JW]}   Jackson R. K., Weinstein M.: {\em Geometric analysis of
  bifurcation and symmetry breaking in a Gross-Pitaevskii equation},
  {J. Stat. Phys.},  {\bf 116}, 881--905 (2004).
\bibitem{[KKP]} Kirr E., Kevrekidis P. G., Pelinovsky D. E.:
  {\em Symmetry-breaking bifurcation in the nonlinear Schr\"odinger
  equation with symmetric potentials}, arxiv:1012.3921,  to appear on
  Comm. Math. Phys. (2010). 
\bibitem{lacozza} Le Coz S., Fukuizumi R., Fibich G., Ksherim B., Sivan Y.:
{\em Instability of bound states of a nonlinear Schr�dinger equation
  with a Dirac potential},  Phys. D,  {\bf  237}  
n. 8, 1103--1128 (2008). 
\bibitem{[MJS]} Marangell R., Jones C. K. R. T., Susanto H.: {\em Localized
  standing waves in inhomogeneous Schrodinger equations},
  Nonlinearity, {\bf 23}, 9  2059--2080 (2010). 
\bibitem{tuoc} Pelinovsky D. E., Phan  T.: {\em Normal form for the
symmetry-breaking bifurcation in the nonlinear Schr\"odinger
equation}, preprint arXiv:1101.5402 (2011).

\bibitem{rs4} Reed M., Simon B.: {\em Methods of Modern Mathematical
    Physics IV: Analysis of Operators}, Academic Press Inc., San Diego
  (CA) (1978).

\bibitem{[W2]} Weinstein M.: {\em Modulational stability of ground states
  of nonlinear Schroedinger equations}, 
SIAM J. Math. Anal., {\bf 16}, 472--491  (1985).
\bibitem{[W3]} Weinstein M.: {\em Lyapunov stability of ground states of
  nonlinear dispersive evolution equations},  
Comm. Pure Appl. Math., {\bf 39}  51--68 (1986).
\bibitem{[WMK]} Witthaut D., Mossmann S., Korsch H. J.: 
{\em Bound and resonance states of the nonlinear 
Schr\"odinger equation in simple model systems},
 {J. Phys. A}, {\bf 38}, 1777-1702 (2005).


\end{thebibliography}
\end{document}